\def\supplementaryMaterial{0}	%
\def\tightOnSpace{0}			%
\newif\ifsubmission
\submissionfalse

\newif\ifeprint
\eprinttrue

\ifeprint
\submissionfalse
\def\tightOnSpace{0}
\def\supplementaryMaterial{0}
\fi

\RequirePackage{amsmath}  %

\ifeprint
\documentclass[a4paper,runningheads]{llncs}
	\usepackage[all]{nowidow}
\else
\documentclass{llncs} %
\fi

\usepackage[utf8]{inputenc}
\usepackage[T1]{fontenc}

\usepackage{float}

\usepackage{tikz}
\usetikzlibrary{positioning}
\usetikzlibrary{calc}
\usetikzlibrary{shapes.geometric}

\usepackage{nicefrac}
\usepackage{nicodemus}
\usepackage{stmaryrd}

\usepackage{amsfonts,amsmath,amscd,amssymb,array}
\usepackage{algorithm2e,algpseudocode}
\usepackage{graphicx}
\usepackage{xspace}
\usepackage{ifthen}
\usepackage{braket}
\usepackage{dsfont}
\usepackage{breakcites}

\usepackage{thmtools}
\usepackage{thm-restate}

\newenvironment{reusefigure}[2][htbp]
  {\addtocounter{figure}{-1}%
   \renewcommand{\thefigure}{\ref{#2}}%
   \renewcommand{\addcontentsline}[3]{}%
   \begin{figure}[#1]}
  {\end{figure}}

\usepackage{xcolor}
\definecolor{linkcolor}{rgb}{0.65,0,0}
\definecolor{citecolor}{rgb}{0,0.65,0}
\definecolor{urlcolor}{rgb}{0,0,0.65}
\usepackage[colorlinks=true, linkcolor=linkcolor, urlcolor=urlcolor, 
citecolor=citecolor]{hyperref}
\usepackage[capitalise]{cleveref}
\ifnum\supplementaryMaterial=1
	\ifnum\tightOnSpace=1
		\crefname{appendix}{Suppl. Mat.}{Supplementary Material} 
	\else
		\crefname{appendix}{Supplementary Material}{Supplementary Material} 
	\fi
\fi

\usepackage[framemethod=default]{mdframed}

\newcommand{\eps}{\varepsilon}

\usepackage{mathtools}

\usepackage{ctable} %
\newcommand{\A}{\ensuremath{\Ad{A}}\xspace}
\newcommand{\B}{\ensuremath{\Ad{B}}\xspace}
\newcommand{\C}{\ensuremath{\Ad{C}}\xspace}
\newcommand{\M}{\ensuremath{\Ad{M}}\xspace}
\newcommand{\succfun}{\ensuremath{\mathrm{Succ}}\xspace}
\newcommand{\succf}[3]{\ensuremath{{\mathrm{Succ}}^{{#1}}_{#2}\left( {#3} 
\right) }}
\newcommand{\nmmetcr}{\ensuremath{\mathsf{nM\text{-}eTCR}}\xspace}
\newcommand{\mmetcr}{\ensuremath{\mathsf{M\text{-}eTCR}}\xspace}
\newcommand{\succnmmetcr}[2]{\ensuremath{\succf{\text{\nmmetcr}}{#1}{#2}}}
\newcommand{\succmmetcr}[2]{\ensuremath{\succf{\text{\mmetcr}}{#1}{#2}}}
\newcommand{\bbox}{\ensuremath{\mathsf{Box}}\xspace}
\newcommand{\pr}{\ensuremath{\mathsf{Pr}}} 

\newcommand{\id}{\ensuremath{\textrm{id}}\xspace}

\newcommand{\exec}{\ensuremath{\leftarrow}}

\newcommand{\qsign}{\ensuremath{q_s}\xspace}
\newcommand{\qhash}{\ensuremath{q_{\RO{H}}}\xspace}

\newcommand{\todo}[1]{
{\color{red} \bf TODO: #1}
}

\def\subheading#1{\medskip\noindent{\boldmath\textbf{#1}}~\ignorespaces}

\definecolor{dgreen}{rgb}{.1,.5,.1}

\newcommand{\proj}[1]{\ket{#1}\!\!\bra{#1}}
\newcommand{\Tr}{\mathrm{Tr}}

\newcommand{\kb}[1]{\ket{#1}\bra{#1}}

\DeclareMathOperator*{\E}{\mathbb E}
\newcommand{\poly}{\mathsf{poly}}
\newcommand{\norm}[1]{\left|\left| #1 \right|\right|}

\newmdenv[backgroundcolor=orange]{frameForAuthnotes}

\newcommand\bits{\{0,1\}}								%
\newcommand{\uni}{\leftarrow_\$}						%
\DeclareMathOperator\supp{supp}							%
\DeclareMathOperator*{\Exp}{\E}
\newcommand{\Adv}{\mathrm{Adv}}						%
\newcommand{\Ad}[1]{\ensuremath{\mathsf{#1}}\xspace}	%
\newcommand{\Time}{\textnormal{Time}}					%
\newcommand{\List}[1]{\mathfrak{L}_{#1}}				%
\newcommand{\param}{\textnormal{par}}					%
\newcommand{\state}{\textnormal{st}}
\newcommand{\inputVar}{\textnormal{inp}}

\newcommand{\mathsc}[1]{\text{\textsc{#1}}}
\newcommand{\heading}[1]{{\vspace{1ex}\noindent\sc{#1}}}

\mathchardef\ordinarycolon\mathcode`\:				%
\mathcode`\:=\string"8000							%
\begingroup \catcode`\:=\active						%
\gdef:{\mathrel{\mathop\ordinarycolon}}				%
\endgroup

\newcommand{\FIND}{\mathsc{FIND}}

\newcommand{\RO}[1]{\mathsf{{#1}}}					%
\newcommand{\qRO}[1]{{\ket{\RO{#1}}}}				%

\newcommand{\pcfor}{\textbf{for }}
\newcommand{\pcif}{\textbf{if }}
\newcommand{\pcelse}{\textbf{else }}
\newcommand{\pcand}{\textbf{and }}
\newcommand{\pcreturn}{\textbf{return }}
\newcommand{\gcom}[1]{\hfill $\sslash$#1}				%
\newcommand\boxedFull[1]{\tikz [baseline=(boxed word.base)] \node (boxed word) [draw, rectangle, line cap=round] {#1};}

\newcommand\dashboxed[1]{\tikz [baseline=(boxed word.base)] \node (boxed word) [draw, rectangle, dashed, line cap=round] {#1};}
\newcommand{\before}[1]{\the\numexpr\value{#1}-1\relax}	%
\newcommand{\gameDist}[2]{|\Pr[G_{\before{#1}}^{\Ad{#2}}=1] - \Pr[G_{\the\numexpr\value{#1}\relax}^{\Ad{#2}}=1]|}		%
\newcommand{\Measure}{\mathsc{Measure}}					%

\newcommand{\adaptiveReprogramming}{\mathsf{AR}}		%
\newcommand{\numberReprogrammingInstances}{R}
\newcommand{\indexReprogrammingInstances}{r}

\newcommand{\gameRepro}{\mathsc{Repro}}					%
\newcommand{\Reprogram}{\mathsc{Reprogram}\xspace}

\newcommand{\xFirst}{x}
\newcommand{\xSecond}{x'}
\newcommand{\XFirst}{X}
\newcommand{\XSecond}{X'}

\newcommand{\SignatureScheme}{\ensuremath{\mathsf{SIG}}\xspace}
\newcommand{\sigS}{\SignatureScheme}
\newcommand{\KG}{\ensuremath{\mathsf{KG}}\xspace}
\newcommand{\Sign}{\ensuremath{\mathsf{Sign}}\xspace}
\newcommand{\VerifySig}{\ensuremath{\mathsf{Vrfy}}\xspace}

\newcommand{\SKSpace}{{\mathcal{SK}}}					%
\newcommand{\MSpace}{{\mathcal{M}}}						%
\newcommand{\RSpace}{{\mathcal{R}}}						%
\newcommand{\RSpaceSign}{\RSpace_{\Sign}}				%

\newcommand{\FS}{\mathsf{FS}}							%
\newcommand{\ROChallenge}{\RO{H}}						%

\newcommand{\hts}{\ensuremath{\mathrm{HaS}}\xspace}%
\newcommand{\sigSp}{\ensuremath{\SignatureScheme'}\xspace}
\newcommand{\MSpaceAfterHTS}{{\mathcal{M}'}}			%
\newcommand{\messageAfterHTS}{m'}						%
\newcommand{\MSpaceBeforeHTS}{\MSpace}					%
\newcommand{\messageBeforeHTS}{m}						%
\newcommand{\signatureAfterHTS}{\sigma'}				%

\newcommand{\UF}{\ensuremath{\mathsf{UF}}}
\newcommand{\KOA}{\ensuremath{\mathsf{KOA}}\xspace}
\newcommand{\CMAZero}{\ensuremath{\CMA_0}\xspace}
\newcommand{\CMA}{\ensuremath{\mathsf{CMA}}\xspace}
\newcommand{\UFCMAZero}{\ensuremath{\UF\text{-}\CMAZero}\xspace}
\newcommand{\UFKOA}{\ensuremath{\UF\text{-}\KOA}\xspace}
\newcommand{\UFCMA}{\ensuremath{\UF\text{-}\CMA}\xspace}

\newcommand{\eufrma}{\ensuremath{\mathsf{UF\text{-}RMA}}\xspace}
\newcommand{\eufcma}{\ensuremath{\mathsf{UF\text{-}CMA}}\xspace}%
\newcommand{\RMA}{\ensuremath{\mathsf{RMA}}\xspace}

\newcommand{\oracleSIGN}{\textsc{SIGN}\xspace}
\newcommand{\getSignature}{\ensuremath{\mathsf{getSignature}}}

\newcommand{\ListOfMessages}{\List{\MSpace}}

{
}
\newcommand{\TrafoHedging}{{\mathsf{R2H}}}
\newcommand{\SigSchemeHedged}{{\SignatureScheme'}}

\newcommand{\SigningHedged}{\Sign'}
\newcommand{\ROforHedging}{\RO{G}}

\newcommand{\nonce}{n}	
\newcommand{\NonceSpace}{\mathcal{N}}

\newcommand{\simSignature}{\mathsf{simSignature}}

\newcommand{\SetFaultFunctions}{\mathcal{F}}					%
\newcommand{\faultCMA}%
	{{\mathsf{F}_\SetFaultFunctions\text{-}\CMA}}
\newcommand{\UFfaultCMA}{{\UF\text{-}\faultCMA}}
\newcommand{\oracleFaultSIGN}{\mathsf{FAULTSIGN}}
\newcommand{\faultCMADifferentSet}[1]%
	{{\mathsf{F}_{#1}\text{-}\CMA}}
\newcommand{\UFfaultCMADifferentSet}[1] {{\UF\text{-}\faultCMADifferentSet{#1}}}

\newcommand{\nonceFaultCMA}{{\mathsf{N}\text{-}\faultCMA}}
\newcommand{\nonceFaultCMADifferentSet}[1]{{\mathsf{N}\text{-}\faultCMADifferentSet{#1}}}
\newcommand{\UFnonceFaultCMA}{{\UF\text{-}\nonceFaultCMA}}
\newcommand{\UFnonceFaultCMADifferentSet}[1]{{\UF\text{-}{\nonceFaultCMADifferentSet{#1}}}}
\newcommand{\oracleNonceFaultSIGN}{\mathsf{N}\text{-}\mathsf{FAULTSIGN}} %
\newcommand{\IdScheme}{\mathsf{ID}}
\newcommand{\IG}{\mathsf{IG}}
\newcommand{\Commit}{\mathsf{Commit}}
\newcommand{\Respond}{\mathsf{Respond}}
\newcommand{\VerifyId}{\mathsf{V}}

\newcommand{\CommSpace}{\mathcal{W}}				%
\newcommand{\ChallengeSpace}{\mathcal{C}}			%
\newcommand{\ResponseSpace}{\mathcal{Z}}			%

\newcommand{\pk}{\ensuremath{\mathit{pk}}\xspace}
\newcommand{\sk}{\ensuremath{\mathit{sk}}\xspace}
\newcommand{\commitment}{w}
\newcommand{\challenge}{c}
\newcommand{\response}{z}
\newcommand{\transcript}{trans}

\newcommand{\maxEntropyCommit}{\gamma(\Commit)}			

\newcommand{\HVZK}{\mathsf{HVZK}}
\newcommand{\specialHVZK}{\mathsf{sHVZK}}			%
\newcommand{\Sim}{\mathsf{Sim}}						%

\newcommand{\boundStatisticalHVZK}{\Delta_{\HVZK}}
\newcommand{\boundSpecialStatisticalHVZK}{\Delta_{\specialHVZK}}

\newcommand{\getTrans}{\mathsf{getTrans}}			%
\newcommand{\oraclegetTrans}{\mathsf{getTrans}}		%
\newcommand{\getTransChallenge}{\mathsf{getTransChall}}	%

\newcommand{\AdversaryHVZK}{\Ad{C}}					%

\newenvironment{nicodemus}[1][\thenicolinenr]{%
	\begin{enumerate}[
		topsep=0ex,
		label=\nicolinenrformat\PaddingUp*,
		ref=\nicorefprefix\PaddingUp*,
		align=right,
		leftmargin=0em,
		itemindent=!,
		labelindent=0em,
		labelwidth=\nicolinenrwidth,
		labelsep=\nicolinenrsep,
		listparindent=\parindent,
		noitemsep,
		]%
		\setcounter{enumi}{#1}%
		\addtocounter{enumi}{-1}%
	}{%
	\end{enumerate}%
	\addtocounter{enumi}{1}%
	\setcounter{nicolinenr}{\theenumi}%
}

\newcommand{\INDCPA}{\mathsf{IND}\text{-}\mathsf{CPA}}

\titlerunning{Tight adaptive reprogramming in the QROM}
\title{Tight adaptive reprogramming in the QROM}

\ifsubmission
\author{\vspace{-0.5in}}
\institute{}
\else
\author{
  Alex B. Grilo\inst{1}
\and
  Kathrin Hövelmanns\inst{2}
\and
  Andreas Hülsing\inst{3}
\and
  Christian Majenz\inst{4}
}
\institute{
  Sorbonne Universit\'{e}, CNRS, LIP6, France\\
\and
Ruhr-Universit\"at Bochum, Germany \\ %
\and
Eindhoven University of Technology,
The Netherlands\\
\and
Centrum Wiskunde \& Informatica and QuSoft, Amsterdam, The Netherlands \\
\email{authors-qrom-reprog@huelsing.net}
}
\authorrunning{A. B. Grilo, K. Hövelmanns, A. Hülsing, C. Majenz}
\fi

\begin{document}

\maketitle
	
\begin{abstract}
The random oracle model (ROM) enjoys widespread popularity, mostly because it tends to allow for \emph{tight} and  \emph{conceptually simple} proofs where provable security in the standard model is elusive or costly. While being the adequate replacement of the ROM in the post-quantum security setting, the quantum-accessible random oracle model (QROM) has thus far failed to provide these advantages in many settings. In this work, we focus on \emph{adaptive reprogrammability}, a feature of the ROM enabling tight and simple proofs in many settings. We show that the straightforward quantum-accessible generalization of adaptive reprogramming is feasible by proving a bound on the adversarial advantage in distinguishing whether a random oracle has been reprogrammed or not. We show that our bound is tight by providing a matching attack. We go on to demonstrate that our technique recovers the mentioned advantages of the ROM in three QROM applications:  1) We give a tighter proof of security of the message compression routine as used by XMSS.
2) We show that the standard ROM proof of chosen-message security for Fiat-Shamir signatures can be lifted to the QROM, straightforwardly, achieving a tighter reduction than previously known.
3) We give the first QROM proof of security against fault injection and nonce attacks for the hedged Fiat-Shamir transform.
\\[6pt]
\end{abstract}

\ifsubmission
\else
\begingroup
\makeatletter
\def\@thefnmark{} \@footnotetext{\relax
Part of this work was done while A.G. was affiliated to CWI and QuSoft and part of it was done while A.G. was visiting the Simons Institute for the Theory of Computing.
K.H. was supported by the European Union PROMETHEUS project (Horizon 2020 Research and Innovation Program, grant 780701) and the Deutsche Forschungsgemeinschaft (DFG, German Research Foundation) under Germany’s Excellence Strategy (EXC 2092 CASA, 390781972).
C.M. was funded by a NWO VENI grant (Project No. VI.Veni.192.159).
Date: \today}
\endgroup
\fi

\section{Introduction}
Since its introduction, the Random oracle model (ROM) has allowed cryptographers to prove efficient practical cryptosystems secure 
for which proofs in the standard model have been elusive.
In general, the ROM allows for proofs that are conceptually simpler and often tighter than standard model security proofs. 

With the advent of post-quantum cryptography, and the introduction of quantum adversaries,
the ROM had to be generalized:
In this scenario, a quantum adversary interacts with a non-quantum network, meaning that "online" primitives (like signing) stay classical,
while the adversary can compute all "offline" primitives (like hash functions) on its own, and hence, in superposition.
To account for these stronger capabilities, the quantum-accessible ROM (QROM) was introduced \cite{AC:BDFLSZ11}.
While successfully fixing the definitional gap, 
the QROM does not generally come with the advantages of its classical counterpart:
\begin{itemize}
	\item[-] \emph{Lack of conceptual simplicity.}
		QROM proofs are extremely complex for various reasons.
		One reason is that they require some understanding of quantum information theory.
		More important, however, is the fact that many of the useful properties of the ROM (like preimage awareness and adaptive programmability)
		are not known to translate directly to the QROM.
	\item[-] \emph{Tightness.} Many primitives that come with tight security proofs in the ROM are not known to be supported by tight proofs in the QROM.
		For example, there has been an ongoing effort \cite{EC:SaiXagYam18,C:JZCWM18,PKC:JiaZhaMa19,TCC:BHHHP19,EC:KSSSS20,PKC:HKSU20} to give tighter QROM proofs for the well-known Fujisaki-Okamoto transformation \cite{C:FujOka99,JC:FujOka13},
		which is proven tightly secure in the ROM as long as the underlying scheme fulfills $\INDCPA$ security \cite{TCC:HofHovKil17}.
\end{itemize}

In many cases, we expect certain generic attacks to only differ from the ROM counterparts by a square-root factor in the required number of queries if the attack involves a search problem,
or no significant factor in the case of guessing. %
Hence, it was conjectured that it might be sufficient to prove security in the ROM, and then add a square-root factor for search problems.
However, recent results \cite{SeparationQROM} demonstrate a separation of ROM and QROM, showing that this conjecture does not hold true in general,
as there exist schemes which are provably secure in the ROM and insecure in the QROM. As a consequence, a QROM proof is crucial to establish confidence in a post-quantum cryptosystem.\footnote{Unless, of course, a standard model proof is available.}

\heading{Adaptive programmability.}
A desirable property of the (classical) ROM is that any oracle value $\RO{O}(x)$ can be chosen
when $\RO{O}$ is queried on $x$ for the first time (lazy-sampling).
This fact is often exploited by a reduction simulating a security game without knowledge of some secret information.
Here, an adversary $\Ad{A}$ will not recognize the reprogramming of $\RO{O}(x)$ as long as the new value is  uniformly distributed and consistent with the rest of $\Ad{A}$'s view.
This property is called \textit{adaptive programmability}.

The ability to query an oracle in superposition renders this formerly simple approach more involved, similar to the difficulties arising from the question how to extract classical preimages from a quantum query (preimage awareness) \cite{EC:Unruh14,C:AmbHamUnr19,TCC:BHHHP19,EC:KSSSS20,C:Zhandry19,C:DFMS19,C:LiuZha19,BL20,CMP20}.
Intuitively, a query in superposition can be viewed as a query that might contain all input values at once.
Already the first answer of $\RO{O}$ might hence contain information about every value $\RO{O}(x)$
that might need to be reprogrammed as the game proceeds.
It hence was not clear whether it is possible to adaptively reprogram a quantum random oracle without causing a change in the adversary's view.

Until recently, both properties only had extremely non-tight variants in the QROM.
For preimage awareness, it was essentially necessary to randomly guess the right query and measure it (with an unavoidable loss of at least $1/q$ for $q$ queries, and the additional disadvantage of  potentially rendering the adversary's output unusable due to measurement disturbance).
In a recent breakthrough result, Zhandry developed the compressed oracle technique that provides preimage awareness~\cite{C:Zhandry19} in many settings. 
For adaptive reprogramming, variants of Unruh's one-way-to-hiding lemma allowed to prove bounds but only with a square-root loss in the entropy of the reprogramming position   \cite{C:Unruh14,eaton2015making,PKC:HulRijSon16}.

In some cases \cite{AC:BDFLSZ11,EC:KilLyuSch18,EC:SaiXagYam18,PKC:HKSU20},
reprogramming could even be avoided by giving a proof that rendered the oracle ``a-priori consistent'',
which is also called a ``history-free'' proof:
In this approach, the oracle is completely redefined in a way such that it is enforced to be \emph{a priori} consistent with the rest of an adversary's view,
meaning that it is redefined before execution of the adversary, and on \emph{all} possible input values.
Unfortunately, it is not always clear whether it is possible to lift a classical proof to the QROM with this strategy.
Even if it is, the ``a-priori'' approach usually leads to conceptually more complicated proofs.
More importantly, it can even lead to reductions that are non-tight with respect to runtime, and may necessitate stronger or additional requirements like, e.g., the statistical counterpart of a property that was only used in its computational variant in the ROM. %
An example is the CMA-security proof for Fiat-Shamir signatures that was given in \cite{EC:KilLyuSch18}.

Hence, in this work we are interested in the question:
\begin{quote}
	{\bf 
		Can we \emph{tightly} prove that adaptive reprogramming can also be done in the quantum random oracle model?}
\end{quote}

\subheading{Our contribution.}
For common use cases in the context of post-quantum cryptography, this work answers the question above in the affirmative.
In more detail, we present a tool for adaptive reprogramming that comes with a tight bound,
supposing that the reprogramming positions hold sufficiently large entropy,
and reprogramming is triggered by classical queries to an oracle that is provided by the security game (e.g., a signing oracle).
These preconditions are usually met in (Q)ROM reductions: The reprogramming is usually triggered by adversarial signature or decryption queries,
which remain classical in the post-quantum setting, as the oracles represent honest users.

While we prove a very general lemma, using the simplest variant of the superposition oracle technique \cite{C:Zhandry19}, we present two corollaries, 
tailored to cases like a) hash-and-sign with randomized hashing and b) Fiat-Shamir signatures.
In both cases, reprogramming occurs at a position of which one part is an adversarially choosen string.
For a), the other part is a random string $z$, sampled by the reduction (simulating the signer).
For b), the other part is a commitment $\commitment$ chosen from a distribution with sufficient min-entropy, together with additional side-information.
In both cases, we manage to bound the distinguishing advantage of any adversary that makes $\qsign$ signing and $\qhash$ random oracle queries by 
$$1.5 \cdot \qsign \sqrt{\qhash \cdot 2^{-r}} \enspace ,$$
where $r$ is the length of $z$ for a), and the min-entropy of $\commitment$ for b). 

We then demonstrate the applicability of our tool, by giving 
\begin{itemize}
 \item a tighter proof for hash-and-sign applications leading a tighter proof for the message-compression as used by the hash-based signature scheme XMSS in RFC 8391~\cite{RFC8391} as a special case,
 \item a runtime-tight reduction of unforgeability under adaptive chosen message attacks ($\UFCMA$) to plain unforgeability ($\UFCMAZero$, sometimes denoted $\UFKOA$ or $\mathsf{UF}\text{-}\mathsf{NMA}$) for Fiat Shamir signatures.
 \item the first proof of fault resistance for the hedged Fiat-Shamir transform, recently proposed in~\cite{EC:AOTZ20}, in the post-quantum setting.
\end{itemize}

\heading{Hash-and-sign.}
As a first motivating and mostly self-contained application we analyze the hash-and-sign construction that takes a fixed-message-length signature scheme \sigS and turns it into a variable-message-length signature scheme \sigSp by first compressing the message using a hash function. We show that if \sigS is secure under random message attacks (\eufrma), \sigSp is secure under adaptively chosen message attacks (\eufcma). Then we show that along the same lines, we can tighten a recent security proof~\cite{XMSSembed} for message-compression as described for XMSS~\cite{PQCRYPTO:BucDahHul11} in RFC 8391. Our new bound shows that one can use random strings of half the length to randomize the message compression in a provably secure way.

\heading{The Fiat-Shamir transform.}
In \cref{sec:FiatShamirInQROM}, we show that if an identification scheme $\IdScheme$ is \underline{H}onest-\underline{V}erifier \underline{Z}ero-\underline{K}nowledge ($\HVZK$),
and if the resulting Fiat-Shamir signature scheme $\SignatureScheme := \FS[\IdScheme, \ROChallenge]$ furthermore possesses $\UFCMAZero$ security,
then $\SignatureScheme$ is also $\UFCMA$ secure, in the quantum random oracle model.
Here, $\UFCMAZero$ denotes the security notion in which the adversary only obtains the public key and has to forge a valid signature without access to a signing oracle.
While this statement was already proven in \cite{EC:KilLyuSch18},
we want to point out several advantages of our proof strategy and the resulting bounds. %
	\vspace{.1cm}

	\noindent{\bf Conceptual simplicity.}  
		A well-known proof strategy for $\HVZK, \UFCMAZero \Rightarrow \UFCMA$
		in the random oracle model (implicitly contained in \cite{EC:AFLT12})
		is to replace honest transcripts with simulated ones,
		and to render $\ROChallenge$ \emph{a-posteriori} consistent with the signing oracle
		during the proceedings of the game.
		I.e., $\ROChallenge(\commitment, m)$ is patched \textit{after} oracle $\oracleSIGN$ was queried on $m$.
		Applying our lemma, we observe that this approach actually works in the quantum setting as well.
		We obtain a very simple QROM proof that is congruent with its ROM counterpart.
		
		In \cite{EC:KilLyuSch18}, the issue of reprogramming quantum random oracle $\ROChallenge$ was circumvented by giving a history-free proof:
		In the proof, messages are tied to potential transcripts by generating the latter with message-dependent randomness, \textit{a priori},
		and $\ROChallenge$ is patched accordingly, right from the beginning of the game.
		During each computation of $\ROChallenge(\commitment, m)$, 
		the reduction therefore has to keep $\ROChallenge$ a-priori consistent 
		by going over all transcript candidates $(\commitment_i, \challenge_i, \response_i)$ belonging to $m$, and returning $\challenge_i$ if $\commitment = \commitment_i$.
		\vspace{.1cm}

	\noindent{\bf Tightness with regards to running time.}
		Our reduction $\Ad{B}$ has about the running time of the adversary $\Ad{A}$,
		as it can simply sample simulated transcripts and reprogram $\ROChallenge$, accordingly.
		The reduction in \cite{EC:KilLyuSch18} suffers from a quadratic blow-up in its running time:
		They have running time
		$\Time(\Ad{B}) \approx \Time(\Ad{A}) + q_{\RO{H}}q_S$,
		as the reduction has to execute $q_S$ computations upon each query to $\ROChallenge$ in order to keep it a-priori consistent.
		As they observe, this quadratic blow-up renders the reduction non-tight in all practical aspects.
		On the other hand, our upper bound of the advantage comes with a bigger disruption in terms of commitment entropy (the min-entropy of the first message (the \emph{commitment}) in the identification scheme).
		While the source of non-tightness in \cite{EC:KilLyuSch18} can not be balanced out,
		however, we offer a trade-off:
		If needed, the commitment entropy can be increased by appending a random string to the commitment.\footnote{While this increases the signature size, the increase is mild in typical post-quantum Fiat-Shamir based digital signature schemes. As an example, suppose Dilithium-1024x768, which has a signature size of  2044 bytes, had zero commitment entropy (it actually has quite some, see remarks in \cite{EC:KilLyuSch18}). To ensure that about $2^{128}$ hash queries are necessary to make the term in our security bound that depends on the commitment entropy equal 1, about 32 bytes would need to be added, an increase of  about 1.6\% (assuming $2^{64}$ signing queries).}
		\vspace{.1cm}

	\noindent{\bf Generality.}
		To achieve a-priori consistency, \cite{EC:KilLyuSch18} crucially relies on \textit{statistical} $\HVZK$.
		Furthermore, they require that the $\HVZK$ simulator outputs transcripts such that the challenge $\challenge$ is uniformly distributed.
		We are able to drop the requirement on $\challenge$ altogether,
		and to only require \textit{computational} $\HVZK$.
		(As a practical example, alternate NIST candidate Picnic \cite{NIST:PICNIC} satisfies only \emph{computational} $\HVZK$.)

\heading{Robustness of the hedged Fiat-Shamir transform against fault attacks.}
When it comes to real-world implementations, the assessment of a signature scheme will not solely take into consideration
whether an adversary could forge a fresh signature as formalized by the $\UFCMA$ game,
as the $\UFCMA$ definition does not capture all avenues of real-world attacks.
For instance, an adversary interacting with hardware that realizes a cryptosystem can try to induce a hardware malfunction, also called fault injection,
in order to derail the key generation or signing process.
Although it might not always be straightforward to predict where exactly a triggered malfunction will affect the execution,
it is well understood that even a low-precision malfunction can seriously injure a schemes' security.
In the context of the ongoing effort to standardize post-quantum secure primitives \cite{NIST:Competition}, 
it hence made sense to affirm \cite{NIST:StatusReport2020} that desirable additional security features include, amongst others,
resistance against fault attacks and randomness generation that has some bias.

Very recently \cite{EC:AOTZ20}, the hedged Fiat-Shamir construction
was proven secure against biased nonces and several types of fault injections, in the ROM.
This result can for example be used to argue that alternate NIST candidate Picnic \cite{NIST:PICNIC} is robust against many types of fault injections.
We revisit the hedged Fiat-Shamir construction in \cref{sec:HedgedFiatShamir}
and lift the result of \cite{EC:AOTZ20} to the QROM.
In particular, we thereby obtain that Picnic is resistant against many fault types, even when attacked by an adversary with quantum capabilities.

We considered to generalize the result further by replacing the standard Fiat-Shamir transform with the Fiat–Shamir with aborts transform that was introduced by Lyubashevsky \cite{AC:Lyubashevsky09,EC:KilLyuSch18}.
Recall that Fiat–Shamir with aborts was established due to the fact that for some underlying lattice-based ID schemes (e.g., NIST finalist Dilithium \cite{NIST:DILITHIUM}), the prover sometimes cannot create a correct response to the challenge,
and the protocol therefore allows for up to $\kappa$ many retries during the signing process.
While our security statements can be extended in a straightforward manner,
we decided not to further complicate our proof with the required modifications.
For Dilithium, the implications are limited anyway, as several types of faults are only proven ineffective if the underlying scheme is subset-revealing, which Dilithium is not.%
\footnote{
	Intuitively, an identification scheme is called subset-revealing if its responses do not depend on the secret key. 
	Dilithium computes its responses as $\response := y + \challenge \cdot s_1$, where $s_1$ is part of the secret key.
}

\heading{Optimality of our bound.}
We also show that our lower bound is tight for the given setting, presenting a 
quantum attack that matches our bound, up to a constant factor. Let us restrict our attention to the simple case where $H : \{0,1\}^n \to \{0,1\}^k$ is a random function, which is potentially reprogrammed at a random position $x^*$ resulting in a new oracle $H'$.  Consider an attacker that is allowed $2q$ queries to the random oracle.

A classical attack that matches the classical bound for the success probability, $O(q\cdot 2^{-n})$, is the following: pick values $x_1,...,x_q$\ and compute the XOR of the outputs  $H(x_i)$%
. After the oracle is potentially reprogrammed, the attacker outputs $0$ iff %
the checksum computed before is unchanged.  %

In order to match the quantum lower bound, we use the same attack, but on a superposition of tuples of inputs: the attacker queries $H$ with the superposition of all possible inputs, and then applies a cyclic permutation $\sigma$ on the input register.  This process is repeated $q-1$ times (on the same state). After the potential reprogramming, we repeat the same process, but now applying the permutation $\sigma^{-1}$ and querying $H'$.  Using techniques from \cite{EC:AlaMajRus20}, we show how to distinguish the two cases with advantage $\Omega\left(\sqrt{\frac{q}{2^n}}\right)$ in time $\poly(q,n)$.
\section{Adaptive reprogramming: the toolbox}\label{sec:Presentation_of_reprogramming}
Before we describe our adaptive reprogramming theorem, let us quickly recall how we usually model adversaries with quantum access to a random oracle:
As established in \cite{AC:BDFLSZ11,FOCS:BBCMW98}, we model quantum access to a random oracle $\RO{O}: X \times Y$ via oracle access to a unitarian $U_{\RO{O}}$, which is defined as the linear completion of $\ket{x}_X \ket{y}_Y	\mapsto \ket{x}_X \ket{y \oplus \RO{O}(x)}_Y$,
and adversaries $\Ad{A}$ with quantum access to $\RO{O}$ as a sequence of unitarians, interleaved with applications of $U_{\RO{O}}$.
We write $\Ad{A}^{\qRO{O}}$	to indicate that $\RO{O}$ is quantum-accessible.

As a warm-up, we will first present our reprogramming lemma in the simplest setting.
Say we reprogram an oracle $\numberReprogrammingInstances$ many times,
where the position is partially controlled by the adversary, and partially picked at random.
More formally, let $\XFirst_1$ and $\XFirst_2$ be two finite sets,
where $\XFirst_1$ specifies the domain from which the random portions are picked,
and $\XFirst_2$ specifies the domain of the adversarially controlled portions.
We will now formalize what it means to distinguish a random oracle $\RO{O_0}: 
\XFirst_1 \times \XFirst_2 \rightarrow Y$ from its reprogrammed version 
$\RO{O_1}$.
Consider the two $\gameRepro$ games, given in \cref{fig:Def:Repro:Basic}:
In games $\gameRepro_b$, the distinguisher has quantum access to oracle $\RO{O_b}$ (see line~\ref{line:Def:Reprogram:OracleAcces:Basic})
that is either the original random oracle $\RO{O_0}$ (if $b=0$),
or the oracle $\RO{O_1}$ which gets reprogrammed adaptively ($b=1$).
To model the actual reprogramming, we endow the distinguisher with (classical) access to a reprogramming oracle $\Reprogram$.
Given a value $\xFirst_2 \in \XFirst_2$, oracle $\Reprogram$ samples random 
values $\xFirst_1$ and $y$,
and programs the random oracle to map $\xFirst_1 \| \xFirst_2$ to $y$ (see line~\ref{line:Def:Reprogram:Reprogramming:Basic}).
Note that apart from already knowing $\xFirst_2$, the adversary even learns the part $\xFirst_1$ of the position at which $\RO{O_1}$ was reprogrammed.

\begin{figure}[h] \begin{center} \fbox{
			
	\nicoresetlinenr	
	
	\begin{minipage}[t]{3.7cm}
		
		\underline{\textbf{GAME} $\gameRepro_{b}$}
		\begin{nicodemus}
			
			\item $\RO{O}_0 \uni Y^{\XFirst_1 \times \XFirst_2}$
			
			\item $\RO{O}_1 := \RO{O}_0$ 
			
			\item $b' \leftarrow \Ad{A}^{\qRO{O_b}, \Reprogram}$
			\label{line:Def:Reprogram:OracleAcces:Basic}
			
			\item \pcreturn $b'$
			
		\end{nicodemus}
		
	\end{minipage}
	
	\quad
	
	\begin{minipage}[t]{3.5cm}
		
		\underline{$\Reprogram(\xFirst_2)$}
		\begin{nicodemus}
			
			\item $(\xFirst_1, y) \uni \XFirst_1 \times Y$
			
			\item $\RO{O}_1 := \RO{O}_1^{(\xFirst_1 \| \xFirst_2) \mapsto y}$ \label{line:Def:Reprogram:Reprogramming:Basic}
			
			\item \pcreturn $\xFirst_1$
			
		\end{nicodemus}
		
	\end{minipage}
			
}\end{center}
	\ifnum\tightOnSpace=1 \vspace{-0.4cm} \fi
	\caption{Adaptive reprogramming games $\gameRepro_b$ for bit $b \in \bits$ in the most basic setting.}
	\label{fig:Def:Repro:Basic}
	\ifsubmission
	\vspace{-.2in}
	\fi
\end{figure}

\begin{proposition}\label{prop:basic}
	Let $\XFirst_1$, $\XFirst_2$ and $Y$ be finite sets,
	and let $\Ad{A}$ be any algorithm 
	issuing $\numberReprogrammingInstances$ many calls to $\Reprogram$
	and $q$ many (quantum) queries to $\RO{O_b}$ as defined in \cref{fig:Def:Repro:Basic}.
	Then the distinguishing advantage of $\Ad{A}$ is bounded by
	\begin{equation}	
	\label{eq:Advantage:ReprogrammingGameBasedBasic}
	| \Pr[\gameRepro_{1}^{\Ad{A}} \Rightarrow 1] -  \Pr[\gameRepro_{0}^{\Ad{A}} \Rightarrow 1] |
	\leq \frac {3R} 2  \sqrt{\frac{q}{|\XFirst_1|}}. 
	\end{equation}
\end{proposition}
The above theorem constitutes a significant improvement over previous bounds. In \cite{C:Unruh14} and \cite{eaton2015making}, a bound proportional to $q|\XFirst_1|^{-1/2}$ for the distinguishing advantage in similar settings, but for $R=1$, was given. In \cite{PKC:HulRijSon16}, a bound proportional to $q^2|\XFirst_1|^{-1}$ is claimed, but that seems to have resulted from a ``translation mistake'' from \cite{eaton2015making} and should be similar to the bounds from \cite{C:Unruh14,eaton2015making}. What is more, we show in \cref{sec:Attack} that the above bound, and therefore also its generalizations, are tight, by presenting a distinguisher that achieves an advantage equal to the right hand side of \cref{eq:Advantage:ReprogrammingGameBasedBasic} for trivial $\XFirst_1$, up to a constant factor.

In fact, we prove something more general than \cref{prop:basic}: We prove that an adversary will not behave significantly different, even if
\begin{itemize}
	\item[-] the adversary does not only control a portion $\xFirst_2$, but 
	instead it even controls the distributions according to which the whole positions $x := (\xFirst_1, \xFirst_2)$ are sampled at which $\RO{O}_1$ is reprogrammed,
	\item[-] it can additionally pick different distributions, adaptively, and
	\item[-] the distributions produce some additional side information $\xSecond$ which the adversary also obtains, 
\end{itemize}
as long as the reprogramming positions $x$ hold enough entropy.

Overloading notation, we formalize this generalization by games $\gameRepro$, 
given in \cref{fig:Def:Repro}:
Reprogramming oracle $\Reprogram$ now takes as input the description of a distribution $p$ that generates a whole reprogramming position $\xFirst$,
together with side information $\xSecond$.
$\Reprogram$ samples $\xFirst$ and $\xSecond$ according to $p$,
programs the random oracle to map $\xFirst$ to a random value $y$,
and returns $(\xFirst, \xSecond)$.

\begin{figure}[h] \begin{center} \fbox{
			
	\nicoresetlinenr	
	
	\begin{minipage}[t]{3.7cm}
		
		\underline{\textbf{GAME} $\gameRepro_{b}$}
		\begin{nicodemus}
			
			\item $\RO{O}_0 \uni Y^X$
			
			\item $\RO{O}_1 := \RO{O}_0$ 
			
			\item $b' \leftarrow \Ad{D}^{\qRO{O_b}, \Reprogram}$
			\label{line:Def:Reprogram:OracleAcces}
			
			\item \pcreturn $b'$
			
		\end{nicodemus}
		
	\end{minipage}
	
	\quad
	
	\begin{minipage}[t]{2.8cm}
		
		\underline{$\Reprogram(p)$}
		\begin{nicodemus}
			
			\item $(\xFirst,\xSecond) \leftarrow p$
			
			\item $y \uni Y$
			
			\item $\RO{O}_1 := \RO{O}_1^{\xFirst \mapsto y}$
			\label{line:Def:Reprogram:Reprogramming}
			
			\item \pcreturn $(\xFirst, \xSecond)$
			
		\end{nicodemus}
		
	\end{minipage}
			
}\end{center}
	\ifnum\tightOnSpace=1 \vspace{-0.4cm} \fi
	\caption{Adaptive reprogramming games $\gameRepro_b$ for bit $b \in \bits$.}
	\label{fig:Def:Repro}
	\ifsubmission
	\vspace{-.2in}
	\fi
\end{figure}

We are now ready to present our main \cref{theorem:reprGameBased}.
On a high level, the only difference between the statement of \cref{prop:basic} and \cref{theorem:reprGameBased} is that we now have
to consider $\numberReprogrammingInstances$ many (possibly different) joint distributions on  $\XFirst \times \XSecond$,
and to replace $\frac{1}{|\XFirst_1|}$ (the probability of the uncontrolled reprogramming portion)
with the highest likelihood of any of those distributions generating a position $\xFirst$.
\begin{restatable}[``Adaptive reprogramming'' ($\adaptiveReprogramming$)]{theorem}{thmone}
\label{theorem:reprGameBased}
	Let $\XFirst$, $\XSecond$, $Y$ be some finite sets,
	and let $\Ad{D}$ be any distinguisher, issuing $\numberReprogrammingInstances$ many reprogramming instructions
	and $q$ many (quantum) queries to $\RO{O}$.
	Let $q_r$ denote the number of queries to $\RO{O}$ that are issued
	inbetween the $(r-1)$-th and the $r$-th query to $\Reprogram$.
	Furthermore, let $p^{(r)}$ denote the $r$th distribution that $\Reprogram$ is queried on.
	By $p^{(r)}_{\XFirst}$ we will denote the marginal distribution of $\XFirst$, according to $p^{(r)}$,
	and define
	\[
	p_{\max}^{(r)} := \mathbb E \max_{x} p^{(r)}_{\XFirst}(x),
	\]
	where the expectation is taken over $\Ad{D}$'s behaviour until its $r$th query to $\Reprogram$.
	\begin{equation}	\label{eq:Advantage:ReprogrammingGameBased}
	| \Pr[\gameRepro_{1}^{\Ad{D}} \Rightarrow 1] - \Pr[\gameRepro_{0}^{\Ad{D}} \Rightarrow 1] |
	\leq \sum_{r = 1}^{\numberReprogrammingInstances}  \left( \sqrt{ \hat q_{r} p_{\max}^{(r)} }  + \frac 1 2 \hat q_{r} p_{\max}^{(r)} 
	\right) \enspace ,
	\end{equation}
	where $\hat q_{r} := \sum_{i=0}^{r-1} q_{i}$.

\end{restatable}

For $R=1$ and without additional side information output $x'$, the proof of \cref{theorem:reprGameBased} is given in \cref{sec:AdaptiveReprogramming}.  The extension to general $R$ is proven in \cref{sec:AdaptiveReprogramming:Appendix} via a standard hybrid argument. Finally, all our bounds are information-theoretical, i.e. they hold against arbitrary query bounded adversaries. The additional output $x'$ can therefore be sampled by the adversary (see details in \cref{sec:AdaptiveReprogramming:Appendix}).

We will now quickly discuss how to simplify the bound given in \cref{eq:Advantage:ReprogrammingGameBased} for our applications,
and in particular, how we can derive \cref{eq:Advantage:ReprogrammingGameBasedBasic} from \cref{theorem:reprGameBased}:
Throughout sections \ref{sec:Hashapps} and \ref{sec:SignaturesInQROM}, we will only have to consider reprogramming instructions
that occur on positions $\xFirst = (\xFirst_1, \xFirst_2)$ such that
\begin{itemize}
	\item[-] $\xFirst_1$ is drawn according to the same distribution $p$ for each reprogramming instruction, and
	\item[-] $\xFirst_2$ represents a message that is already fixed by the adversary.			
\end{itemize}

To be more precise, $\xFirst_1$ will represent a uniformly random string $z$ in \ref{sec:Hashapps}, and no side information $\xSecond$ has to be considered.
In \cref{sec:SignaturesInQROM}, $(\xFirst_1, \xSecond)$ will represent a tuple $(\commitment, \state)$ that is drawn according to $\Commit(\sk)$.

In the language of \cref{theorem:reprGameBased},
the marginal distribution $p^{(r)}_{\XFirst}$ will always be the same distribution $p$, apart from the already fixed part $\xFirst_2$.
We can hence upper bound $p_{\max}^{(r)}$ by $p_{\max} := \max_{\xFirst_1} p (\xFirst_1)$,
and $\hat q_{r}$ by $q$,
to obtain that $\hat q_{r} p_{\max}^{(r)} < q p_{\max}$
for all $1 \leq r \leq \numberReprogrammingInstances$.

In our applications, we will always require that $p$ holds sufficiently large entropy.
To be more precise, we will assume that $p_{\max} < \frac{1}{q}$.
In this case, we have that $ q p_{\max} < 1$, and that we can upper bound $ q p_{\max}$ by $\sqrt{q p_{\max}}$
to obtain

\begin{proposition}\label{prop:advanced}
	Let $\XFirst_1$, $\XFirst_2$, $\XSecond$ and $Y$ be some finite sets, and let $p$ be a distribution on $\XFirst_1 \times \XSecond$.
	Let $\Ad{D}$ be any distinguisher, issuing $q$ many (quantum) queries to $\RO{O}$
	and $\numberReprogrammingInstances$ many reprogramming instructions such that each instruction consists of a value $\xFirst_2$,
	together with the fixed distribution $p$.
	Then
	\begin{equation}
	| \Pr[\gameRepro_{1}^{\Ad{D}} \Rightarrow 1] - \Pr[\gameRepro_{0}^{\Ad{D}} \Rightarrow 1] |
	\leq \frac {3R} 2  \sqrt{ q p_{\max}} 
	\enspace , \nonumber
	\end{equation}
	where $p_{\max} := \max_{\xFirst_1} p (\xFirst_1)$.
\end{proposition}

From this we obtain \cref{prop:basic} setting $p_{max} = |\XFirst_1|$. 
\section{Basic applications}\label{sec:Hashapps}
In this section, we present two motivating examples that benefit from the most basic version of our bound as stated in Proposition~\ref{prop:basic}.
As a first example we chose the canonical hash-and-sign construction
when used to achieve security under adaptive chosen message attacks (\eufcma) from a scheme that is secure under random message attacks (\eufrma). It is mostly self-contained and similar to our second example.
The second example is a tighter bound for the security of hash-and-sign as used in RFC 8391, the recently published standard for the stateful hash-based signature scheme XMSS. For missing definitions and detailed transforms see \cref{sec:SignaturesInQROM:Defs:Appendix}.

\subsection{From \RMA to \CMA security via Hash-and-Sign} \label{subse:HTS:RMA}
In the following, we present a conceptually easy proof with a tighter bound for the canonical \eufrma to \eufcma transform using hash-and-sign $\sigSp=\hts[\sigS,\RO{H}]$, in the QROM (which additionally allows for arbitrary message space expansion).
Recall that $\Sign'(\sk, \messageAfterHTS)$ first samples a uniformly random 
bitstring $z\uni Z$, computes $\sigma \leftarrow \Sign(\sk,\RO{H}(z\| \messageAfterHTS))$ and returns the pair $(z,\sigma)$.
$\VerifySig'$ accordingly first computes $\messageBeforeHTS := \RO{H}(z\| \messageAfterHTS)$ and then calls $\VerifySig(\pk,\messageBeforeHTS,\sigma)$.

The reduction \M from \eufrma to \eufcma in this case works as follows:
First, we have to handle collision attacks.
We show that an adversary which finds a forgery for $\sigSp$ that contains no forgery for $\sigS$ breaks the multi-target version of extended target collision resistance (\mmetcr) of $\RO{H}$, and give a QROM bound for this property.
Having dealt with collision attacks leaves us with the case where $\A$ generates a forgery that contains a forgery for $\sigS$. The challenge in this case is how to simulate the signing oracle $\oracleSIGN$.
Our respective reduction $\M$ against $\eufrma$ proceeds as follows: Collect the $\qsign$ many message-signature pairs $\{(\messageBeforeHTS_i,\sigma_i)\}_{1 \leq i \leq \qsign}$, provided by the \eufrma game.
When $\A$ queries $\oracleSIGN(\messageAfterHTS_i)$ for the $i$th 
time, sample a random $z_i$, reprogram $\RO{H}(z_i\|\messageAfterHTS_i) := \messageBeforeHTS_i$, and return 
$(z_i,\sigma_i)$. See also \cref{fig:cma2rma} below.

In the QROM, this reduction has previously required $\qsign$ applications of the O2H Lemma in two steps, loosing an additive $\mathcal{O}(\qsign\cdot q/\sqrt{|Z|})$ term. In contrast, we only loose a $\mathcal{O}(\qsign\sqrt{q/|Z|})$ (both constants hidden by the $\mathcal{O}$ are small):

\begin{theorem}\label{thm:RMAtoCMA}
	For any (quantum) $\eufcma$ adversary $\A$ issuing at most $\qsign$ (classical) queries to the signing oracle $\oracleSIGN$
	and at most $\qhash$ quantum queries to $\RO{H}$,
	there exists an $\eufrma$ adversary $\M$
	such that
	\begin{align*} \label{eq:bound:FS}
		\succfun^{\eufcma}_{\sigSp}(\A)
		\leq
			\succfun^{\eufrma}_{\sigS}(\M)
			+ \frac{8 \qsign(\qsign+\qhash+2)^2}{|\MSpaceAfterHTS|} 
			+ 3\qsign  \sqrt{ \frac{\qhash  + \qsign + 1}{|Z|}}	\enspace,
	\end{align*}
	and the running time of $\M$
	is about that of $\A$. 
\end{theorem}

The second term accounts for the complexity to find a second preimage for one of the messages $\messageBeforeHTS_i$, which is an unavoidable generic attack. The third term is the result of $2\qsign$ reprogrammings. Half of them are used in the QROM bound for \mmetcr, the other half in the reduction $\M$. This term accounts for an attack that correctly guesses the random bitstring used by the signing oracle for one of the queries (such an attack still would have to find a collision for this part but this is inherently not reflected in the used proof technique).

\begin{proof}
We now relate the \eufcma security of \sigSp to the \eufrma security of \sigS via a sequence of games. 

\heading{Game $G_0$.} We begin with the original \eufcma game for \sigSp in game $G_0$. The success probability of \A in this game is $\Adv^{\eufcma}_{\sigSp}(\A)$ per definition.

\heading{Game $G_1$.} We obtain game $G_1$ from game $G_0$ by adding an additional condition. Namely, game $G_1$ returns $0$ if there exists an $0 < i \leq \qsign$ such that $\RO{H}(z^*\|\messageAfterHTS^*) = \RO{H}(z_i\|\messageAfterHTS_i)$, where $z^*$ is the random element in the forgery signature,
and $z_i$ is the random element in the signature returned by $\oracleSIGN(\messageAfterHTS_i)$ as the answer to the $i$th query.
We will now argue that 
\[
	| \Pr [G_{0}^{\A} \Rightarrow 1] - \Pr [G_1^{\A} \Rightarrow 1] | \leq \frac{8\qsign(\qsign+\qhash+2)^2}{|\MSpaceAfterHTS|} + \frac {3\qsign}{2}  \sqrt{ \frac{\qhash + \qsign + 1}{|Z|}}
\enspace .\]

Towards this end, we give a reduction \B in \cref{fig:mmetcr:reduction}, that breaks the \mmetcr security of $\RO{H}$ whenever the additional condition is triggered, making $\qsign+\qhash+1$ queries to its random oracle. \B simulates the \eufcma game for \sigSp, using $\RO{H}$ and an instance of \sigS.
Clearly, \B runs in about the same time as game $G_0^\A$, and succeeds whenever \A succeeds and the additional condition is triggered. To complete this step, it hence remains to show that the success probability of any such $(\qsign+\qhash+1)$-query adversary is
\begin{equation}\label{eq:Advantage:MeTCR}
	\succmmetcr{\RO{H}}{\B,\qsign} \leq \frac{8 \qsign(\qsign+\qhash+2)^2}{|\MSpaceAfterHTS|} + \frac {3\qsign}{2}  \sqrt{ \frac{\qhash + \qsign +1}{|Z|}} \enspace .
\end{equation}

We delay the proof of \cref{eq:Advantage:MeTCR} until the end. 
\ifeprint
\begin{figure}[t]
 \else 
\begin{figure}[h] \fi
	\begin{center} \fbox{
			
	\nicoresetlinenr	
	
	\begin{minipage}[t]{0.45\textwidth}
		
		\underline{$\B^{\bbox,\qRO{H}}( )$}
		\begin{nicodemus}
			
			\item $(\pk, \sk) \leftarrow \KG$
			
			\item $(\messageAfterHTS^*,\signatureAfterHTS^*) = \A^{\oracleSIGN,\qRO{H}}(\pk)$ 
			
			\item Parse $\signatureAfterHTS^*$ as $(z^*,\sigma^*)$
			
			\item \pcif $\exists j: \RO{H}(z^*\|\messageAfterHTS^*) = \RO{H}(z_j\|\messageAfterHTS_j)$
			
			\item \quad $i := j$
			
			\item \pcelse $i \uni [1, \qsign]$
			
			\item \pcreturn $(\messageAfterHTS^*, z^*,i)$
			
		\end{nicodemus}
		
	\end{minipage}
	
	\qquad
	
	\begin{minipage}[t]{0.35\textwidth}
		
		\underline{$\oracleSIGN(\messageAfterHTS_i)$}
		\begin{nicodemus}
			\item $z_i \exec \bbox(\messageAfterHTS_i)$
			
			\item $\sigma_i \exec \Sign(\sk,\RO{H}(z_i, \messageAfterHTS_i))$
			
			\item \pcreturn $(z_i,\sigma_i)$
			
		\end{nicodemus}
	\end{minipage}
}\end{center}
    \ifnum\tightOnSpace=1 \vspace{-0.4cm} \fi
	\caption{Reduction $\B$ breaking \mmetcr.  Here, $\bbox$ is the \mmetcr challenge oracle.}
	\label{fig:mmetcr:reduction}
	\vspace{-.2in}
\end{figure}

\ifeprint
\begin{figure}[b] 
\begin{center} \fbox{
			
	\nicoresetlinenr	
	
	\begin{minipage}[t]{0.53\textwidth}
		
		\underline{Game $G_2$}
		\begin{nicodemus}
			
			\item $i := 1$
			
			\item $(\pk, \sk) \leftarrow \KG()$
			
			\item $(\messageAfterHTS^*,\signatureAfterHTS^*) = \A^{\oracleSIGN,\qRO{H}}(\pk)$ 
			
			\item Parse $\signatureAfterHTS^*$ as $(z^*,\sigma^*)$
			
			\item \pcif $\exists 1 \leq i\leq \qsign: \RO{H}(z^*\|\messageAfterHTS^*) = \RO{H}(z_i\|\messageAfterHTS_i)$
			
                        \item \quad \pcreturn $0$
			
			\item \pcreturn $\VerifySig(\pk,\messageAfterHTS^*,\sigma^*) \wedge \messageAfterHTS^* \not\in\{\messageAfterHTS_i\}_{i=1}^{\qsign}$
			
		\end{nicodemus}
		
	\end{minipage}
	
	\qquad
	
	\begin{minipage}[t]{0.3\textwidth}
		
		\underline{$\oracleSIGN(\messageAfterHTS_i)$}
		\begin{nicodemus}
			\item $z_i \uni Z, \messageBeforeHTS_i \uni \MSpace$ \label{line:g2:sampling}
			
			\item $\RO{H} := \RO{H}^{(z_i \| \messageAfterHTS_i) \mapsto \messageBeforeHTS_i}$ \label{line:g2:reprog}
			
			\item $\sigma_i \exec \Sign(\sk,\messageBeforeHTS_i)$
			
			\item $i := i+1$
			
			\item \pcreturn $(z_i,\sigma_i)$
			
		\end{nicodemus}
	\end{minipage}
}\end{center}
	\ifnum\tightOnSpace=1 \vspace{-0.4cm} \fi
	\caption{Game $G_2$.}
	\label{fig:cma2rma:game}
\end{figure}
\fi

\heading{Game $G_2$.} The next game differs from $G_1$ in the way the signing oracle works. In game $G_2$ (see \cref{fig:cma2rma:game}), the $i$th query to $\oracleSIGN$ is answered by first sampling a random value $z_i$, as well as a random message $\messageBeforeHTS_i$,
and programming $\RO{H'} := \RO{H'}^{(z_i \| \messageAfterHTS_i) \mapsto \messageBeforeHTS_i}$. Then $\messageBeforeHTS_i$ is signed using the secret key.
We will now show that 

\[| \Pr [G_{1}^{\A} \Rightarrow 1]  - \Pr [G_2^{\A} \Rightarrow 1] | \leq  \frac {3\qsign}{2}  \sqrt{ \frac{\qhash + \qsign + 1}{|Z|}} \enspace .\]

Consider a reduction \C that simulates game $G_2$ for \A to distinguish the $\gameRepro_{b}$ game. Accordingly, \C forwards access to its own oracle $\RO{O}_b$ to \A instead of $\RO{H}$. Instead of sampling $z_i,\messageBeforeHTS_i$ itself in line~\ref{line:g2:sampling} and programming $\RO{H}$ in line~\ref{line:g2:reprog}, \C obtains $z_i\exec\Reprogram(\messageAfterHTS_i)$ from its own oracle and computes $\messageBeforeHTS_i := \RO{O}_b(z_i\|\messageAfterHTS_i)$ as the output of its random oracle. Now, if \C plays in $\gameRepro_{0}$ it perfectly simulates $G_1$ for \A, as the oracle remains unchanged. If \C plays in $\gameRepro_{1}$ it perfectly simulates $G_2$, as can be seen by inlining \Reprogram and removing doubled calls used to recompute $\messageBeforeHTS_i$. Consequently,
\ifeprint
\[
| \Pr[G_1^{\A} \Rightarrow 1] -  \Pr[G_2^{\A} \Rightarrow 1] |
 = 
| \Pr[\gameRepro_{0}^{\C^\A} \Rightarrow 1] -  \Pr[\gameRepro_{1}^{\C^\A} \Rightarrow 1] |
\leq \frac {3\qsign} 2  \sqrt{\frac{\qhash+\qsign+1}{|Z|}} \enspace .
\]
\else
\begin{align*}
| \Pr[G_1^{\A} \Rightarrow 1] -  \Pr[G_2^{\A} \Rightarrow 1] |&\\
&\hspace{-1in} = 
| \Pr[\gameRepro_{0}^{\C^\A} \Rightarrow 1] -  \Pr[\gameRepro_{1}^{\C^\A} \Rightarrow 1] |
\leq \frac {3\qsign} 2  \sqrt{\frac{\qhash+\qsign+1}{|Z|}} \enspace .
\end{align*}
\fi

\ifeprint
\else
\begin{figure}[h] 
\begin{center} \fbox{
			
	\nicoresetlinenr	
	
	\begin{minipage}[t]{0.53\textwidth}
		
		\underline{Game $G_2$}
		\begin{nicodemus}
			
			\item $i := 1$
			
			\item $(\pk, \sk) \leftarrow \KG()$
			
			\item $(\messageAfterHTS^*,\signatureAfterHTS^*) = \A^{\oracleSIGN,\qRO{H}}(\pk)$ 
			
			\item Parse $\signatureAfterHTS^*$ as $(z^*,\sigma^*)$
			
			\item \pcif $\exists 1 \leq i\leq \qsign: \RO{H}(z^*\|\messageAfterHTS^*) = \RO{H}(z_i\|\messageAfterHTS_i)$
			
                        \item \quad \pcreturn $0$
			
			\item \pcreturn $\VerifySig(\pk,\messageAfterHTS^*,\sigma^*) \wedge \messageAfterHTS^* \not\in\{\messageAfterHTS_i\}_{i=1}^{\qsign}$
			
		\end{nicodemus}
		
	\end{minipage}
	
	\qquad
	
	\begin{minipage}[t]{0.3\textwidth}
		
		\underline{$\oracleSIGN(\messageAfterHTS_i)$}
		\begin{nicodemus}
			\item $z_i \uni Z, \messageBeforeHTS_i \uni \MSpace$ \label{line:g2:sampling}
			
			\item $\RO{H} := \RO{H}^{(z_i \| \messageAfterHTS_i) \mapsto \messageBeforeHTS_i}$ \label{line:g2:reprog}
			
			\item $\sigma_i \exec \Sign(\sk,\messageBeforeHTS_i)$
			
			\item $i := i+1$
			
			\item \pcreturn $(z_i,\sigma_i)$
			
		\end{nicodemus}
	\end{minipage}
}\end{center}
	\ifnum\tightOnSpace=1 \vspace{-0.4cm} \fi
	\caption{Game $G_2$.}
	\label{fig:cma2rma:game}
\end{figure}
\fi

\ifeprint
\begin{figure}[t] 
\else
\begin{figure}[h] \fi
\begin{center} \fbox{
			
	\nicoresetlinenr	
	
	\begin{minipage}[t]{0.36\textwidth}
		
		\underline{$\M^{\A,\qRO{H}}(\pk, \{(\messageBeforeHTS_i,\sigma_i)\}_{1 \leq i \leq \qsign})$}
		\begin{nicodemus}
			
			\item $\RO{H'} := \RO{H}; i := 1$
			
			\item $(\messageAfterHTS^*,\signatureAfterHTS^*) = \A^{\oracleSIGN,\qRO{H'}}(\pk)$ 
			
			\item Parse $\signatureAfterHTS^*$ as $(z^*,\sigma^*)$
			
			\item \pcreturn $(\RO{H}(z^*\|\messageAfterHTS^*), \sigma)$
			
		\end{nicodemus}
		
	\end{minipage}
	
	\;
	
	\begin{minipage}[t]{0.4\textwidth}
		
		\underline{$\oracleSIGN(\messageAfterHTS_i)$}
		\begin{nicodemus}
			\item $z_i \uni Z$

			\item \pcif $\exists \hat{\messageBeforeHTS_i}$ s. th. $(z_i \| \messageAfterHTS_i, \hat{\messageBeforeHTS_i}) \in \List{\RO{H'}}$
			
				\item \quad $\List{\RO{H'}} := \List{\RO{H'}} \setminus \lbrace  (z_i \| \messageAfterHTS_i, \hat{\messageBeforeHTS_i}) \rbrace$ \label{line:cma2rma:OverwriteOracleValues}
			
			\item  $\List{\RO{H'}} := \List{\RO{H'}} \cup \lbrace (z_i \| \messageAfterHTS_i, \messageBeforeHTS_i)\rbrace$ \label{line:cma2rma:LazySampling}
			
			\item $i := i+1$
			
			\item \pcreturn $(z_i,\sigma_i)$
			
		\end{nicodemus}
		
		\

		\underline{$\RO{H'}(z \| \messageAfterHTS)$}
		\begin{nicodemus}
			
			\item \pcif $\exists \messageBeforeHTS$ s. th. $(z \| \messageAfterHTS, \messageBeforeHTS) \in \List{\RO{H'}}$
			
				\item \quad \pcreturn $\messageBeforeHTS$
			
			\item \pcelse \pcreturn $\RO{H}(z \| \messageAfterHTS)$
			
		\end{nicodemus}
	
	\end{minipage}

}\end{center}
    \ifnum\tightOnSpace=1 \vspace{-0.4cm} \fi
	\caption{Reduction $\M$ reducing \eufrma to \eufcma.}
	\label{fig:cma2rma}
\vspace{-.2in}
\end{figure}

To conclude our main argument, we will now argue that 

\[
\Pr[G_2^{\A} \Rightarrow 1] = \Adv^{\eufrma}_{\sigS}(\M) \enspace,
\]
where reduction $\M$ is given in \cref{fig:cma2rma}.
Since reprogramming is done a-posteriori in game $G_{2}$, $\M$ can simulate a reprogrammed oracle $\RO{H'}$ via access to its own oracle $\RO{H}$ and an initial table look-up:
$\M$ keeps track of the (classical) values on which $\RO{H'}$ has to be reprogrammed (see line~\ref{line:cma2rma:LazySampling})
and tweaks $\Ad{A}$'s oracle $\RO{H'}$, accordingly.
The latter means that, given the table $\List{\RO{H'}}$ of pairs $(z_i \| \messageAfterHTS_i, \messageBeforeHTS_i)$ that were already defined in previous signing queries,
controlled on the query input being equal to $z_i \| \messageAfterHTS_i$ output $\messageBeforeHTS_i$,
and controlled on the input not being equal to any $z_i \| \messageAfterHTS_i$, forward the query to $\M$'s own oracle $\RO{H}$.
If needed, $\M$ reprograms values (see line~\ref{line:cma2rma:OverwriteOracleValues}) by adding an entry to its look-up table. Given quantum access to $\RO{H}$, \M can implement this as a quantum circuit, allowing quantum access to $\RO{H'}$.
\label{desc:M}

Hence, \M perfectly simulates game $G_2$ towards \A. The only differences are that \M neither samples the $\messageBeforeHTS_i$ itself, nor computes the signatures for them. Both are given to \M by the \eufrma game. However, they follow the same distribution as in game $G_2$. Lastly, whenever \A would win in game $G_2$, \M succeeds in its \eufrma game as it can extract a valid forgery for \sigS on a new message. This is enforced with the condition we added in game $G_1$. 

The final bound of the theorem follows from collecting the bounds above, and it remains to prove the bound on \mmetcr claimed in \cref{eq:Advantage:MeTCR}.
We improve a bound from~\cite{PKC:HulRijSon16}, in which it was shown that for a small constant $c$,\footnote{\label{foo:corrected}This is a corrected bound from \cite{PKC:HulRijSon16}, see discussion in \cref{sec:Presentation_of_reprogramming}.}
\[
\succmmetcr{\RO{H}}{\B,\qsign} \leq \frac{8 \qsign(\qhash+1)^2}{|\MSpaceAfterHTS|} + c\frac{ \qsign \qhash}{\sqrt{|Z|}} \enspace .
\]

Their proof of this bound is explicitly given for the single target step.
It is then argued that the multi-target step can be easily obtained, which was recently confirmed in~\cite{XMSSembed}.
The proof proceeds in two steps. The authors construct a reduction that generates a random function from an instance of an average-case search problem which requires to find a 1 in a boolean function $f$. The function has the property that all preimages of a randomly picked point $\messageBeforeHTS$ in the image correspond to $1$s of $f$. When \A makes its query to \bbox, the reduction picks a random $z$ and programs $\RO{H}^{(z\|\messageAfterHTS) \mapsto \messageBeforeHTS}$. An extended target collision for $(z\| \messageAfterHTS)$ hence is a $1$ in $f$ by design. This gives the first term in the above bound, which is known to be optimal. 

The second term in the bound is the result of above reprogramming. I.e., it is a bound on the difference in success probability of \A when playing the real game or when run by the reduction. More precisely, the bound is the result of analyzing the distinguishing advantage between the following two games (which we rephrased to match our notation):
        
\heading{Game $G_a$.} \A gets access to $\RO{H}$. In phase 1, after making at
  most $q_1$ queries to $\RO{H}$, \A outputs a message $\messageAfterHTS \in \MSpaceAfterHTS$. Then
  a random $z \uni Z$ is sampled and $(z,\RO{H}(z\| \messageAfterHTS))$ is
  handed to $\A$. $\A$ continues to the second phase and makes at most
  $q_2$ queries. $\A$ outputs $b\in \{0,1\}$ at the end.

\heading{Game $G_b$.} \A gets access to $\RO{H}$.  After making at most $q_1$
  queries to $\RO{H}$, \A outputs a message $\messageAfterHTS \in \MSpaceAfterHTS$. Then a random
  $z \uni Z$ is sampled as well as a random range element
  $\messageBeforeHTS \uni \MSpaceBeforeHTS$. Program $\RO{H} := \RO{H}^{(z\|\messageAfterHTS)\mapsto \messageBeforeHTS}$. $\A$ receives $(z, \messageBeforeHTS=\RO{H}(z\| \messageAfterHTS))$ and proceeds to the second
  phase. After making at most $q_2$ queries, $\A$ outputs
  $b\in \{0,1\}$ at the end.

The authors of~\cite{PKC:HulRijSon16} showed that for a small constant $c$ (see \cref{foo:corrected}),%
\[
 | \Pr[G_b^{\A} \Rightarrow 1] -  \Pr[G_a^{\A} \Rightarrow 1] |\leq c\frac{\qhash}{\sqrt{|Z|}}\,.
\]
A straightforward application of \cref{prop:basic} shows that  
\[
 | \Pr[G_b^{\A} \Rightarrow 1] -  \Pr[G_a^{\A} \Rightarrow 1] |\leq \frac{3}{2}  \sqrt{\frac{\qhash+1}{|Z|}}\,.
\]
as the games above virtually describe the games $\gameRepro_{b}$ with the exception that in $\gameRepro_{b}$ the oracle \Reprogram only returns $z$ and not $\RO{H}(z\|\messageAfterHTS))$. Hence, a reduction needs one additional query per reprogramming.

When applying this to the \qsign-target case, a hybrid argument shows that the bound becomes $\nicefrac{3\qsign}{2} \sqrt{\nicefrac{\qhash+1}{|Z|}}$. Combining this with the reduction of~\cite{PKC:HulRijSon16} and taking into account that \B makes (\qsign+\qhash+1) queries confirms 
\ifsubmission
the bound claimed in \cref{eq:Advantage:MeTCR}.
\else
the claimed bound of 
\[
\succmmetcr{\RO{H}}{\B,\qsign} \leq \frac{8 \qsign(\qsign+\qhash+2)^2}{|\MSpaceAfterHTS|} + \frac {3\qsign}{2}  \sqrt{ \frac{\qhash + \qsign + 1}{|Z|}} \enspace .
\]
\fi
\end{proof}

\subsection{Tight security for message hashing of RFC 8391}
Another extremely similar application of our basic bound is for another 
case of the hash-and-sign construction, used to turn a fixed message length \eufcma-secure signature scheme \sigS into a variable input length one \sigSp. This case is essentially covered already by \cref{subse:HTS:RMA}: A proof can omit game $G_2$ and state a simple reduction that simulates game $G_1$ to extract a forgery. The bound changes accordingly, requiring one reprogramming bound less and becoming  
$\succfun^{\eufcma}_{\sigSp}(\A)\leq\succfun^{\eufcma}_{\sigS}(\M)
			+ \nicefrac{8 \qsign(\qsign+\qhash)^2}{|\MSpaceAfterHTS|} 
			+ 1.5\qsign  \sqrt{ \nicefrac{\qhash  + \qsign}{|Z|}}$. 

In~\cite{RFC8391}%
\ifeprint
, the authors
\else
 it was
\fi
 suggested that for stateful hash-based signature schemes%
\ifeprint
, like, e.g.,  XMSS~\cite{RFC8391}%
\else
like XMSS~\cite{RFC8391}%
\fi
,
the multi-target attacks which cause the first occurence of \qsign in the bound could be avoided. This was recently formally proven in~\cite{XMSSembed}. The idea is to exploit the property of hash-based signature schemes that every signature has an index which binds the signature to a one-time public key. Including this index into the hash forces an adversary to also include it in a collision to make it useful for a forgery. Even more, the index is different for every signature and therefore for every target hash. 

Summarizing, the authors of~\cite{XMSSembed} showed that 
there exists a tight standard model proof for the hash-and-sign construction,
as used by XMSS in RFC 8391, if the used hash function is \qsign-target extended target-collision 
resistant with nonce (\nmmetcr, see \cref{app:def:hash}), an extension of \mmetcr that considers the index.
To demonstrate the relevance of this result, the authors analyzed the 
\nmmetcr-security of hash functions under generic attacks, proving a bound for \nmmetcr-security in the QROM in the same way as outlined for \mmetcr above.
So far, this bound was suboptimal, as it included a bound on distinguishing variants of games $G_a$ and $G_b$ above in which $\RO{H}$ takes an additional, externally given index as input (for the modified games see \cref{app:def:hash}).
Hence, the bound was $\succnmmetcr{\RO{H}}{\A,p} \leq \nicefrac{8 (\qsign+\qhash)^2}{|\MSpaceAfterHTS|} + \nicefrac{32\qsign\qhash^2}{|Z|}$.
Due to the translation error, we believe that the second term needs to be updated to $32 \qsign \cdot \alpha$,
where $\alpha = \nicefrac{\qhash}{\sqrt{|Z|}}$, instead of $32  \qsign \cdot \alpha^2$.
In~\cite{XMSSembed}, it was conjectured that in $\alpha$, a factor of $\sqrt{\qhash}$ can be removed.
We can confirm this conjecture. 
As in the case above, \cref{prop:basic} can be directly applied to the distinguishing bound for games $G_a$ and $G_b$.
A reduction would simply treat the index as part of the message sent to $\Reprogram$. Plugging this into the proof in~\cite{XMSSembed} leads to the bound
\[ 
\succnmmetcr{\RO{H}}{\A,p} \leq \frac{8 (\qsign+\qhash)^2}{|\MSpaceAfterHTS|} + 1.5\qsign\sqrt{\frac{\qhash+\qsign}{|Z|}} \enspace .
\]
\section{Applications to the Fiat-Shamir transform}\label{sec:SignaturesInQROM}

For the sake of completeness, we include all used definitions
for identification and signature schemes in \cref{sec:SignaturesInQROM:Defs:Appendix}.
The only non-standard (albeit straightforward) definition is computational $\HVZK$ for multiple transcripts, which we give below.

\heading{(Special) $\HVZK$ simulator.}
We first recall the notion of an $\HVZK$ simulator.
Our definition comes in two flavours:
While a standard $\HVZK$ simulator generates transcripts relative to the public key,
a \textit{special} $\HVZK$ simulator generates transcripts relative to (the public key and) a particular challenge.
\begin{definition}[(Special) $\HVZK$ simulator]	\label{Def:HVZK:Simulator}
	An \emph{$\HVZK$ simulator} is an algorithm $\Sim$ that takes as input the public key $\pk$
	and outputs a transcript $(\commitment, \challenge, \response)$.
	A \emph{special $\HVZK$ simulator} is an algorithm $\Sim$ that takes as input the public key $\pk$ and a challenge $c$
	and outputs a transcript $(\commitment, \challenge, \response)$.
\end{definition}

\heading{Computational HVZK for multiple transcripts.}
In our security proofs, we will have to argue that collections of honestly generated transcripts
are indistinguishable from collections of simulated ones.
Since it is not always clear whether computational $\HVZK$ implies computational $\HVZK$ for \emph{multiple } transcripts,
we extend our definition, accordingly:
In the multi-$\HVZK$ game, the adversary obtains a collection of transcripts (rather than a single one).
Similarly, we extend the definition of \textit{special} computational $\HVZK$ from \cite{EC:AOTZ20}.
\begin{definition}[(Special) computational multi-$\HVZK$]	\label{Def:HVZK:Computational}
	Assume that $\IdScheme$ comes with an $\HVZK$ simulator $\Sim$.
	We define multi-$\HVZK$ games $t\textnormal{-}\HVZK$ as in \cref{fig:Def:HVZK},
	and the  multi-$\HVZK$ \textit{advantage function of an adversary $\Ad{A}$ against $\IdScheme$}
	as
	\[ \Adv^{t\textnormal{-}\HVZK}_{\IdScheme}(\Ad{A})
		:= \left|\Pr [{t\textnormal{-}\HVZK_1^{\Ad{A}}}_\IdScheme \Rightarrow 1 ]
			 - \Pr [{t\textnormal{-}\HVZK_0^{\Ad{A}}}_\IdScheme \Rightarrow 1 ] \right|
	\enspace. \]
	To define \textit{special} multi-$\HVZK$, assume that $\IdScheme$ comes with a special $\HVZK$ simulator $\Sim$.
	We define multi-$\specialHVZK$ games as in \cref{fig:Def:HVZK},
	and the  multi-$\specialHVZK$ \textit{advantage function of an adversary $\Ad{A}$ against $\IdScheme$}
	as
	\[ \Adv^{t\textnormal{-}\specialHVZK}_{\IdScheme}(\Ad{A})
	:= \left|\Pr [{t\textnormal{-}\specialHVZK_1^{\Ad{A}}}_\IdScheme \Rightarrow 1 ]
	- \Pr [{t\textnormal{-}\specialHVZK_0^{\Ad{A}}}_\IdScheme \Rightarrow 1 ] \right|
	\enspace. \]
	
	\begin{figure}[h] \begin{center} \fbox{ \small
				
				\nicoresetlinenr
				
				\begin{minipage}[t]{4.2cm}
					
					\underline{\textbf{GAME} $t$-$\HVZK_b$}
					
					\begin{nicodemus}
						
						\item $(\pk, \sk) \leftarrow \IG(\param)$
						
						\item \pcfor $i \in \lbrace 1, \cdots, t \rbrace$
						
						\item \quad $\transcript_i^0 \leftarrow \getTrans(\sk)$
						
						\item \quad $\transcript_i^1 \leftarrow \Sim(\pk)$
						
						\item $b' \leftarrow \Ad{A} (\pk, (\transcript_i^b)_{1 \leq i  \leq t })$
						
						\item \pcreturn $b'$
						
					\end{nicodemus}		
					
				\end{minipage}
				
				\quad 
				
				\begin{minipage}[t]{4.7cm}
					
					\underline{\textbf{GAME} $t$-$\specialHVZK_b$}
					\begin{nicodemus}
						
						\item $i:=1$
						
						\item $(\pk, \sk) \leftarrow \IG(\param)$
						
						\item $b' \leftarrow \Ad{A}^{\oraclegetTrans} (\pk)$
						
						\item \pcreturn $b'$
						
					\end{nicodemus}		
					
					\ \\
					
					\underline{$\oraclegetTrans(\challenge)$}
					
					\begin{nicodemus}
						
						\item \pcif $i > t$ \pcreturn $\bot$
						
						\item $i := i+1 $
						
						\item $\transcript^0 \leftarrow \getTransChallenge(\sk, \challenge)$
						
						\item $\transcript^1 \leftarrow \Sim(\pk, \challenge)$
						
						\item \pcreturn $\transcript^b$
						
					\end{nicodemus}		
					
				\end{minipage}

		} \end{center}
		\caption{Multi-$\HVZK$ game and multi-$\specialHVZK$ game for $\IdScheme$.
			Both games are defined relative to bit $b \in \bits$,
			and to the number $t$ of transcripts the adversary is given.
		}
		\label{fig:Def:HVZK}
	\end{figure}
	
\end{definition}

\heading{Statistical HVZK.}
Unlike computational $\HVZK$, \textit{statistical} $\HVZK$ can be generalized generically, %
we therefore do not need to deviate from known statistical definitions (included in \cref{sec:SignaturesInQROM:Defs:Appendix}).
We denote the respective upper bound for (special) statistical $\HVZK$
by $\boundStatisticalHVZK$ ($\boundSpecialStatisticalHVZK$).

\subsection{Revisiting the Fiat-Shamir transform}\label{sec:FiatShamirInQROM}

In this section, we show that if an identification scheme $\IdScheme$ is $\HVZK$,
and if $\SignatureScheme := \FS[\IdScheme, \ROChallenge]$ possesses $\UFCMAZero$ security
(also known as $\UFKOA$ security),
then $\SignatureScheme$ is also $\UFCMA$ secure, in the QROM.
Note that our theorem makes no assumptions on how $\eufcma_0$ is proven. For arbitrary ID schemes this can be done using a general reduction for the Fiat-Shamir transform \cite{C:DFMS19}, incurring a $\qhash^2$ multiplicative loss that is, in general, unavoidable \cite{C:DonFehMaj20}. For a \emph{lossy} ID scheme $\IdScheme$, $\eufcma_0$ of  $\FS[\IdScheme, \ROChallenge]$ can be reduced tightly to the extractability of $\IdScheme$ in the QROM \cite{EC:KilLyuSch18}. In addition, while we focus on the standard Fiat-Shamir transform for ease of presentation, the following theorem generalizes to signatures constructed using the multi-round generalization of the Fiat-Shamir transform like, e.g., MQDSS \cite{AC:CHRSS16}.

\begin{theorem}\label{thm:FS_NoMA_to_CMA_tight}
	
	For any (quantum) $\UFCMA$ adversary $\Ad{A}$ issuing at most $\qsign$ (classical) queries to the signing oracle $\oracleSIGN$
	and at most $q_{\ROChallenge}$ quantum queries to $\ROChallenge$,
	there exists a $\UFCMAZero$ adversary $\Ad{B}$
	and a multi-$\HVZK$ adversary $\AdversaryHVZK$
	such that
	\begin{align} \label{eq:bound:FS}
		\succfun^{\UFCMA}_{\FS[\IdScheme, \ROChallenge]}(\Ad{A})
		\leq
			\succfun^{\UFCMAZero}_{\FS[\IdScheme, \ROChallenge]}(\Ad{B})
			+  \Adv^{\qsign-\HVZK}_{\IdScheme}(\AdversaryHVZK) \\ 
			+ \frac {3\qsign} 2  \sqrt{ (\qhash + \qsign + 1) \cdot  \maxEntropyCommit}	\enspace,
	\end{align}
	and the running time of $\Ad{B}$
	and $\AdversaryHVZK$
	is about that of $\Ad{A}$.
	The bound given in \cref{eq:bound:FS} also holds for the modified Fiat-Shamir transform
	that defines challenges by letting $\challenge := \ROChallenge(\commitment, m, \pk)$
	instead of letting $\challenge := \ROChallenge(\commitment, m)$.
\end{theorem}

Note that if $\IdScheme$ is statistically $\HVZK$, we can replace
	$\Adv^{\qsign-\HVZK}_{\IdScheme}(\AdversaryHVZK)$
with $\qsign \cdot \boundStatisticalHVZK$.

\begin{proof}

\newcounter{CounterNoMAtoCMA} %
\setcounter{CounterNoMAtoCMA}{1}

{
	\newcounter{ReprogramNoMAtoCMA}
	\setcounter{ReprogramNoMAtoCMA}{\theCounterNoMAtoCMA}
	\stepcounter{CounterNoMAtoCMA}
}

{
	\newcounter{SwitchToSimulatedTranscriptNoMAtoCMA}
	\setcounter{SwitchToSimulatedTranscriptNoMAtoCMA}{\theCounterNoMAtoCMA}
}

\newcounter{LastGameNoMAtoCMA}
\setcounter{LastGameNoMAtoCMA}{\theCounterNoMAtoCMA}

\newcommand{\gameReprogramNoMAtoCMA}{G_{\theReprogramNoMAtoCMA}\xspace}
\newcommand{\gameReprogramOnSimulatedTranscriptNoMAtoCMA}{G_{\theReprogramOnSimulatedTranscriptNoMAtoCMA}\xspace}
\newcommand{\gameSwitchToSimulatedTranscriptNoMAtoCMA}{G_{\theSwitchToSimulatedTranscriptNoMAtoCMA}\xspace}
\newcommand{\LastGameNoMAtoCMA}{G_{\theLastGameNoMAtoCMA}\xspace} 
Consider the sequence of games given in \cref{fig:games:FS_NoMA_to_CMA_tight}.

\begin{figure}[h] \begin{center} \makebox[\textwidth][c]{ \fbox{ \small
	
	\nicoresetlinenr
		
	\begin{minipage}[t]{4.2cm}
		
		\underline{\textbf{GAMES} $G_{0}$ - $\LastGameNoMAtoCMA$}
		\begin{nicodemus}
			
			\item  $(\pk, \sk) \leftarrow \IG(\param)$
				
			\item $(m^*, \sigma^*) \leftarrow \Ad{A}^{\oracleSIGN, \qRO{\ROChallenge}} (\pk)$
			
			\item \pcif $m^* \in \ListOfMessages$ \pcreturn 0
			
			\item Parse $(\commitment^*, \response^*) := \sigma^*$ 
				
			\item $\challenge^* := \ROChallenge(\commitment^*, m^*)$
				\label{line:FS_NoMA_to_CMA:DefineHashChallenge}

			\item \pcreturn $\VerifyId(\pk, \commitment^*, \challenge^*, \response^*)$
			
		\end{nicodemus}		
		
	\end{minipage}
	
	\quad
	
	\begin{minipage}[t]{5.1cm}

		\underline{$\oracleSIGN(m)$}
		\begin{nicodemus}
			
			\item 
			 $\ListOfMessages := \ListOfMessages \cup \lbrace m \rbrace$
			
			\item $(\commitment, \challenge, \response) \leftarrow \getTrans(m)$
				\gcom{$G_0$-$G_\before{SwitchToSimulatedTranscriptNoMAtoCMA}$}

			\item $(\commitment, \challenge, \response) \leftarrow \Sim(\pk)$
				\label{line:FS_NoMA_to_CMA:SimulatedTranscript}
				\gcom{$\gameSwitchToSimulatedTranscriptNoMAtoCMA$}%

			\item $\ROChallenge := \ROChallenge^{(\commitment, m) \mapsto \challenge}$
				\label{line:FS_NoMA_to_CMA:ReprogramRO}
				\gcom{$\gameReprogramNoMAtoCMA$ -$\LastGameNoMAtoCMA$}

			\item \pcreturn $\sigma := (\commitment, \response)$
			
		\end{nicodemus}		
	
\end{minipage}

\quad

\begin{minipage}[t]{4cm}

		\underline{$\getTrans(m)$} \gcom{$G_0$-$G_\before{SwitchToSimulatedTranscriptNoMAtoCMA}$}
		\begin{nicodemus}
			
			\item 
			
			$(\commitment, \state) \leftarrow \Commit(\sk)$
			
			\item $\challenge := \ROChallenge(\commitment, m)$
			\label{line:FS_NoMA_to_CMA:DefineHashTranscript}
			\gcom{$G_0$}%
			
			\item $\challenge' \uni \ChallengeSpace$
			\label{line:FS_NoMA_to_CMA:RandomChallenge}
			\gcom{$G_\before{SwitchToSimulatedTranscriptNoMAtoCMA}$}
			
			\item $\response \leftarrow \Respond(\sk, \commitment, \challenge, \state)$
			
			\item \pcreturn $(\commitment, \challenge, \response)$
			
		\end{nicodemus}
		
	\end{minipage}

} }\end{center}
\ifnum\tightOnSpace=1 \vspace{-0.4cm} \fi
\caption{Games $G_{0}$ - $\LastGameNoMAtoCMA$ for the proof of \cref{thm:FS_NoMA_to_CMA_tight}.}
\label{fig:games:FS_NoMA_to_CMA_tight}
\end{figure}

\heading{Game $G_0$.}
{
	Since game $G_0$ is the original $\UF$-$\CMA$ game,
	\[ \succfun^{\UF\text{-}\CMA}_{\FS[\IdScheme, \ROChallenge]}(\Ad{A}) = \Pr [ G_0^{\Ad{A}} \Rightarrow 1]\enspace.\]
}

\heading{Game $\gameReprogramNoMAtoCMA$.}
{
	In game $\gameReprogramNoMAtoCMA$, we change the game twofold:
	First, the transcript is now drawn according to the underlying ID scheme, i.e., it is drawn uniformly at random as opposed to letting $c:=\ROChallenge(\commitment, m)$,
	see line~\ref{line:FS_NoMA_to_CMA:RandomChallenge}.
	Second, we reprogram the random oracle $\ROChallenge$ in line~\ref{line:FS_NoMA_to_CMA:ReprogramRO}
	such that it is rendered a-posteriori-consistent with this transcript,
	i.e., we reprogram $\ROChallenge$ such that $\ROChallenge(\commitment, m) = \challenge$.
	
	To upper bound the game distance, we construct a quantum distinguisher $\Ad{D}$
	in \cref{fig:FS_NoMA_to_CMA_tight:DistinguisherO2H}
	that is run in the adaptive reprogramming games $\gameRepro_{R,b}$ with $R:= q_S$ many reprogramming instances.
	We identify reprogramming position $\xFirst$ with $(\commitment, m)$,
	additional input $\xSecond$ with $\state$, and $y$ with $\challenge$.
	Hence, the distribution $p$ consists 
	of the constant distribution that always returns $m$ (as $m$ was already chosen by $\Ad{A}$),
	together with the distribution $\Commit(\sk)$.
	Since $\Ad{D}$ perfectly simulates game $G_b$ if run in its respective game $\gameRepro_{b}$,
	we have
	\[ \gameDist{ReprogramNoMAtoCMA}{A}
		= | \Pr[\gameRepro_{1}^{\Ad{D}} \Rightarrow 1] - \Pr[\gameRepro_{0}^{\Ad{D}} \Rightarrow 1]  |\enspace. \]
	
	Since $\Ad{D}$ issues $q_S$ reprogramming instructions and $(q_H + q_S + 1)$ many queries to $\RO{H}$,
	\cref{prop:advanced} yields
	\begin{equation}
	| \Pr[\gameRepro_{1}^{\Ad{D}} \Rightarrow 1] - \Pr[\gameRepro_{0}^{\Ad{D}} \Rightarrow 1] |
	\leq \frac {3q_S} 2  \sqrt{ (q_H + q_S + 1) \cdot p_{\max}} \enspace , \label{eq:FS_NoMA_to_CMA_tight:O2H}
	\end{equation}
	
	where $p_{\max} = \mathbb{E}_\IG \max_{\commitment}
		\Pr_{\MakeUppercase{\commitment}, \MakeUppercase{\state} \leftarrow \Commit(\sk)}
			[\MakeUppercase{\commitment} = \commitment] = \maxEntropyCommit$.
	
	\begin{figure}[h] \begin{center} \fbox{ \small
		
		\nicoresetlinenr
		
		\begin{minipage}[t]{4.5cm}
			
			\underline{\textbf{Distinguisher} $\Ad{D}^{\qRO{\ROChallenge}}$}
			\begin{nicodemus}
				
				\item $(\pk, \sk) \leftarrow \IG(\param)$
				
				\item $(m^*, \sigma^*) \leftarrow \Ad{A}^{\oracleSIGN, \qRO{\ROChallenge}} (\pk)$
				
				\item \pcif $m^* \in \ListOfMessages$ \pcreturn 0
				
				\item Parse $(\commitment^*, \response^*) := \sigma^*$
				
				\item $\challenge^* := \ROChallenge(\commitment^*, m^*)$
				
				\item \pcreturn $\VerifyId(\pk, \commitment^*, \challenge^*, \response^*)$
				
			\end{nicodemus}		
			
		\end{minipage}
		
		\quad
		
		\begin{minipage}[t]{6.3cm}
			
			\underline{$\oracleSIGN(m)$}
			\begin{nicodemus}
				
				\item  $\ListOfMessages := \ListOfMessages \cup \lbrace m \rbrace$
				
				\item $(\commitment, \state) \leftarrow \Reprogram(m, \Commit(\sk))$
				
				\item $\challenge := \ROChallenge(\commitment, m)$
				
				\item $\response \leftarrow \Respond(\sk, \commitment, \challenge, \state)$
				
				\item \pcreturn $\sigma := (\commitment, \response)$
				
			\end{nicodemus}
			
		\end{minipage}
				
	} \end{center}
	\ifnum\tightOnSpace=1 \vspace{-0.4cm} \fi
	\caption{Reprogramming distinguisher $\Ad{D}$ for the proof of \cref{thm:FS_NoMA_to_CMA_tight}.}
	\label{fig:FS_NoMA_to_CMA_tight:DistinguisherO2H}
	\end{figure}
	
}

\heading{Game $\gameSwitchToSimulatedTranscriptNoMAtoCMA$.}
{
	In game $\gameSwitchToSimulatedTranscriptNoMAtoCMA$, we change the game such that the signing algorithm does not make use of the secret key any more:
	Instead of being defined relative to the honestly generated transcripts,
	signatures are now defined relative to the simulator's transcripts.
	We will now upper bound $\gameDist{SwitchToSimulatedTranscriptNoMAtoCMA}{A}$
	via computational multi-$\HVZK$.
	Consider multi-$\HVZK$ adversary $\AdversaryHVZK$ in \cref{fig:FS_NoMA_to_CMA_tight:DistinguisherHVZK}.
	$\AdversaryHVZK$ takes as input a list of $q_s$ many transcripts,
	which are either all honest transcripts or simulated ones.
	Since reprogramming is done a-posteriori in game $G_{\before{SwitchToSimulatedTranscriptNoMAtoCMA}}$,
	$\AdversaryHVZK$ can simulate it via an initial table look-up, like the reduction $\M$ that was given in \cref{subse:HTS:RMA}
	(see the description on p.~\pageref{desc:M}).
	$\AdversaryHVZK$ perfectly simulates game $G_{\before{SwitchToSimulatedTranscriptNoMAtoCMA}}$
	if run on honest transcripts,
	and game $\gameSwitchToSimulatedTranscriptNoMAtoCMA$
	if run on simulated ones,
	hence
	\[ \gameDist{SwitchToSimulatedTranscriptNoMAtoCMA}{A}
		\leq \Adv^{q_S-\HVZK}_{\IdScheme}(\AdversaryHVZK)
	\enspace .\]

	\ifeprint
	\nicoresetlinenr
		\begin{figure}[h] \begin{center} \makebox[\textwidth][c]{ \fbox{ \small
					
					\begin{minipage}[t]{5cm}
						
						\underline{\textbf{Adversary} $\AdversaryHVZK^{\qRO{H}}(\pk, ((\commitment_i, \challenge_i, \response_i)_{i =1}^{q_s})$}
						\begin{nicodemus}
							
							\item $i:= 0$
							
							\item  $\List{\ROChallenge'} := \emptyset$
							
							\item $(m^*, \sigma^*) \leftarrow \Ad{A}^{\oracleSIGN, \qRO{\ROChallenge'}} (\pk)$
							
							\item \pcif $m^* \in \ListOfMessages$ \pcreturn 0
							
							\item Parse $(\commitment^*, \response^*) := \sigma^*$
							
							\item $\challenge^* := \ROChallenge(\commitment^*, m^*)$
							
							\item \pcreturn $\VerifyId(\pk, \commitment^*, \challenge^*, \response^*)$
							
						\end{nicodemus}		
						
					\end{minipage}
					
					\quad
					
					\begin{minipage}[t]{4.3cm}
						
						\underline{$\oracleSIGN(m)$}
						\begin{nicodemus}
							
							\item $i++$
							
							\item  $\ListOfMessages := \ListOfMessages \cup \lbrace m \rbrace$
							
							\item $(\commitment, \challenge, \response) := (\commitment_i, \challenge_i, \response_i)$
							
							\item \pcif $\exists \challenge'$ s. th. $(\commitment, m, \challenge') \in \List{\ROChallenge'}$
							
							\item \quad $\List{\ROChallenge'} := \List{\ROChallenge'} \setminus \lbrace (\commitment, m, \challenge')\rbrace$\label{line:FS_NoMA_to_CMA:OverwriteOracleValues}
							
							\item  $\List{\ROChallenge'} := \List{\ROChallenge'} \cup \lbrace (\commitment, m, \challenge) \rbrace$\label{line:FS_NoMA_to_CMA:LazySampling}
							
							\item \pcreturn $\sigma := (\commitment, \response)$
							
						\end{nicodemus}
						
					\end{minipage}
					
					\quad
					
					\begin{minipage}[t]{4.3cm}
						
						\underline{$\ROChallenge'(\commitment, m)$}
						\begin{nicodemus}
							
							\item \pcif $\exists \challenge$ s. th. $(\commitment, m, \challenge) \in \List{\ROChallenge'}$
							
							\item \quad \pcreturn $\challenge$
							
							\item \pcelse \pcreturn $\ROChallenge(\commitment, m)$
							
						\end{nicodemus}
						
					\end{minipage}
					
		} }\end{center}
		\ifnum\tightOnSpace=1 \vspace{-0.4cm} \fi
		\caption{$\HVZK$ adversary $\AdversaryHVZK$ for the proof of \cref{thm:FS_NoMA_to_CMA_tight}.}
		\label{fig:FS_NoMA_to_CMA_tight:DistinguisherHVZK}
	\end{figure}
	\else
	\nicoresetlinenr
	\begin{figure}[h] \begin{center} \makebox[\textwidth][c]{ \fbox{ \small
				
		\begin{minipage}[t]{5.9cm}
			
			\underline{\textbf{Adversary} $\AdversaryHVZK^{\qRO{H}}(\pk, ((\commitment_i, \challenge_i, \response_i)_{i \in \lbrace 1, \cdots, q_s \rbrace})$}
			\begin{nicodemus}
				
				\item $i:= 0$
				
				\item  $\List{\ROChallenge'} := \emptyset$
				
				\item $(m^*, \sigma^*) \leftarrow \Ad{A}^{\oracleSIGN, \qRO{\ROChallenge'}} (\pk)$
				
				\item \pcif $m^* \in \ListOfMessages$ \pcreturn 0
				
				\item Parse $(\commitment^*, \response^*) := \sigma^*$
				
				\item $\challenge^* := \ROChallenge(\commitment^*, m^*)$
				
				\item \pcreturn $\VerifyId(\pk, \commitment^*, \challenge^*, \response^*)$
				
			\end{nicodemus}		
			
		\end{minipage}
		
		\quad
		
		\begin{minipage}[t]{4.3cm}
			
			\underline{$\oracleSIGN(m)$}
			\begin{nicodemus}
				
				\item $i++$
				
				\item  $\ListOfMessages := \ListOfMessages \cup \lbrace m \rbrace$
				
				\item $(\commitment, \challenge, \response) := (\commitment_i, \challenge_i, \response_i)$
				
				\item \pcif $\exists \challenge'$ s. th. $(\commitment, m, \challenge') \in \List{\ROChallenge'}$
				
					\item \quad $\List{\ROChallenge'} := \List{\ROChallenge'} \setminus \lbrace (\commitment, m, \challenge')\rbrace$\label{line:FS_NoMA_to_CMA:OverwriteOracleValues}
				
				\item  $\List{\ROChallenge'} := \List{\ROChallenge'} \cup \lbrace (\commitment, m, \challenge) \rbrace$\label{line:FS_NoMA_to_CMA:LazySampling}
				
				\item \pcreturn $\sigma := (\commitment, \response)$
				
			\end{nicodemus}
			
			\end{minipage}
	
		\quad
		
		\begin{minipage}[t]{4.3cm}
			
			\underline{$\ROChallenge'(\commitment, m)$}
			\begin{nicodemus}
				
				\item \pcif $\exists \challenge$ s. th. $(\commitment, m, \challenge) \in \List{\ROChallenge'}$
				
				\item \quad \pcreturn $\challenge$
				
				\item \pcelse \pcreturn $\ROChallenge(\commitment, m)$
				
			\end{nicodemus}
			
		\end{minipage}
				
	} }\end{center}
		\ifnum\tightOnSpace=1 \vspace{-0.4cm} \fi
		\caption{$\HVZK$ adversary $\AdversaryHVZK$ for the proof of \cref{thm:FS_NoMA_to_CMA_tight}.}
		\label{fig:FS_NoMA_to_CMA_tight:DistinguisherHVZK}
	\end{figure}
\fi
	
	It remains to upper bound $\Pr [\gameSwitchToSimulatedTranscriptNoMAtoCMA^{\Ad{A}} \Rightarrow 1]$.
	Consider adversary $\Ad{B}$, given in \cref{fig:adversary:FS_NoMA_to_CMA_tight}.
	$\Ad{B}$ is run in game $\UF\text{-}\CMAZero$ and perfectly simulates game $\gameSwitchToSimulatedTranscriptNoMAtoCMA$ to $\Ad{A}$.
	If $\Ad{A}$ wins in game $\gameSwitchToSimulatedTranscriptNoMAtoCMA$, it cannot have queried $\oracleSIGN$ on $m^*$.
	Therefore, $\ROChallenge'$ is not reprogrammed on $(m^*, \commitment^*)$
	and hence, $\sigma^*$ is a valid signature in $\Ad{B}$'s $\UF\text{-}\CMAZero$ game.
	\begin{equation}
	\Pr [\gameSwitchToSimulatedTranscriptNoMAtoCMA^{\Ad{A}} \Rightarrow 1]
	\leq \succfun^{\UF\text{-}\CMAZero}_{\FS[\IdScheme, \ROChallenge]}(\Ad{B}) \enspace . \nonumber
	\end{equation}	
	
	Collecting the probabilities yields the desired bound.

	\begin{figure}[h] \begin{center} \makebox[\textwidth][c]{ \fbox{ \small
		
		\nicoresetlinenr
		
		\begin{minipage}[t]{4.3cm}
			
			\underline{\textbf{Adversary} $\Ad{B}^\qRO{\ROChallenge}(\pk)$}
			\begin{nicodemus}
				
				\item  $\List{\ROChallenge'} := \emptyset$
				
				\item $(m^*, \sigma^*) \leftarrow \Ad{A}^{\oracleSIGN, \qRO{\ROChallenge'}} (\pk)$
				
				\item \pcif $m^* \in \ListOfMessages$ ABORT
				
				\item \pcreturn $(m^*, \sigma^*)$
				
			\end{nicodemus}

		\end{minipage}
		
		\quad
				
		\begin{minipage}[t]{4.4cm}
			
			\underline{$\oracleSIGN(m)$}
			\begin{nicodemus}
				
				\item  $\ListOfMessages := \ListOfMessages \cup \lbrace m \rbrace$
				
				\item $(\commitment, \challenge, \response) \leftarrow \Sim(\pk)$
				
				\item \pcif $\exists \challenge'$ s. th. $(\commitment, m, \challenge') \in \List{\ROChallenge'}$
				
					\item \quad $\List{\ROChallenge'} := \List{\ROChallenge'} \setminus \lbrace (\commitment, m, \challenge')\rbrace$
					
				\item  $\List{\ROChallenge'} := \List{\ROChallenge'} \cup \lbrace (\commitment, m, \challenge) \rbrace$
						
				\item \pcreturn  $\sigma := (\commitment, \response)$
				
			\end{nicodemus}
		
		\end{minipage}
	
	\quad
	
	\begin{minipage}[t]{4.1cm}
		
		\underline{$\ROChallenge'(\commitment, m)$}
		\begin{nicodemus}
			
			\item \pcif $\exists \challenge$ s. th. $(\commitment, m, \challenge) \in \List{\ROChallenge'}$
			
			\item \quad \pcreturn $\challenge$
			
			\item \pcelse 
			
				\item \quad \pcreturn $\ROChallenge(\commitment, m)$
			
		\end{nicodemus}
	\end{minipage}
		}} \end{center}
		\ifnum\tightOnSpace=1 \vspace{-0.4cm} \fi
		\caption{Adversary $\Ad{B}$ for the proof of \cref{thm:FS_NoMA_to_CMA_tight}.}
		\label{fig:adversary:FS_NoMA_to_CMA_tight}
	\end{figure}
}

It remains to show that the bound also holds if challenges are derived by letting $\challenge := \ROChallenge(\commitment, m, \pk)$.
To that end, we revisit the sequence of games given in \cref{fig:games:FS_NoMA_to_CMA_tight}:
We replace $\challenge := \ROChallenge(\commitment, m)$
(and $\challenge^* := \ROChallenge(\commitment^*, m^*)$) 
with
$\challenge := \ROChallenge(\commitment, m, \pk)$
(and $\challenge^* := \ROChallenge(\commitment^*, m^*, \pk)$)
in line~\ref{line:FS_NoMA_to_CMA:DefineHashTranscript}
(line~\ref{line:FS_NoMA_to_CMA:DefineHashChallenge}),
and change the reprogram instruction in line~\ref{line:FS_NoMA_to_CMA:ReprogramRO},
accordingly.
Since $\pk$ is public, we can easily adapt both distinguisher $\Ad{D}$
and adversaries $\Ad{B}$ and $\Ad{C}$ to account for these changes.
In particular, $\Ad{D}$ will simply include $\pk$ as a (fixed) part of the probability distribution that is forwarded to its reprogramming oracle.
Since the public key holds no entropy once that it is fixed by the game,
this change does not 
affect the upper bound given in \cref{eq:FS_NoMA_to_CMA_tight:O2H}.

\end{proof}

\subsection{Revisiting the hedged Fiat-Shamir transform}\label{sec:HedgedFiatShamir}

In this section, we show how \cref{theorem:reprGameBased}
can be used to extend the results of \cite{EC:AOTZ20} to the quantum random oracle model:
We show that the Fiat-Shamir transform is robust against several types of one-bit fault injections,
even in the quantum random oracle model,
and that the hedged Fiat-Shamir transform is as robust, even if an attacker is in control of the nonce that is used to generate the signing randomness.
In this section, we follow \cite{EC:AOTZ20} and consider the modified Fiat-Shamir transform that includes the public key into the hash when generating challenges.
We consider the following one-bit tampering functions:
\begin{itemize}
	\item[] $\mathsf{flip}$-$\mathsf{bit}_i(x)$: Does a logical negation of the $i$-th bit of $x$.
	\item[] $\mathsf{set}$-$\mathsf{bit}_i(x,b)$: Sets the $i$-th bit of $x$ to $b$.
\end{itemize}

{\heading{Hedged signature schemes.}
	Let $\NonceSpace$ be any nonce space.
	\ifeprint Given \else With \fi a signature scheme $\SignatureScheme = (\KG, \Sign, \VerifySig)$
	with secret key space $\SKSpace$ and signing randomness space $\RSpaceSign$,
	and random oracle $\ROforHedging: \SKSpace \times \MSpace \times \NonceSpace \rightarrow \RSpaceSign$,
	we \ifeprint define\else associate\fi
	\[ \TrafoHedging [\SignatureScheme, \ROforHedging]
	:= \SigSchemeHedged := (\KG, \SigningHedged, \VerifySig) \enspace ,\]
	where the signing algorithm $\SigningHedged$ of $\SigSchemeHedged$ takes as input $(\sk, m, \nonce)$,
	deterministically computes $r := \ROforHedging(\sk, m, \nonce)$,
	and returns $\sigma := \Sign(\sk, m; r)$.
}

{\heading{Security of (hedged) Fiat-Shamir against fault injections and nonce attacks.}
Next, we define \underline{U}n\underline{F}orgeability
in the presence of \underline{F}aults,
under \underline{C}hosen \underline{M}essage \underline{A}ttacks ($\UFfaultCMADifferentSet{}$),
for Fiat-Shamir transformed schemes.
In game $\UFfaultCMADifferentSet{}$,
the adversary has access to a faulty signing oracle $\oracleFaultSIGN$ which returns signatures that were created relative to an injected fault.
To be more precise, game $\UFfaultCMA$
is defined relative to a set $\SetFaultFunctions$ of indices,
and the indices $i \in \SetFaultFunctions$ specify at which point during the signing procedure exactly
the faults are allowed to occur.
An overview is given in \cref{fig:Overview:FaultIndices}.

\tikzset{%
	algo/.style				= {draw, thick, rectangle, minimum height = 3em, minimum width = 2em},
	hash/.style				= {draw, thick, trapezium, shape border rotate = 270, trapezium angle=70, minimum height = 3em},
	fault/.style    		= {draw, circle}, 
	faultTwoWires/.style	= {draw, dashed, ellipse, minimum height= 4em},
	input/.style   			= {}, %
	output/.style			= {} %
}

\tikzstyle{my above of} = [above = 0.8cm of #1.north]
\tikzstyle{my below of} = [below = 0.8cm of #1.south]
\tikzstyle{my right of} = [right = of #1.east]
\tikzstyle{my left of} 	= [left = of #1.west]

\begin{figure}[h] \begin{center} \fbox{ \scalebox{0.9}{

\begin{tikzpicture}[auto, node distance= 0.4cm]

\draw

node [input, name=inputSK] {$\sk$} 										

node [fault,	right = 0.4cm of inputSK] (fault1) {1}												 		%
node [fault,	below = 0.5 cm of fault1, draw = gray] (fault0) {\textcolor{gray!60}{0}}				 	%
node [my above of = fault1] (aboveFault1) {}															 	%

node [input,	left = 0.4cm of fault0] (inputNonce) {$\nonce$}
node [input,	my below of = inputNonce] (inputM) {$m$}

node [hash, draw = none, right = 0.3cm of fault1, text opacity=0] (oracleG1) {$\ROforHedging$}				%
node [hash, draw = none, right = 0.3cm of fault0, text opacity=0] (oracleG0) {$\ROforHedging$}				%
([ yshift = -1.7em] oracleG1 |- fault1) node [hash] (oracleG) {$\ROforHedging$}						 		%

node [fault,	right = 0.3cm of oracleG1, draw = gray] (fault2) {\textcolor{gray!60}{2}}				 	%
node [my above of = fault2] (aboveFault2) {}						 										%
node [fault,	below = 0.3cm of fault2, draw = gray] (fault3) {\textcolor{gray!60}{3}}						%

node [algo, draw = none, right = 0.3cm of fault2, text opacity=0] (commit2) {$\Commit$}						%
node [algo, draw = none, right = 0.3cm of fault3, text opacity=0] (commit3) {$\Commit$}						%
([ yshift = -1.7em] commit2 |- fault2) node [algo, minimum height = 4.5em] (commit){$\Commit$}		 		%

node [faultTwoWires, draw = none, my right of = commit2, text opacity=0] (fault4Upper) {4}					%
node [faultTwoWires, draw = none, my right of = commit3, text opacity=0, yshift= -0.3cm] (fault4Lower) {4}	%
([ yshift = -1.7em] fault4Upper |- fault1) node [faultTwoWires] (fault4) {4}							 	%

node [input,	below = 0.4cm of fault4] (inputPK) {$\pk$}													%

node [fault,	my right of = fault4Lower] (fault5) {5}													 	%

node [hash, right = 0.7 cm of fault5] (oracleH) {$\ROChallenge$}											%

node [fault, my right of = oracleH] (fault6) {6}															%
([xshift=0.2cm]fault6 |- aboveFault2) node [input] (aboveFault6ForSK){}										%
([xshift=-0.2cm] fault6 |- fault4Upper) node [input] (aboveFault6For4){}									%

node [faultTwoWires, draw = none, my right of = fault6, text opacity=0] (fault7Lower) {7}					%
node [faultTwoWires, draw = none, my above of = fault7Lower, text opacity=0] (fault7Upper) {7}				%
([ yshift = -1.7em] fault7Upper |- fault1) node [faultTwoWires] (fault7) {7}							 	%

node [algo, my right of = fault7] (respond) {$\Respond$}													%

node [fault, below = 0.5cm of respond] (fault9) {9}															%

node [output, below = 0.5cm of fault9] (signature) {$\sigma$} 		
;

\draw				(inputSK)					-- node[at start] {}	node [midway] (skToFault1) {} 			(fault1);
\draw[gray!60]		(aboveFault1.west)			-| node[near end, left] {\textcolor{black}{$\sk$}}				(fault2);
\draw				(skToFault1|- inputSK)		|- 		 														(aboveFault2.west);
\draw[->, gray!60]	(fault2)					-- 																(commit2);

\draw[gray!60]		(inputNonce)				-- node [midway] (nonceToFault0) {} 							(fault0);
\draw[gray!60]		(inputM) 					-| 																(nonceToFault0);

\draw[->]			(fault1)					-- 																(oracleG1);
\draw[->, gray!60]	(fault0) 					-- 																(oracleG0);

\draw[gray!60]		(oracleG.east |- fault3)	-- node [pos=0.6] {\textcolor{black}{$r$}}  						(fault3);
\draw[->, gray!60]	(fault3) 					-- 																(commit3);

\draw	(commit.east |- fault2)					-- node [midway] (commitToFault4) {$\state$}					(fault4.west|- fault2);
\draw	([yshift= -0.3cm]commit.east |- fault3)	-- node [midway, above] {$\commitment$}							([yshift= -0.3cm]fault4.west|- fault3);

\draw				(fault4Lower)				-- node [midway, above] (fault4to5) {$\commitment$} 			(fault5);
\draw				(inputM)					-| 																(fault5);
\draw				(inputPK)					-| 																([xshift = -0.5cm]fault5);

\draw[->]			(fault5)					-- 																(oracleH);
\draw				(oracleH)					-- node [midway, above] {$\challenge$}							(fault6);

\draw				(aboveFault2.west)			-- 		 														(aboveFault6ForSK.west);
\draw				(aboveFault6ForSK.west)		|- node [pos=0.83, above] {$\sk$}								(fault7.west|- fault2);
\draw				(fault4Upper)				-- 		 														(aboveFault6For4.north|-aboveFault6For4.east);
\draw				(aboveFault6For4.north|-aboveFault6For4.east)|- node [pos=0.9, above] {$\state$}			([yshift=-0.54cm]fault7.west|- fault2);
\draw				(fault6)					-- node [midway, above] {$\challenge$}							(fault7.west|- fault6);

\draw[->]			(fault7)					-- 																(respond);

\draw				(respond)					-- 																(fault9);
\draw[->]			(fault9)					-- 																(signature);
\end{tikzpicture}
}}
\end{center}
	\ifnum\tightOnSpace=1 \vspace{-0.4cm} \fi
	\caption{
		Faulting a (hedged) Fiat-Shamir signature. Circles represent faults, and their numbers are the respective fault indices $i \in \SetFaultFunctions$ (following \cite{EC:AOTZ20}, for the formal definition see \cref{fig:Def:UF_Fault_CMA}).
		Greyed out fault wires indicate that the hedged construction can not be proven robust against these faults, in general.
		Dashed fault nodes indicate that the Fiat-Shamir construction is robust against these faults if the scheme is subset-revealing.
	}	
	\label{fig:Overview:FaultIndices}
\end{figure}
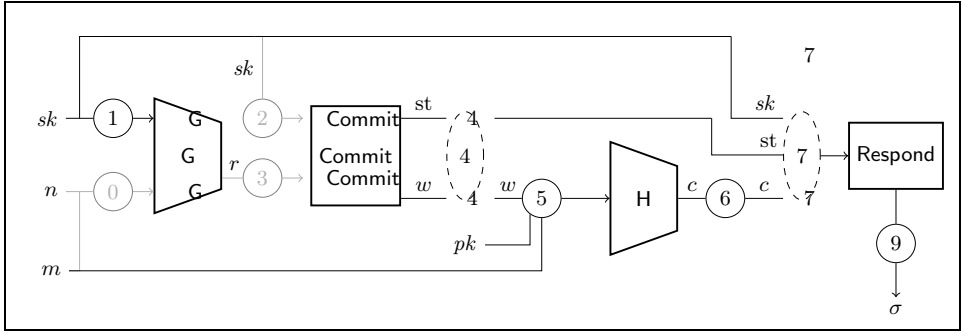 
For the hedged Fiat-Shamir construction, we further define \underline{U}n\underline{F}orgeability,
with control over the used \underline{N}onces and in the presence of \underline{F}aults,
under \underline{C}hosen \underline{M}essage \underline{A}ttacks ($\UFnonceFaultCMADifferentSet{}$).
In game $\UFnonceFaultCMADifferentSet{}$, 
the adversary is even allowed to control the nonce $\nonce$ that is used to derive the internal randomness of algorithm $\Commit$.
We therefore denote the respective oracle by $\oracleNonceFaultSIGN$.
Our definitions slightly simplify the one of \cite{EC:AOTZ20}:
While \cite{EC:AOTZ20} also considered fault attacks on the input of algorithm $\Commit$
(with corresponding indices 2 and 3),
they showed that the hedged construction can not be proven robust against these faults, in general.
We therefore omitted them from our games, but adhered to the numbering for comparability.

The hedged Fiat-Shamir scheme derandomizes the signing procedure by replacing 
the signing randomness by $r := \ROforHedging(\sk, m, \nonce)$.
Hence, game $\UFnonceFaultCMADifferentSet{}$ considers two additional faults:
An attacker can fault the input of $\ROforHedging$,
i.e., either the secret key (fault index 1), or the tuple $(m, \nonce)$ (fault index 0).
As shown in \cite{EC:AOTZ20}, the hedged construction can not be proven robust against faults on $(m, \nonce)$, in general, therefore we
only consider index 1.

Furthemore, we do not formalize derivation/serialisation and drop the corresponding indices 8 and 10
to not overly complicate our application example.
A generalization of our result that also considers derivation/serialisation, however, is straightforward.

\begin{definition}($\UFfaultCMADifferentSet{}$ and $\UFnonceFaultCMADifferentSet{}$)
	For any subset $\SetFaultFunctions \subset \lbrace 4, \cdots, 9 \rbrace$,
	\ifeprint let \else we define\fi the $\UFfaultCMA$ game \ifeprint be defined \fi as in \cref{fig:Def:UF_Fault_CMA},
	and the $\UFfaultCMA$ success probability of a quantum adversary $\Ad{A}$
	against $\FS[\IdScheme, \ROChallenge]$ as
	\[\succfun^{\UFfaultCMA}_{\FS[\IdScheme, \ROChallenge]}(\Ad{A})
		:= \Pr[\UFfaultCMA_{\FS[\IdScheme, \ROChallenge]}^\Ad{A} \Rightarrow 1 ]\enspace .\]
	
	Furthermore, we define the $\UFnonceFaultCMA$ game (also in \cref{fig:Def:UF_Fault_CMA})
	for any subset $\SetFaultFunctions \subset \lbrace 1, 4, \cdots, 9 \rbrace$,
	and the $\UFnonceFaultCMA$ success probability of a quantum adversary $\Ad{A}$
	against $\SigSchemeHedged := \TrafoHedging[\FS[\IdScheme, \ROChallenge], \ROforHedging]$ as
	\[\succfun^{\UFnonceFaultCMA}_{\SigSchemeHedged}(\Ad{A})
	:= \Pr[\UFnonceFaultCMA_{\SigSchemeHedged}^\Ad{A} \Rightarrow 1 ]\enspace .\]
	
	\begin{figure}[h] \begin{center}\makebox[\textwidth][c]{ \fbox{ \small
				
		\begin{minipage}[t]{5.5cm}
			
			\nicoresetlinenr
			\underline{\textbf{Game} \boxedFull{$\UFfaultCMA$} \dashboxed{$\UFnonceFaultCMA$}}
			\begin{nicodemus}
				
				\item $(\pk, \sk) \leftarrow \IG(\param)$
				
				\item \boxedFull{$(m^*, \sigma^*) \leftarrow \Ad{A}^{\oracleFaultSIGN, \qRO{\ROChallenge}} (\pk)$}
				
				\item \dashboxed{ $(m^*, \sigma^*) \leftarrow \Ad{A}^{\oracleNonceFaultSIGN, \qRO{\ROChallenge}, \qRO{\ROforHedging}} (\pk)$}
				
				\item \pcif $m^* \in \ListOfMessages$ \pcreturn 0
				
				\item Parse $(\commitment^*, \response^*) := \sigma^*$
				
				\item $\challenge^* := \ROChallenge(\commitment^*, m^*)$
				
				\item \pcreturn $\VerifyId(\pk, \commitment^*, \challenge^*, \response^*)$
				
			\end{nicodemus}		
			
		\end{minipage}
		
		\;\ifeprint\!\!\fi
		
		\begin{minipage}[t]{4.4cm}
			
			\underline{$\oracleFaultSIGN(m, i \in\SetFaultFunctions, \phi)$}
			
			\begin{nicodemus}
				
				\item $f_i := \phi$ and $f_j := \id  \  \forall \ j \neq i$
				
				\item
				
				\item $(\commitment, \state) \leftarrow \Commit(\sk)$
				
				\item $(\commitment, \state) := f_4(\commitment, \state)$ \label{line:Def:UF_Fault_CMA:Fault4}
				
				\item $(\hat{\commitment}, \hat{m}, \hat{\pk}) := f_5(\commitment, m, \pk)$ \label{line:Def:UF_Fault_CMA:Fault5}
				
				\item $\challenge := f_6(\ROChallenge(\hat{\commitment}, \hat{m}, \hat{\pk}))$ \label{line:Def:UF_Fault_CMA:Fault6}
							
				\item $\response \leftarrow \Respond(f_7(\sk, c, \state))$ \label{line:Def:UF_Fault_CMA:Fault7}

				\item $\ListOfMessages := \ListOfMessages \cup \lbrace \hat{m} \rbrace$
				
				\item \pcreturn $\sigma := f_9(\commitment, \response)$ \label{line:Def:UF_Fault_CMA:Fault9}
				
			\end{nicodemus}
					
		\end{minipage}
	
		\;\ifeprint\!\!\fi
		
		\begin{minipage}[t]{4.4cm}
			
			\underline{$\oracleNonceFaultSIGN(m, \nonce, i \in\SetFaultFunctions, \phi)$}
			
			\begin{nicodemus}
				
				\item $f_i := \phi$ and $f_j := \id \ \forall \ j \neq i$
				
				\item $r := \ROforHedging(f_1(\sk), m, \nonce) $
				
				\item $(\commitment, \state) \leftarrow \Commit(\sk; r)$
				
				\item $(\commitment, \state) := f_4(\commitment, \state)$ \label{line:Def:UF_Nonce_Fault_CMA:Fault4}
				
				\item $(\hat{\commitment}, \hat{m}, \hat{\pk}) := f_5(\commitment, m, \pk)$
				\label{line:Def:UF_Nonce_Fault_CMA:Fault5}
				
				\item $\challenge := f_6(\ROChallenge( \hat{\commitment}, \hat{m}, \hat{\pk}))$
				\label{line:Def:UF_Nonce_Fault_CMA:Fault6}
				
				\item $\response \leftarrow \Respond(f_7(\sk, c, \state))$	
				\label{line:Def:UF_Nonce_Fault_CMA:Fault7}

				\item $\ListOfMessages := \ListOfMessages \cup \lbrace \hat{m} \rbrace$
				
				\item \pcreturn $\sigma := f_9(\commitment, \response)$
				\label{line:Def:UF_Nonce_Fault_CMA:Fault9}
				
			\end{nicodemus}
			
		\end{minipage}
				
	}}
	\end{center}
		\ifnum\tightOnSpace=1 \vspace{-0.4cm} \fi
		\caption{Left: Game $\UFfaultCMA$ for $\SignatureScheme = \FS[\IdScheme, \ROChallenge]$,
			and game $\UFnonceFaultCMA$ for the hedged Fiat-Shamir construction $\SigSchemeHedged := \TrafoHedging[\FS[\IdScheme, \ROChallenge], \ROforHedging]$,
			both defined relative to a set $\SetFaultFunctions$ of allowed fault index positions.
			$\phi$ denotes the fault function,
			which either negates one particular bit of its input,
			sets one particular bit of its input to 0 or 1, or does nothing.
			We implicitly require fault index $i$ to be contained in $\SetFaultFunctions$,
			i.e., we make the convention that both faulty signing oracles return $\bot$
			if $i \notin\SetFaultFunctions$.
		}
		\label{fig:Def:UF_Fault_CMA}
	\end{figure}
	
\end{definition}
}

\heading{From $\UFCMAZero$ to $\UFfaultCMADifferentSet{}$.}
First, we generalize \cite[Lemma 5]{EC:AOTZ20} to the quantum random oracle model.
The proof is given in \cref{sec:HedgedFiatShamir:ProofFCMA}.
\begin{restatable}{theorem}{thmfour}
\label{theorem:UFfaultCMA}
	Assume $\IdScheme$ to be validity aware (see \cref{Def:ValidityAwareness}, \cref{sec:SignaturesInQROM:Defs:Appendix}).
	If $\SignatureScheme := \FS[\IdScheme, \ROChallenge]$ is $\UFCMAZero$ secure,
	then $\SignatureScheme$ is also $\UFfaultCMA$ secure
	for $\SetFaultFunctions := \lbrace 5,6,
	9 \rbrace$,
	in the quantum random oracle model.
	Concretely, for any adversary $\Ad{A}$ against the $\UFfaultCMA$ security of $\SignatureScheme$, 
	issuing at most $q_{S}$ (classical) queries to $\oracleFaultSIGN$
	and $q_\ROChallenge$ (quantum) queries to $\ROChallenge$,
	there exists an $\UFCMAZero$ adversary $\Ad{B}$ and a multi-$\HVZK$ adversary $\AdversaryHVZK$ such that
	\begin{align}	\label{eq:Advantage:UF_F_CMA}
	\succfun^{\UFfaultCMADifferentSet{\lbrace 5,6,9 \rbrace}}_{\SignatureScheme}(\Ad{A})
		\leq \ & \succfun^{\UFCMAZero}_{\SignatureScheme}(\Ad{B})
		+  \Adv^{\qsign-\HVZK}_{\IdScheme}(\AdversaryHVZK)
	\nonumber \\
		& + \frac {3q_{S}} 2  \sqrt{ 2 \cdot (q_H + q_{S} + 1) \cdot \maxEntropyCommit} 
	\enspace .
 	\end{align}
	and $\Ad{B}$ and $\AdversaryHVZK$ have about the running time of $\Ad{A}$.
	
	If we assume that $\IdScheme$ is subset-revealing,
	then $\SignatureScheme$ is even $\UFfaultCMADifferentSet{\SetFaultFunctions'}$ secure
	for $\SetFaultFunctions' :=\SetFaultFunctions \cup \lbrace 4, 7\rbrace$.
	Concretely, the bound of \cref{eq:Advantage:UF_F_CMA} then holds also for
	$\SetFaultFunctions' = \lbrace 4, 5, 6, 7, 9\rbrace$.
\end{restatable}

\heading{From $\UFfaultCMADifferentSet{}$ to $\UFnonceFaultCMADifferentSet{}$.}
Second, we generalize \cite[Lemma 4]{EC:AOTZ20} to the QROM.
The proof is given in \cref{sec:HedgedFiatShamir:ProofNonceFCMA}.
\begin{restatable}{theorem}{thmfive}\label{theorem:UFNonceFaultCMA}
	If $\SignatureScheme := \FS[\IdScheme, \ROChallenge]$ is $\UFfaultCMA$ secure
	for a fault index set $\SetFaultFunctions$,
	then $\SigSchemeHedged := \TrafoHedging [\SignatureScheme, \ROforHedging]$
	is $\UFnonceFaultCMA$ secure for $\SetFaultFunctions' := \SetFaultFunctions \cup \lbrace 1 \rbrace$,
	in the quantum random oracle model,
	against any adversary that issues no query $(m, \nonce)$ to $\oracleNonceFaultSIGN$ more than once. 
	Concretely, for any adversary $\Ad{A}$ against the $\UFnonceFaultCMA$ security of $\SigSchemeHedged$
	for $\SetFaultFunctions'$,
	issuing at most $q_{S}$ queries to $\oracleNonceFaultSIGN$,
	at most $q_\ROChallenge$ queries to $\ROChallenge$,
	and at most $q_\ROforHedging$ queries to $\ROforHedging$,
	there exist $\UFfaultCMA$ adversaries $\Ad{B_1}$ $\Ad{B_2}$ such that
	\begin{align*}
		\succfun^{\UFnonceFaultCMA}_{\SigSchemeHedged}(\Ad{A})
		\leq	\succfun^{\UFfaultCMA}_{\SignatureScheme}(\Ad{B_1})
				+	2q_{\ROforHedging} \cdot \sqrt{
						\succfun^{\UFfaultCMA}_{\SignatureScheme}(\Ad{B_2})
					}
		\enspace ,
	\end{align*}
	and $\Ad{B_1}$ has about the running time of $\Ad{A}$,
	while $\Ad{B_2}$ has a running time of roughly
	$\Time(\Ad{B_2}) \approx \Time(\Ad{A}) + |\sk| \cdot(\Time(\Sign) + \Time(\VerifySig))$,
	where $|\sk|$ denotes the length of $\sk$.
\end{restatable}
With regards to the reduction's advantage, this proof is not as tight as the one in \cite{EC:AOTZ20}:
$\TrafoHedging [\SignatureScheme, \ROforHedging]$
derives the commitment randomness as $r := \ROforHedging(\sk, m, \nonce)$.
During our proof, we need to decouple $r$ from the secret key.
In the ROM, it is straightforward how to turn any adversary noticing this change into an extractor that returns the secret key.
In the QROM, however, all currently known extraction techniques still come with a quadratic loss in the extraction probability.
On the other hand, our reduction is tighter with regards to running time, which we reduce by a factor of $q_{\ROforHedging}$ when compared to \cite{EC:AOTZ20}.
If we hedge with an independent seed $s$ of length $\ell$ (instead of $\sk$), it can be shown with a multi-instance generalization of \cite[Lem.~2.2]{EC:SaiXagYam18} that
\begin{align*}
\succfun^{\UFnonceFaultCMA}_{\SigSchemeHedged}(\Ad{A})
	\leq \succfun^{\UFfaultCMA}_{\SignatureScheme}(\Ad{B})
	+ (\ell + 1 )\cdot (q_{S} + q_{\ROforHedging}) \cdot \sqrt{\nicefrac{1}{2^{\ell-1}}}
\enspace .
\end{align*}
\section{Adaptive reprogramming: proofs}\label{sec:AdaptiveReprogramming}

We will now give the proof for our main \cref{theorem:reprGameBased}, which can be broken down into three steps:
In this section, we consider the simple special case in which only a single reprogramming instance occurs,
and where no additional input $\xSecond$ is provided to the adversary.
The generalisation to multiple reprogramming instances follows from a standard hybrid argument.
The generalisation that considers additional input is also straightforward, as the achieved bounds are information-theoretical and a reduction can hence compute marginal and conditioned distributions on its own.
For the sake of completeness, we include the generalisation steps in \cref{sec:AdaptiveReprogramming:Appendix}.

In this and the following sections, we need quantum theory. We stick to the common notation as introduced in, e.g. \cite{NC00}. Nevertheless we introduce some of the most important basics and notational choices we make. For a vector $\ket\psi\in\mathcal H$ in a complex Euclidean space $\mathcal H$, we denote the standard Euclidean norm by $\|\ket\psi\|$. We use a subscript to indicate that a vector $\ket\psi$ is the state of a quantum register $A$ with Hilbert space $\mathcal H$, i.e. $\ket\psi_A$. Similarly, $M_A$ indicates that a matrix $M$ acting on $\mathcal H$ is considered as acting on register $A$. The joint Hilbert space of multiple registers is given by the tensor product of the single-register Hilbert spaces. Where it helps simplify notation, we take the liberty to reorder registers, keeping track of them using register subscripts.  The only other norm we will require is the trace norm. For a matrix $M$ acting on $\mathcal H$, the trace norm $\|M\|_1$ is defined as the sum of the singular values of  $M$. An important quantum gate is the quantum extension of the classical CNOT gate. This quantum gate is a unitary matrix $\mathrm{CNOT}$ acting on two qubits, i.e. on the vector space $\mathbb C^2\otimes \mathbb C^2$, as $\mathrm{CNOT}\ket{b_1}\ket{b_2}=\ket{b_1}\ket{b_2\oplus b_1}$. We sometimes subscript a CNOT gate with control register $A$ and target register $B$ with $A:B$, and extend this notation to the case where many CNOT gates are applied, i.e. $\mathrm{CNOT}^{\otimes n}_{A:B}$ means a CNOT gate is applied to the $i$-th qubit of the $n$-qubit registers $A$ and $B$ for each $i=1,...,n$ with the qubits in $A$ being the controls and the ones in $B$ the targets.
\subsection{The superposition oracle}
For proving the main result of this section, we will use the (simplest version of the) superposition oracle introduced in \cite{C:Zhandry19}. In the following, we introduce that technique, striving  to keep this explanation accessible even to readers with minimal knowledge about quantum theory.%

Superposition oracles are perfectly correct methods for simulating a quantum-accessible random oracle $\RO O : \{0,1\}^n\to\{0,1\}^m$. Different variants of the superposition oracle have different additional features that make them more useful than the quantum-accessible random oracle itself. We will use the fact that in the superposition oracle formalism, the reprogramming can be directly implemented by replacing a part of the quantum state held by the oracle, instead of using a simulator that sits between the original oracle and the querying algorithm. Notice that for this, we only need the 
simplest version of the superposition oracle from~\cite{C:Zhandry19}.\footnote{Note that this basic superposition oracle does not provide an \emph{efficient} simulation of a quantum-accessible random oracle, which is fine for proving a query lower bound that holds without assumptions about time complexity.} In that basic form, there are only three relatively simple conceptual steps underlying the construction of the superposition oracle, with the third one being key to its usefulness in analyses:\\
\ifnum\tightOnSpace=1
--
\else
\begin{itemize}
	\item 
	\fi
	For each $x\in\{0,1\}^n$, $\RO O(x)$ is a random variable uniformly distributed on $\{0,1\}^m$. This random variable can, of course, be sampled using a \emph{quantum measurement}, more precisely a computational basis measurement of the state 
    \begin{align*}
		\ket{\phi_0}=2^{-m/2}\sum_{y\in\{0,1\}^m}\ket  y.
	\end{align*}
		\ifnum\tightOnSpace=1
	--
	\else
	\item 
	\fi  For a function $o : \{0,1\}^n\to\{0,1\}^m$, we can store the string $o(x)$ in a quantum register $F_x$. In fact, to sample $\RO O(x)$, we can prepare a register $F_x$ in state $\ket{\phi_0}$, perform a computational basis measurement and keep the \emph{collapsed} so-called \emph{post-measurement state}. Outcome $y$ of the measurement corresponds to the projector $\proj{y}$, and a post-measurement state proportional to
	\begin{align*}
		\proj y\ket{\phi_0}=2^{-\frac m 2}\ket y.
	\end{align*}
Now a query with input $\ket x_X\ket \psi_Y$ can be answered using CNOT gates, i.e. we can answer queries with a superposition oracle unitary $O$ acting on input registers $X,Y$ and an oracle register $F=F_{0^m}F_{0^{m-1}1}...F_{1^m}$ such that 
\begin{align*}
	O_{XYF}\proj x_X=\proj x_X\otimes \left(\mathrm{CNOT}^{\otimes m}\right)_{F_x:Y}.
\end{align*}
	\ifnum\tightOnSpace=1
--
\else
\item 
\fi 
    Since the matrices $\proj y_{F_x}$ and $\left(\mathrm{CNOT}^{\otimes m}\right)_{F_x:Y}$ commute, we can delay the measurement that performs the sampling of the random oracle until the end of the runtime of  the querying algorithm. Queries are hence answered using the unitary $O$, but acting on oracle registers $F_x$ that are all initialized in the uniform superposition state $\ket{\phi_0}$, and only after the querying algorithm has finished, the register $F$ is measured to obtain the concrete random function $\RO O$.
\ifnum\tightOnSpace=1
\else
\end{itemize}
\fi

A quantum-accessible oracle for a random function $\RO O:\{0,1\}^n\to \{0,1\}^m$ is thus implemented  as follows:\\
\ifnum\tightOnSpace=1
--
\else
\begin{itemize}
	\item 
	\fi
	Initialize: Prepare the initial state
	\begin{align*}
	\ket{\Phi}_F=\bigotimes_{x\in\{0,1\}^n}\ket{\phi_0}_{F_x}.
	\end{align*}
	\ifnum\tightOnSpace=1
	--
	\else
	\item 
	\fi 
	Oracle: A quantum query on registers $X$ and $Y$ is answered using $O_{XYF}$\\
\ifnum\tightOnSpace=1
--
\else
\item 
\fi Post-processing: Register $F$ is measured to obtain a random function $\RO O$.
\ifnum\tightOnSpace=1
\else
\end{itemize}
\fi
The last step can be (partially) omitted whenever the function $\RO O$ is not needed for evaluation of the success or failure of the algorithm. In the following, the querying algorithm is, e.g. tasked with distinguishing two oracles, a setting where the final sampling measurement can be omitted.%

Note that it is straightforward to implement the operation of reprogramming a random oracle to a fresh random value on a certain input $x$: just discard the contents of register $F_x$ and replace them with a freshly prepared state $\ket{\phi_0}$.
In addition, we need the following lemma
\begin{lemma}[Lemma 2 in \cite{EC:AMRS20}, reformulated]\label{lem:qqueries}
	Let $\ket{\psi_q}_{AF}$ be the joint adversary-oracle state after an adverary has made $q$ queries to the superposition oracle with register $F$. Then this state can be written as
	\begin{align*}
		\ket{\psi_q}_{AF}=\sum_{\substack{S\subset\{0,1\}^n\\|S|\le q}}\ket{\psi_q^{(S)}}_{AF_S}\otimes \left(\ket{\phi_0}^{\otimes (2^n-|S|)}\right)_{F_{S^c}},
	\end{align*}
	where for any set $R=\{x_1,x_2,...,x_{|R|}\}\subset \{0,1\}^n$ we have defined $F_R=F_{x_1}F_{x_2}...F_{x_{|R|}}$ and $\ket{\psi_q^{(S)}}_{AF_{S}}$ are vectors such that $\bra{\phi_0}_{F_x}\ket{\psi_q^{(S)}}_{AF_{S}}=0$ for all $x\in S$.
\end{lemma}

\subsection{Reprogramming once}

We are now ready to study our simple special case.
Suppose a random oracle $\RO O$ is reprogrammed at a single input $x^*\in\{0,1\}^n$, sampled according to some probability distribution $p$,
to a fresh random output  $y^*\leftarrow \{0,1\}^m$. We set $\RO O_0=\RO O$ and  define $\RO O_1$ by $\RO O_1(x^*)=y$ and $\RO O_1(x)=\RO O$ for $x\neq x^*$.
We will show that if $x^*$ has sufficient min-entropy given $\RO O$, such reprogramming is hard to detect.

More formally, consider a two-stage distinguisher $\Ad{D}=(\Ad{D}_0, \Ad{D}_1)$. The first stage $\Ad{D}_0$ has trivial input, makes $q$ quantum queries to $\RO O$ and ouputs a quantum state $\ket{\psi_{int}}$ and a sampling algorithm for a probability distribution $p$ on $\{0,1\}^n$. The second stage $\Ad{D}_1$ gets $x^*\leftarrow p$ and $\ket{\psi_{int}}$ as input, has arbitrary quantum query access to $\RO O_b$ and outputs a bit $b'$ with the goal that $b'=b$. 
We prove the following.
\begin{theorem}\label{thm:repr}  \label{thm:Reprogramming:SingleInstance}
	The success probability for any distinguisher $\Ad{D}$ as defined above is bounded by
	\begin{align*}
	\Pr[b=b']\le\frac 1 2+ \frac 1 2\sqrt{q p^{\Ad{D}}_{\max}}+\frac 1 4 q p^{\Ad{D}}_{\max},
	\end{align*}
	where the probability is taken over $b\leftarrow\{0,1\},(\ket{\psi_{int}}, p)\leftarrow\Ad{D}_0^{\RO O}(1^n)$ and $b'\leftarrow\Ad{D}_1^{\RO O_b}(x^*,\ket{\psi_{int}})$, and $p^{\Ad{D}}_{\max}=\mathbb E_{(\ket{\psi_{int}}, p)\leftarrow\Ad{D}_0^{\RO O_0}(1^n)}\max_{x}p(x)$.
\end{theorem}

\begin{proof}
	We implement $\RO O=\RO O_0$ as a superposition oracle. Without loss of generality\footnote{This can be seen by employing the Stinespring dilation theorem, or by using standard techniques to delay measurement and discard operations until the end of a quantum algorithm.}, we can assume that $\Ad{D}$ proceeds by performing a unitary quantum computation, followed by  a measurement to produce the classical output $p$ and the discarding of a working register $G$. Let $\ket{\gamma}_{RGF}$ be the algorithm-oracle-state after the unitary part of $\Ad{D}_0$ and the measurement have been performed, conditioned on its second output being a fixed probability distribution $p$. $R$ contains $\Ad{D}_0$'s first output%
	. 
	
	Define
	$	\varepsilon_x=1-\big\|\bra{\phi_0}_{F_{x}}\ket{\gamma}_{RGF}\big\|^2,$
a measure of how far the contents of register $F_x$ are from the uniform superposition.
Intuitively, this is the `probability' that the distinguisher knows $\RO{O}(x)$, and should be small in expectation over $x\leftarrow p$. We therefore begin by bounding the distinguishing advantage in terms of this quantity.
	For a fixed $x$,  we can write the density matrix $\rho^{(0)}=\proj{\gamma}$ as 
	\begin{align}
		\rho^{(0)}_{RGF}=&\bra{\phi_0}_{F_x}\rho^{(0)}_{RGF}\ket{\phi_0}_{F_x}\otimes \proj{\phi_0}_{F_x}+\rho^{(0)}_{RGF}\left(\mathds 1- \proj{\phi_0}_{F_x}\right) \nonumber \\
		&+\left(\mathds 1- \proj{\phi_0}_{F_x}\right)\rho^{(0)}_{RGF}\proj{\phi_0}_{F_x}.\label{eq:rho0}
	\end{align}
	The density matrix $\rho^{(1,x)}_{RGF}$ for the algorithm-oracle-state after $\Ad{D}_0$ has finished and the oracle has been reprogrammed at $x$ (i.e. $b=1$) is
	\begin{align}
		&\rho^{(1,x)}_{RGF}=\Tr_{F_{x}}[\rho^{(1,x)}_{RGF}]\otimes \proj{\phi_0}_{F_x} 
		=\bra{\phi_0}_{F_x}\rho^{(0)}_{RGF}\ket{\phi_0}_{F_x}\otimes \proj{\phi_0}_{F_x}\nonumber \\
		&+\Tr_{F_{x}}[(\mathds 1-\proj{\phi_0}_{F_x})\rho^{(0)}_{RGF}]\otimes \proj{\phi_0}_{F_x},\label{eq:rho1}
	\end{align}
	where the second equality is immediate when computing the partial trace in an orthonormal basis containing $\ket{\phi_0}$. 
	
	We analyze the success probability of $\Ad{D}$. In the following, set $x^*=x$. The second stage, $\Ad{D}_1$, has arbitrary query access to the oracle $\RO O_b$. In the superposition oracle framework, that means $\Ad{D}_1$ can apply arbitrary unitary operations on its registers $R$ and $G$, and the oracle unitary $O$ to some sub-register registers $XY$ of $G$ and the oracle register $F$. We bound the success probability by allowing arbitraty operations on $F$, thus reducing the oracle distinguishing task to the task of distinguishing the quantum states $\rho^{(b,x)}_{RF}=\Tr_G\rho^{(b,x)}_{RGF}$ for $b=0,1$, where $\rho^{(0,x)}\coloneqq \rho^{(0)}$. By the bound relating distinguishing advantage and trace distance,
	\begin{align}
		\Pr[b=b'|x^*=x]\le&\frac 1 2+ \frac{1}{4}\big\|\rho^{(0)}_{RF}-\rho^{(1,x)}_{RF}\big\|_1\le\frac 1 2+ \frac{1}{4}\big\|\rho^{(0)}_{RGF}-\rho^{(1,x)}_{RGF}\big\|_1,\label{eq:disting}
	\end{align}
	where the probability is taken over $b\leftarrow\{0,1\},\ket{\psi_{int}}\leftarrow\Ad{D}_0^{\RO O_0}(1^n)$ and $b'\leftarrow\Ad{D}_1^{\RO O_b}(x,\ket{\psi_{int}})$, and we have used that the trace distance is non-increasing under partial trace. Using Equation \eqref{eq:rho0} and \eqref{eq:rho1}, we bound
	\begin{align*}
		&\big\|\rho^{(0)}_{RGF}-\rho^{(1,x)}_{RGF}\big\|_1\\
		&\le \Big\|\rho^{(0)}_{RGF}\left(\mathds 1- \proj{\phi_0}_{F_x}\right)+\left(\mathds 1- \proj{\phi_0}_{F_x}\right)\rho^{(0)}_{RGF}\proj{\phi_0}_{F_x}\\
		&\quad\quad-\Tr_{F_{x}}[(\mathds 1-\proj{\phi_0}_{F_x})\rho^{(0)}_{RGF}]\otimes \proj{\phi_0}_{F_x}\Big\|_1\\
		&\le \Big\|\rho^{(0)}_{RGF}\left(\mathds 1- \proj{\phi_0}_{F_x}\right)\Big\|_1+\Big\|\left(\mathds 1- \proj{\phi_0}_{F_x}\right)\rho^{(0)}_{RGF}\proj{\phi_0}_{F_x}\Big\|_1\\
		&\quad\quad+\Big\|\Tr_{F_{x}}[(\mathds 1-\proj{\phi_0}_{F_x})\rho^{(0)}_{RGF}]\otimes \proj{\phi_0}_{F_x}\Big\|_1,
	\end{align*}
	Where the last line is the triangle inequality. The trace norm of a positive semidefinite matrix is equal to its trace, so the last term can be simplified as 
	\begin{align*}
		&		\Big\|\Tr_{F_{x}}[(\mathds 1-\proj{\phi_0}_{F_x})\rho^{(0)}_{RGF}]\otimes \proj{\phi_0}_{F_x}\Big\|_1\\
		&\quad\quad=\Tr[(\mathds 1-\proj{\phi_0}_{F_x})\proj{\gamma}_{RGF}]=\varepsilon_x.
	\end{align*}
The second term is upper-bounded by the first via Hölder's inequality, which simplifies as
	\begin{align*}
		\Big\|\rho^{(0)}_{RGF}\left(\mathds 1- \proj{\phi_0}_{F_x}\right)\Big\|_1&=\Big\|\proj{\gamma}_{RGF}\left(\mathds 1- \proj{\phi_0}_{F_x}\right)\Big\|_1\\
		=\Big\|\left(\mathds 1- \proj{\phi_0}_{F_x}\right)\ket{\gamma}_{RGF}\Big\|_2&=\sqrt{\varepsilon_x}
	\end{align*}
	where the second equality uses that $\ket\gamma$ is normalized.	In summary we have 
	\begin{equation}\label{eq:trdistbound}
		\big\|\rho^{(0)}_{RGF}-\rho^{(1,x)}_{RGF}\big\|_1\le 2\sqrt{\varepsilon_x}+\varepsilon_x.
	\end{equation}
It remains to bound $\varepsilon_x$ in expectation over $x\leftarrow p$. To this end, we prove
\begin{equation}\label{eq:bound1}
	\mathbb E_{x^*\leftarrow p}\left[\big\|\bra{\phi_0}_{F_{x^*}}\ket{\gamma}_{RGF}\big\|^2\right]\ge 1-q p_{\max},
\end{equation}
where $p_{\max}=\max_{x}p(x)$. In the following, sums over $S$ are taken over $S\subset\{0,1\}^n:|S|\le q$, with additional restrictions explicitly mentioned.
We have
\begin{align*}
	&\mathbb E_{x^*\leftarrow p}\left[\big\|\bra{\phi_0}_{F_{x^*}}\ket{\gamma}_{RGF}\big\|^2\right]=\sum_{x^*\in\{0,1\}^n}p(x^*)\big\|\bra{\phi_0}_{F_{x^*}}\ket{\gamma}_{RGF}\big\|^2\\
	&\quad\quad=\sum_{x^*\in\{0,1\}^n}p(x^*)\big\|\sum_S \bra{\phi_0}_{F_{x^*}}\ket{\psi_q^{(S)}}_{RGF_S}\otimes \left(\ket{\phi_0}^{\otimes (2^n-|S|)}\right)_{F_{S^c}}\big\|^2,
\end{align*}
  where we have used Lemma \ref{lem:qqueries} as well as the notation $\ket{\psi_q^{(S)}}$ from there. (Lemma \ref{lem:qqueries} clearly also holds after the projector corresponding to second output equaling $p$ is applied). Using $\bra{\phi_0}_{F_x}\ket{\psi_q^{(S)}}_{RGF_S}=0$ for all $x\in S$ we simplify
\begin{align*}
	&\sum_{x^*\in\{0,1\}^n}p(x^*)\big\|\sum_S \bra{\phi_0}_{F_{x^*}}\ket{\psi_q^{(S)}}_{RGF_S}\otimes \left(\ket{\phi_0}^{\otimes (2^n-|S|)}\right)_{F_{S^c}}\big\|^2\\
	&\quad\quad=\sum_{x^*\in\{0,1\}^n}p(x^*)\big\|\sum_{S\not\ni x^*}\ket{\psi_q^{(S)}}_{RGF_S}\otimes \left(\ket{\phi_0}^{\otimes (2^n-|S|-1)}\right)_{F_{S^c\setminus \{x^*\}}}\big\|^2.
\end{align*}

The summands in the second sum are pairwise orthogonal, so
\begin{align*}
	&\sum_{x^*\in\{0,1\}^n}p(x^*)\big\|\sum_{S\not\ni x^*}\ket{\psi_q^{(S)}}_{RGF_S}\otimes \left(\ket{\phi_0}^{\otimes (2^n-|S|-1)}\right)_{F_{S^c\setminus \{x^*\}}}\big\|^2\\
	&\quad\quad=\sum_{x^*\in\{0,1\}^n}p(x^*)\sum_{S\not\ni x^*}\big\|\ket{\psi_q^{(S)}}_{RGF_S}\otimes \left(\ket{\phi_0}^{\otimes (2^n-|S|-1)}\right)_{F_{S^c\setminus \{x^*\}}}\big\|^2\\
	&\quad\quad=\sum_S \sum_{x^*\in S^c}p(x^*)\big\|\ket{\psi_q^{(S)}}_{RGF_S}\otimes \left(\ket{\phi_0}^{\otimes (2^n-|S|-1)}\right)_{F_{S^c\setminus \{x^*\}}}\big\|^2\\
	&\quad\quad=\sum_S \sum_{x^*\in S^c}p(x^*)\big\|\ket{\psi_q^{(S)}}_{RGF_S}\otimes \left(\ket{\phi_0}^{\otimes (2^n-|S|)}\right)_{F_{S^c}}\big\|^2
\end{align*}
where we have used the fact that the state $\ket{\phi_0}$ is normalized in the last line. %
But for any $S\subset\{0,1\}^n$ we have 
\begin{align*}
	\sum_{x^*\in S^c}p(x^*)=&1-\sum_{x^*\in S}p(x^*)\ge 1-|S|p_{\max},
\end{align*}
where here, $p_{\max}=\max_{x}p(x)$.
We hence obtain
\begin{align*}
	&\sum_S \sum_{x^*\in S^c}p(x^*)\big\|\ket{\psi_q^{(S)}}_{RGF_S}\otimes \left(\ket{\phi_0}^{\otimes (2^n-|S|)}\right)_{F_{S^c}}\big\|^2\\
	&\quad\quad\ge\sum_S (1-|S|p_{\max})\big\|\ket{\psi_q^{(S)}}_{RGF_S}\otimes \left(\ket{\phi_0}^{\otimes (2^n-|S|)}\right)_{F_{S^c}}\big\|^2\\
	&\quad\quad\ge(1-q p_{\max})\sum_S \big\|\ket{\psi_q^{(S)}}_{RGF_S}\otimes \left(\ket{\phi_0}^{\otimes (2^n-|S|)}\right)_{F_{S^c}}\big\|^2=1-q p_{\max},
\end{align*}
where we have used the normalization of  $\ket{\gamma}_{RGF}$ in the last equality. Combining the above equations proves Equation \eqref{eq:bound1}.  Putting everything together, we bound
\begin{align*}
	\Pr[b=b']=&\mathbb E_{p}\mathbb E_{x}\Pr[b=b'|p,x]\le\frac 1 2+\frac 1 4\mathbb E_{p}\mathbb E_{x}[2\sqrt{\varepsilon_x}+\varepsilon_x]\\
	\le&\frac 1 2+\frac 1 4\mathbb E_{p}[2\sqrt{q p_{\max}}+q p_{\max}]\le\frac 1 2+\frac 1 2\sqrt{q p^{\Ad{D}}_{\max}}+q p^{\Ad{D}}_{\max}.
\end{align*}
Here, the inequalities are due to Equation \eqref{eq:disting} and Equation \eqref{eq:trdistbound}, Equation \eqref{eq:bound1} and Jensen's inequality, and another Jensen's inequality, respectively. 
\qed	
\end{proof}
\section{A matching attack}\label{sec:Attack}
We now describe an attack matching the bound presented in \Cref{thm:repr}. 
For simplicity, we restrict our attention to the case where just one point is (potentially) reprogrammed.

Our distinguisher makes $q$ queries to $\RO{O}$, the oracle before the potential reprogramming, and $q$ queries to $\RO{O}'$, the oracle after the potential reprogramming. In our attack, we fix  an arbitrary cyclic permutation $\sigma$ on $[2^n]$, and for the fixed reprogrammed point $x^*$, we define \ifeprint the set \fi $S = \{x^*, \sigma^{-1}(x^*),...,\sigma^{-q+1}(x^*)\}$, $\overline{S} = \{0,1\}^n \setminus S$,
$\Pi_0 = \frac{1}{2}\left( \ket{S} + \ket{\overline{S}}\right)\left( \bra{S} + \bra{\overline{S}}\right)$ and $\Pi_1 = I - \Pi_0$.\footnote{Formally, $S$, $\Pi_0$ and $\Pi_1$ are functions of $x^*$ but we omit this dependence for simplicity, since we can assume that $x^*$ is fixed.}
The distinguisher $\Ad{D}$ is defined in \cref{fig:attack}.

\begin{figure}[h] \begin{center} \makebox[\textwidth][c]{\fbox{
			
	\nicoresetlinenr	
	
	\begin{minipage}[t]{7cm}
		
		\underline{Before potential reprogramming:}
		\begin{nicodemus}
		    \item Prepare registers $XY$ in $\frac{1}{\sqrt{2^n}}\sum_{x \in  [2^n]}\ket{x, 0}_{XY}$
		    \item Query $\RO{O}$ using registers $XY$
		    \item \pcfor $i = 0,...,q-1$:
			    \item  \quad Apply $\sigma$ on register $X$
			    \item  \quad Query $\RO{O}$ using registers $XY$
		\end{nicodemus}
		
	\end{minipage}
	\begin{minipage}[t]{6cm}
		
		\underline{After potential reprogramming:}
		\begin{nicodemus}
		    \item Query ${\RO{O}'}$ using using registers $XY$
		    \item \pcfor $i = q-2,...,0$:
			    \item \quad Apply $\sigma^{-1}$ on register $X$
			    \item \quad Query ${\RO{O}'}$ using registers $XY$
		    \item Measure $X$ according to $\{\Pi_0, \Pi_1\}$
		    \item Output $b$ if the state projects onto $\Pi_b$.
		\end{nicodemus}
		
	\end{minipage}
}}\end{center}
	\ifnum\tightOnSpace=1 \vspace{-0.4cm} \fi
	\caption{Distinguisher for a single reprogrammed point.}
	\label{fig:attack}
\end{figure}

\begin{restatable}{theorem}{lemmatwo}\label{lem:basic-attack}
	 For every $1 \leq q  < 2^{n-3}$, the attack described in \Cref{fig:attack} can be implemented in quantum polynomial-time.
	 Performing $q$ queries each before and after the potential reprogramming,
	 it detects the reprogramming of a random oracle $\RO O : \{0,1\}^n\to\{0,1\}^m$ at a single point
	 with probability at least $\Omega(\sqrt{\frac{q}{2^n}})$.
\end{restatable}
\begin{proof}[sketch]
  We can analyze the state of the distinguisher before its measurement. If the oracle is not reprogrammed, then its state is
  \[
  	\frac{1}{\sqrt{2^n}}\sum_{x}\ket{x}\ket{0},
  \]
  whereas if the reprogramming happens, its state is 
  \[
  	\sum_{x \in S} \ket{x}\ket{\RO{O}(x^*) \oplus \RO{O}'(x^*)} + \sum_{x \in \overline{S}} \ket{x}\ket{0},
  \]
  where $\RO{O}(x^*) \oplus \RO{O}'(x^*)$ is a uniformly random value.
  The advantage follows by calculating the probability that these states project onto $\Pi_0$.

  For the efficiency of our distinguisher, we can use the tools provided in \cite{EC:AlaMajRus20} to efficiently implement $\Pi_0$ and $\Pi_1$,
  which are the only non-trivial operations of the attack.
\end{proof}
Due to space restrictions, we refer to \Cref{ap:attack}, where we give the full proof of \Cref{lem:basic-attack} and discuss its extension to multiple reprogrammed points.

\bibliographystyle{alpha}

\ifnum\tightOnSpace=1
	\bibliography{bib,abbrev4,cryptobib/crypto}
\else
	\bibliography{bib,cryptobib/abbrev0,cryptobib/crypto}
\fi

\ifnum\supplementaryMaterial=1
\newpage
\bigskip\noindent{\Huge\textbf{Supplementary material}}\vspace{1cm}
\fi

\appendix

\section{Adaptive reprogramming: Omitted proofs}\label{sec:AdaptiveReprogramming:Appendix}

We now discuss how to derive our main theorem, \cref{theorem:reprGameBased}, from the simple case proven in \cref{thm:Reprogramming:SingleInstance}. For easier reference we repeat the theorem statement.

\thmone*

First, we extend \cref{thm:Reprogramming:SingleInstance} to multiple reprogramming instances with a hybrid argument. Afterwards, we extend the result to cover side information. 
To prove the first step, we introduce helper games $G_b$ in \cref{fig:Def:Games:Reprogramming:MultipleInstance},
in which the adversary has access to oracle $\Reprogram'$.
(These are already almost the same as the $\gameRepro$ games used in \cref{theorem:reprGameBased}. The only difference is that they do not sample and return the additional side information $\xSecond$.) 
\begin{lemma}\label{cor:Reprogramming:MultipleInstance}
	Let $\Ad{D}$ be any distinguisher, issuing $\numberReprogrammingInstances$ many reprogramming instructions.
	Let $\hat q^{(\indexReprogrammingInstances)}$ denote the total number of $\Ad{D}$'s queries to $\RO{O}$
	until the $\indexReprogrammingInstances$-th query to $\Reprogram'$.
	Furthermore, let $p^{(\indexReprogrammingInstances)}$ denote the $\indexReprogrammingInstances$-th distribution on $\XFirst$ on which $\Reprogram'$ is queried,
	and let
	\begin{equation*}
	p^{(\indexReprogrammingInstances)}_{\max} := \Exp \left[\max_{x} p^{(\indexReprogrammingInstances)}(x)\right] \enspace ,
	\end{equation*}
	where the expectation is taken over $\Ad{D}$'s behaviour until its $\indexReprogrammingInstances$-th query to $\Reprogram'$.
	
	The success probability for any distinguisher $\Ad{D}$ is bounded by
	\begin{align*}
	|\Pr[G_{0}^{\Ad{D}} \Rightarrow 1] - \Pr[G_{1}^{\Ad{D}} \Rightarrow 1]|
	\leq &	\sum_{\indexReprogrammingInstances=1}^\numberReprogrammingInstances
		\left(
		\sqrt{ \hat q^{(\indexReprogrammingInstances)} p^{(\indexReprogrammingInstances)}_{\max} }
		+ \frac{1}{2} \hat q^{(\indexReprogrammingInstances)} p^{(\indexReprogrammingInstances)}_{\max}
		\right) 
	\enspace .
	\end{align*}

	\begin{figure}[H] \begin{center} \fbox{\small
				
		\nicoresetlinenr	
		
		\begin{minipage}[t]{3.8cm}
			
			\underline{\textbf{GAMES} $G_{b}$}
			\begin{nicodemus}
				
				\item $\RO{O}_0 \uni Y^X$
				
				\item $\RO{O}_1 := \RO{O}_0$ 
				
				\item $b' \leftarrow \Ad{D}^{\qRO{O_b}, \Reprogram'}$
				
				\item \pcreturn $b'$
				
			\end{nicodemus}
			
		\end{minipage}
		
		\quad
				
		\begin{minipage}[t]{2.3cm}
			
			\underline{$\Reprogram'(p)$}
			\begin{nicodemus}
				
				\item $x \leftarrow p$
				
				\item $y \uni Y$
				
				\item $\RO{O}_1 := \RO{O}_1^{x \mapsto y}$%
				
				\item \pcreturn $x$
				
			\end{nicodemus}
			
		\end{minipage}
				
	}\end{center}
		\caption{Games $G_b$ of \cref{cor:Reprogramming:MultipleInstance}.}
		\label{fig:Def:Games:Reprogramming:MultipleInstance}
	\end{figure}

\end{lemma}
\begin{proof}
	We define hybrid settings $H_\indexReprogrammingInstances$ for $\indexReprogrammingInstances = 0,...,\numberReprogrammingInstances$,
	in which $\Ad{D}$ has access to oracle $\RO{O}$ which is not reprogrammed at the first $\indexReprogrammingInstances$ many positions,
	but is reprogrammed from the $(\indexReprogrammingInstances+1)$-th position on.
	Hence, $H_0$ is the distinguishing game $G_1$, and $H_\numberReprogrammingInstances$ is $G_0$.
	Any distinguisher $\Ad{D}$ succeeds with advantage
	\begin{align*}
		|\Pr[G_{0}^{\Ad{D}} \Rightarrow 1] - \Pr[G_{1}^{\Ad{D}} \Rightarrow 1]|
		& = \Pr [H_{0}^{\Ad{D}} \Rightarrow 1 ] - \Pr [ H_\numberReprogrammingInstances^{\Ad{D}} \Rightarrow 1  ] |
		\\
		& =  \left|
				\sum_{\indexReprogrammingInstances=1}^\numberReprogrammingInstances
					\left( \Pr [ H_{\indexReprogrammingInstances-1}^{\Ad{D}} \Rightarrow 1 ]
					- \Pr [ H_{\indexReprogrammingInstances}^{\Ad{D}} \Rightarrow 1 ]\right)
			\right|
		\\
		& \leq  \sum_{\indexReprogrammingInstances=1}^\numberReprogrammingInstances
			\left|\Pr [ H_{\indexReprogrammingInstances-1}^{\Ad{D}} \Rightarrow 1 ]
			- \Pr [ H_{\indexReprogrammingInstances}^{\Ad{D}} \Rightarrow 1 ] \right|
	\enspace ,
	\label{eq:AdvantageDecomposition}
	\end{align*}
	where we have used the triangle inequality in the last line.
	
	To upper bound $|\Pr [ H_{\indexReprogrammingInstances-1}^{\Ad{D}} \Rightarrow 1 ] - \Pr [ H_{\indexReprogrammingInstances}^{\Ad{D}} \Rightarrow 1 ] |$,
	we will now define distinguishers
	$\hat{\Ad{D}}_\indexReprogrammingInstances = (\hat{\Ad{D}}_{\indexReprogrammingInstances, 0}, \hat{\Ad{D}}_{\indexReprogrammingInstances, 1})$
	that are run in the single-instance distinguishing games $G'_b$ of \cref{thm:Reprogramming:SingleInstance}:
	Let $\RO{O}'$ denote the oracle that is provided by $G'_b$.
 	Until right before the $\indexReprogrammingInstances$-th query to $\Reprogram'$,
 	the first stage $\hat{\Ad{D}}_{\indexReprogrammingInstances, 0}$
 	uses $\RO{O}'$ to simulate the hybrid setting $H_{\indexReprogrammingInstances-1}$ to $\Ad{D}$.
 	(Until this query, $H_{\indexReprogrammingInstances-1}$ and $H_{\indexReprogrammingInstances}$ do not differ.)
	$\hat{\Ad{D}}_{\indexReprogrammingInstances, 0}$ then uses as its output to game $G'_b$ the $\indexReprogrammingInstances$-th distribution on which $\Reprogram'$ was queried.
	The second stage $\hat{\Ad{D}}_{\indexReprogrammingInstances, 1}$ uses its input $x^*$ to simulate the  $r$-th response of $\Reprogram'$.
	As from (and including) the $(\indexReprogrammingInstances+1)$-th query,
	$\hat{\Ad{D}}_{\indexReprogrammingInstances, 1}$ can simulate the reprogramming by using fresh uniformly random values
	to overwrite $\RO{O}'$.
	To be more precise, during each call to $\Reprogram'$ on some distribution $p$,
	$\hat{\Ad{D}}_{\indexReprogrammingInstances, 1}$ samples $x \leftarrow p$ and $y \uni Y$,
	and adds $(x,y)$ to a list $\List{\RO{O}}$.
	(If $x$ has been sampled before, $\hat{\Ad{D}}_{\indexReprogrammingInstances, 1}$ replaces the former oracle value in the list.)
	$\hat{\Ad{D}}_{\indexReprogrammingInstances, 1}$ defines $\RO{O}$ by
	\begin{align*}
		\RO{O}(x) := \begin{cases}
			y			& \exists y \text{ s.th. } (x, y) \in \List{\RO{O}} \\
			\RO{O'}(x)	& \text{o.w.}
		\end{cases}%
	\end{align*}%
		
	In the case that $\hat{\Ad{D}_\indexReprogrammingInstances}$ is run in game $G_0'$, 
	the reprogramming starts with the $(\indexReprogrammingInstances+1)$-th query
	and $\hat{\Ad{D}_\indexReprogrammingInstances}$ perfectly simulates game $H_{\indexReprogrammingInstances}$.
	In the case that $\hat{\Ad{D}_\indexReprogrammingInstances}$ is run in game $G_1'$, 
	the reprogramming already starts with the $\indexReprogrammingInstances$-th query
	and $\hat{\Ad{D}_\indexReprogrammingInstances}$ perfectly simulates game $H_{\indexReprogrammingInstances-1}$.
	\begin{align*}
	|\Pr [ H_{\indexReprogrammingInstances-1}^{\Ad{D}} \Rightarrow 1 ] - \Pr [ H_{\indexReprogrammingInstances}^{\Ad{D}} \Rightarrow 1 ] |
	=
	|\Pr [ {G_{1}'}^{\hat{\Ad{D}}_\indexReprogrammingInstances} \Rightarrow 1 ] - \Pr [{G_{0}'}^{\hat{\Ad{D}}_\indexReprogrammingInstances} \Rightarrow 1 ] |
	\enspace .
	\end{align*}
	Since the first stage $\hat{\Ad{D}}_{\indexReprogrammingInstances, 0}$ issues $\hat q_{\indexReprogrammingInstances}$ many queries to $\RO{O'}$,
	we can apply \cref{thm:Reprogramming:SingleInstance} to obtain
	\begin{align*}
	|\Pr [ {G_{1}'}^{\hat{\Ad{D}}_\indexReprogrammingInstances} \Rightarrow 1 ] - \Pr [{G_{0}'}^{\hat{\Ad{D}}_\indexReprogrammingInstances} \Rightarrow 1 ] |
	&\leq \sqrt{\hat q_\indexReprogrammingInstances \cdot p^{(\indexReprogrammingInstances)}_{\max}}
		+ \frac 1 2 \hat q_\indexReprogrammingInstances \cdot p^{(\indexReprogrammingInstances)}_{\max}
	\enspace .
	\end{align*}
	\qed
\end{proof}

Second, we prove that \cref{cor:Reprogramming:MultipleInstance} implies our main \cref{theorem:reprGameBased}. For this we have to show that additional side-information can be simulated by a reduction.

\begin{proof}
	
	Consider a distinguisher $\Ad{D}$ run in games $\gameRepro_b$.
	To upper bound $\Ad{D}$'s advantage,
	we now define 
	a distinguisher $\Ad{\hat{D}}$ against the helper games $G_b$ from \cref{fig:Def:Games:Reprogramming:MultipleInstance}.
	
	When queried on a distribution $p$ on $\XFirst \times \XSecond$,
	$\Ad{\hat{D}}$ will simulate $\Reprogram$ as follows:
	$\Ad{\hat{D}}$ will forward the marginal distribution $p_{\XFirst}$ of $\xFirst$
	to its own oracle $\Reprogram'$,
	and obtain some $\xFirst$ that was sampled accordingly.
	It will then sample $\xSecond$ according to $p_{\XSecond|\xFirst}$,
	where $p_{\XSecond|\xFirst}$ is the probability distribution on $\XSecond$,
	conditioned on $\xFirst$,
	i.e.,
	\[p_{\XSecond|\xFirst}(\xSecond)
		:= \frac{\Pr[\xFirst, \xSecond]}
				{\Pr[\xFirst]} \enspace .\]
	where the probabilities are taken over $(\xFirst, \xSecond) \leftarrow p$,
	and the probability in the denominator is taken over $\xFirst \leftarrow p_{\XFirst}$.
	Note that $\Ad{\hat{D}}$ can be unbounded, as the statement of \cref{cor:Reprogramming:MultipleInstance} is information-theoretical. This is important because while $p$ is efficiently sampleable, $p_{\XSecond|\xFirst}$ might not be.
	Since the distribution of $(\xFirst, \xSecond)$ is identical to $p$,
	and since the reprogramming only happens on $\xFirst$, 
	$\Ad{\hat{D}}$ perfectly simulates game $\gameRepro_b$ to $\Ad{D}$ if run in game $G_b$ and
	\begin{align*}
	| \Pr[\gameRepro_{1}^{\Ad{D}} \Rightarrow 1] - \Pr[\gameRepro_{0}^{\Ad{D}} \Rightarrow 1] |
	= 
	| \Pr[G_{1}^{\Ad{\hat{D}}} \Rightarrow 1] - \Pr[G_{0}^{\Ad{\hat{D}}} \Rightarrow 1] |  \enspace . \nonumber
	\end{align*}
	
	Since $\Ad{\hat{D}}$ can answer any random oracle query issued by $\Ad{D}$
	by simply forwarding it,
	$\Ad{\hat{D}}$ issues exactly as many queries to $\RO{O}$ (until the $\indexReprogrammingInstances$-th reprogramming instruction) as $\Ad{D}$.
	We can now apply \cref{cor:Reprogramming:MultipleInstance} to obtain
	
	\begin{align*}
		| \Pr[G_{1}^{\Ad{\hat{D}}} \Rightarrow 1] - \Pr[G_{0}^{\Ad{\hat{D}}} \Rightarrow 1] |
	&\leq \sum_{r = 1}^{R}  \left( \sqrt{ \hat q_{r} p_{\max}^{(r)} }  + \frac 1 2 \hat q_{r} p_{\max}^{(r)} 
	\right) 
	\enspace , \nonumber
	\end{align*}
	
	where $p_{\max}^{(r)} = \mathbb E \max_{x} p^{(r)}_{\XFirst}(x)$.

\end{proof}

\section{Definitions: Hash functions, and identification and signature schemes}	\label{sec:SignaturesInQROM:Defs:Appendix}

\subsection{Security of Hash Functions}\label{app:def:hash}
One of our proofs makes use of a multi-target version of extended target-collision resistance (\mmetcr). Extended target-collision resistance is a variant of target collision resistance, where the adversary also succeeds if the collision is under a different key. As we will consider this notion for a random oracle, we adapt the notion slightly for this setting. 

The success probability of an adversary \A against 
\mmetcr security of a quantum-accessible random oracle $\RO{H}$ is defined as follows. The definition makes use of a (classical) challenge oracle $\bbox(\cdot)$ which on input 
of the $j$-th message $x_j$ outputs a uniformly random function key $z_j$.
\begin{align*}
\succmmetcr{\RO{H}}{\A,p} & = \pr\left[\right. (x',z',i) \exec 
\A^{\bbox,\qRO{H}}( ): \nonumber\\
& \left. x'\neq x_i \wedge \RO{H}(z_i\|x_i) = 
\RO{H}(z'\|x') \wedge 0< i\leq p\right].\,
\end{align*}

A later proof makes use of another version of target collision resistance called multi-target extended target-collision 
resistant with nonce (\nmmetcr). The definition of \nmmetcr again makes use of a (classical) challenge 
oracle $\bbox(\cdot)$ that on input of the $j$-th message $M_j$ outputs a 
uniformly random function key $z_j$. Before the experiment starts, $\A$ 
can select an arbitary $n$-bit string \id.
\begin{align*}
\succnmmetcr{\RO{H}}{\A,p} & = \pr\left[\right. \id\exec\A(), (x',z',i) 
\exec 
\A^{\bbox(\cdot),\qRO{H}}(\id): \nonumber\\
& \left. x'\neq x_i \wedge \RO{H}(z_i,\id,i,x_i) = 
\RO{H}(z',\id,i,x') \wedge 0< i \leq p\right]\,.
\end{align*}

The QROM bound for \nmmetcr makes use of the following two games omitted in the main body of the paper. The games were defined in the proof in~\cite{XMSSembed}, we rephrased them to match our notation:

\heading{Game $G_{a,i}$.} After \A selected \id, it gets access to $\RO{H}$. In phase 1, after making at most $q_1$ queries to $\RO{H}$, \A outputs a message $x\in X$. Then
	a random $z \uni Z$ is sampled and $(z,\RO{H}(z\|\id\|i\|x))$ is
	handed to $\A$. $\A$ continues to the second phase and makes at most
	$q_2$ queries. $\A$ outputs $b\in \{0,1\}$ at the end.
	
\heading{Game $G_{b,i}$.} After \A selected \id, it gets access to $\RO{H}$.  After making at most $q_1$
	queries to $\RO{H}$, \A outputs a message $x\in X$. Then a random
	$z \uni Z$ is sampled as well as a random range element
	$y\uni Y$. Program $\RO{H} := \RO{H}^{(z\|\id\|i\|x)\mapsto y}$. $\A$ receives $(z, y=\RO{H}(z\|\id\|i\|x))$ and proceeds to the second
	phase. After making at most $q_2$ queries, $\A$ outputs
	$b\in \{0,1\}$ at the end.

In a reduction to apply our new bound, \id and $i$ would be sent as part of the message to \Reprogram. Note that the proof in~\cite{XMSSembed} runs a hybrid argument over \qsign programmings and in every programming step a different value for $i$ is used, thereby making those distinct.

\subsection{Identification schemes}

We now define syntax and security of identification schemes.
Let $\param$ be a tuple of common system parameters shared among all participants.

\begin{definition}[Identification schemes]
	An identification scheme $\IdScheme$ is defined as a collection of algorithms
	$\IdScheme = (\IG, \Commit, \Respond, \VerifyId)$.
	\begin{itemize}
		\item The key generation algorithm $\IG$ takes system parameters $\param$ as input and returns
		the public and secret keys $(\pk, \sk)$.
		We assume that $\pk$ defines the challenge space $\ChallengeSpace$, the commitment space $\CommSpace$ and the response space $\ResponseSpace$.
		
		\item $\Commit$ takes as input the secret key $\sk$ and returns a commitment $\commitment \in \CommSpace$ and a state $\state$
		
		\item $\Respond$ takes as input the secret key $\sk$, a commitment  $\commitment $, a challenge $\challenge$, and a state $\state$, and returns a response $\response \in \ResponseSpace \cup \lbrace \bot \rbrace$,
		where $\bot  \notin \ResponseSpace$ is a special symbol indicating failure.
		
		\item The deterministic verification algorithm $\VerifyId(\pk, \commitment, \challenge, \response)$ returns 1 (accept) or 0 (reject).
	\end{itemize}
\end{definition}

Note that during one of our application examples (i.e., in \cref{sec:HedgedFiatShamir}),
we define the response algorithm such that it does not explicitly take a commitment $\commitment$ as input.
If needed, it can be assumed that $\state$ contains a copy of $\commitment$.

A \textit{transcript} is a triplet $\transcript = (\commitment, \challenge, \response) \in  \CommSpace \times \ChallengeSpace \times \ResponseSpace$.
It is called \textit{valid} (with respect to public key $\pk$)
if $\VerifyId(\pk, \commitment, \challenge, \response) = 1$.
Below, we define a transcript oracle $\getTrans$ that returns the
transcript $\transcript = (\commitment, \challenge, \response)$
of a real interaction between prover and verifier.
We furthermore define another transcript oracle $\getTransChallenge$ that returns an honest transcript for a fixed challenge $\challenge$.

\begin{figure}[h] \begin{center} \fbox{ \small
			
			\nicoresetlinenr
			
			\begin{minipage}[t]{5cm}
				
				\underline{\textbf{Algorithm} $\getTrans(\sk)$}
				\begin{nicodemus}
					
					\item $(\commitment, \state) \leftarrow \Commit(\sk)$
					
					\item $\challenge \uni \ChallengeSpace$
					
					\item $\response \leftarrow \Respond(\sk, \commitment, \challenge, \state)$
					
					\item \pcreturn $(\commitment, \challenge, \response)$
					
				\end{nicodemus}
				
			\end{minipage}
			
			\quad 
			
			\begin{minipage}[t]{5cm}
				
				\underline{\textbf{Algorithm} $\getTransChallenge(\sk, \challenge)$}
				\begin{nicodemus}
					
					\item $(\commitment, \state) \leftarrow \Commit(\sk)$
					
					\item $\response \leftarrow \Respond(\sk, \commitment, \challenge, \state)$
					
					\item \pcreturn $(\commitment, \challenge, \response)$
					
				\end{nicodemus}		
				
			\end{minipage}
		}
	\end{center}
	\caption{%
		Generating honest transcripts with oracles $\getTrans$ and $\getTransChallenge$.
		\label{fig:Def:getTrans}
	}
\end{figure}

{%
\heading{Commitment entropy.}
We define
\[\maxEntropyCommit:= \mathbb{E} \max_{\commitment}
\Pr [\commitment] \enspace ,\]
where the expectation is taken over $(\pk, \sk) \leftarrow \IG$,
and the probability is taken over $(\commitment, \state) \leftarrow \Commit(\sk)$. }

{%
In one of our applications (\cref{sec:HedgedFiatShamir}),
the $\Respond$ algorithm is required to reject whenever its challenge input $\challenge$ is malformed.
As observed in \cite{EC:AOTZ20}, this additional requirement is not too severe,
since most practical implementations perform a sanity check on $\challenge$.
We will call this property validity awareness.
\begin{definition}[Validity awareness]\label{Def:ValidityAwareness}
	We say that $\IdScheme$ is \textit{validity aware}
	if $\Respond(\sk, \commitment, \challenge, \state) = \bot$
	for all challenges $\challenge \notin \ChallengeSpace$.
\end{definition} }

{%
}

\heading{Statistical HVZK.}
We \ifeprint \else now \fi recall the definition of statistical honest-verifier zero-knowledge ($\HVZK$),
and the definition of \textit{special} statistical $\HVZK$ from \cite{EC:AOTZ20}.
\begin{definition}[(Special) statistical $\HVZK$]	\label{Def:HVZK:Statistical}
	Assume that $\IdScheme$ comes with an $\HVZK$ simulator $\Sim$.
	We say that $\IdScheme$ is \textit{$\boundStatisticalHVZK$-statistical $\HVZK$} if for any key pair $(\pk, \sk) \in \supp(\IG)$,
	the distribution of $(\commitment, \challenge, \response) \leftarrow \Sim(\pk)$
	has statistical distance at most $\boundStatisticalHVZK$ from an honest transcript
	$(\commitment, \challenge, \response) \leftarrow \getTrans(\sk)$.
	
	To define \textit{special} statistical $\HVZK$, 
	assume that
	$\IdScheme$ comes with a special $\HVZK$ simulator $\Sim$.
	We say that $\IdScheme$ is \textit{$\boundSpecialStatisticalHVZK$-statistical $\specialHVZK$} if 
	for any key pair $(\pk, \sk) \in \supp(\IG)$ and any challenge $\challenge \in \ChallengeSpace$,
	the distribution of $(\commitment, \challenge, \response) \leftarrow \getTransChallenge(\sk, \challenge)$ 
	and the distribution of $(\commitment, \challenge, \response) \leftarrow \Sim(\pk, \challenge)$
	have statistical distance at most $\boundSpecialStatisticalHVZK$.
\end{definition} 
{%
Following \cite{EC:AOTZ20}, we now define subset-revealing identification schemes.
Intuitively, an identification scheme is subset-revealing if
$\Respond$ responds to a challenge by revealing parts of the state that was computed by $\Commit$, and does not depend on $\sk$.
\begin{definition}[Subset-revealing identification protocol]\label{Def:SubsetRevealing}
	
	Let $\IdScheme = (\IG, \Commit, \Respond, \VerifyId)$ be an identification protocol.
	We say that $\IdScheme$ is \textit{subset-revealing} if for any key pair in the support of $\IG$ it holds that

	\begin{itemize}
		
		\item the challenge space $\ChallengeSpace$ is polynomial in the security parameter,
		
		\item for any tuple $(\commitment, \state) \in \supp(\Commit(\sk))$, the state $\state$ consists of a collection $(\state_1, \cdots , \state_N)$ such that $N$ is polynomial in the security parameter,
		
		\item and furthermore there exists an algorithm $\mathsf{DeriveSet}$ such that
			
			\begin{itemize}
				
				\item $\mathsf{DeriveSet}$ takes as input a challenge $\challenge$
				and returns a subset $I \subset \lbrace 1, \cdots , N \rbrace$,
				
				\item for any tuple $(\commitment, \state)\in \supp(\Commit(\sk)$
				and any challenge $\challenge \in \ChallengeSpace$
				\ifeprint we have \else it holds that\fi $\Respond(\sk, \challenge, \state) = (\state_i)_{i \in I}$,
				where $I = \mathsf{DeriveSet}(\challenge)$.
				
			\end{itemize}
	\end{itemize}
	
\end{definition} }

\subsection{Signature schemes}

We now define syntax and security of digital signature schemes.
Let $\param$ be a tuple of common system parameters shared among all participants.

\begin{definition}[Signature scheme]
	A digital signature scheme $\SignatureScheme$ is defined as a triple of algorithms
	$\SignatureScheme = (\KG, \Sign, \VerifySig)$.
	\begin{itemize}
		\item The key generation algorithm $\KG(\param)$
			returns a key pair $(\pk, \sk)$.
			We assume that $\pk$ defines the message space $\MSpace$.
			
		\item The signing algorithm $\Sign(\sk, m)$ returns a signature $\sigma$.
		
		\item The deterministic verification algorithm $\VerifySig(\pk, m, \sigma)$
			returns 1 (accept) or 0 (reject).
	\end{itemize}
\end{definition}

\heading{$\UFCMA$, $\UFCMAZero$ and \eufrma security.}
We define \underline{U}n\underline{F}orgeability under \underline{C}hosen \underline{M}essage \underline{A}ttacks ($\UFCMA$),
\underline{U}n\underline{F}orgeability under \underline{C}hosen \underline{M}essage \underline{A}ttacks with 0 queries to the signing oracle
($\UFCMAZero$, also known as $\UFKOA$ or $\UF\text{-}\mathsf{NMA}$)
and \underline{U}n\underline{F}orgeability under \underline{R}andom \underline{M}essage \underline{A}ttacks (\eufrma)
success functions of a quantum adversary $\Ad{A}$ against $\SignatureScheme$
as
\[\succfun^{\UF - X}_{\SignatureScheme}(\Ad{A})
:= \Pr[\UF\text{-}X_{\SignatureScheme}^\Ad{A} \Rightarrow 1 ] \enspace , \]
where the games for $X \in \lbrace \CMA$, $\CMAZero, \mathsf{RMA} \rbrace$ are given in \cref{fig:Def:UFCMA}.

\begin{figure}[h] \begin{center} \makebox[\textwidth][c]{\fbox{ \small
			
	\begin{minipage}[t]{4.3cm}
		
		\nicoresetlinenr
		\underline{\textbf{Game} \boxedFull{$\UFCMA$} \dashboxed{$\UFCMAZero$}}
		\begin{nicodemus}
			
			\item $(\pk, \sk) \leftarrow \IG(\param)$
			
			\item \boxedFull{$(m^*, \sigma^*) \leftarrow \Ad{A}^{\oracleSIGN} (\pk)$}
			
			\item \dashboxed{ $(m^*, \sigma^*) \leftarrow \Ad{A}(\pk)$}
			
			\item \pcif $m^* \in \ListOfMessages$ \pcreturn 0
			
			\item \pcreturn $\VerifySig(\pk, m^*, \sigma^*)$
			
		\end{nicodemus}		
		
	\end{minipage}
	
	\;
			
	\begin{minipage}[t]{3cm}
		
		\underline{$\oracleSIGN(m)$}
		
		\begin{nicodemus}
			
			\item $\ListOfMessages := \ListOfMessages \cup \lbrace m \rbrace$
			
			\item $\sigma \leftarrow \Sign(\sk,m)$
			
			\item \pcreturn $\sigma$
			
		\end{nicodemus}
		
	\end{minipage}

	\;
	
	\begin{minipage}[t]{5.5cm}
		
		\nicoresetlinenr
		
		\underline{\textbf{GAME} \eufrma}
		\begin{nicodemus}
			
			\item $(\pk, \sk) \leftarrow \KG()$
			
			\item $\{m_1 , \cdots, m_N\} \uni \MSpace^{N}$
			
			\item \pcfor $i \in \lbrace 1, \cdots, N \rbrace$
			
				\item \quad $\sigma_i \exec \Sign(\sk, m_i)$
			
			\item $(m^*, \sigma^*) \exec \Ad{A}(\pk, \{(m_i,\sigma_i)\}_{1 \leq i \leq N})$
			
			\item \pcif  $m^* \in \{m_i\}_{1 \leq i \leq N}$
				
				\item \quad \pcreturn 0
				
			\item \pcreturn $\VerifySig(\pk, m^*, \sigma^*)$
			
		\end{nicodemus}		
		
	\end{minipage}
			
	}}
	\end{center}
	\caption{
		Games $\UFCMA$, $\UFCMAZero$ (left) and $\eufrma$ (right) for $\SignatureScheme$.
		Game $\eufrma$ is defined relative to $N$, the number of message-signature pairs the adversary is given.
	}
	\label{fig:Def:UFCMA}
\end{figure}

{\heading{The hash-and-sign construction.}
	To a signature scheme $\sigS = (\KG, \Sign, \VerifySig)$ with message space $\MSpaceBeforeHTS$,
	and a hash function $\RO{H}: Z \times \MSpaceAfterHTS \rightarrow \MSpaceBeforeHTS$ (later modeled as a RO),
	we associate
	\[ \sigSp := \hts(\sigS,\RO{H}) := (\KG, \Sign', \VerifySig') \enspace \]
	with message space $\MSpaceAfterHTS$,
	where algorithms $\Sign'$ and $\VerifySig'$ of $\sigSp$ are defined in \cref{fig:Def:hts}.

	\begin{figure}[h] \begin{center} \fbox{
				
				\nicoresetlinenr	
				
				\begin{minipage}[t]{3.7cm}
					
					\underline{$\Sign'(\sk, \messageAfterHTS)$}
					\begin{nicodemus}
						
						\item $z \uni Z$
						
						\item $\sigma = \Sign(\sk,\RO{H}(z\|\messageAfterHTS))$ 
						
						\item \pcreturn $\sigma' = (z,\sigma)$
						
					\end{nicodemus}
					
				\end{minipage}
				
				\quad
				
				\begin{minipage}[t]{3.7cm}
					
					\underline{$\VerifySig'(\pk, \sigma', \messageAfterHTS)$}
					\begin{nicodemus}
						
						\item Parse $\sigma'$ as $(z,\sigma)$
						
						\item $\messageBeforeHTS = \RO{H}(z\|\messageAfterHTS)$ 
						
						\item \pcreturn $\VerifySig(\pk, \messageBeforeHTS, \sigma)$
						
					\end{nicodemus}
				\end{minipage}
		}\end{center}
		\caption{Construction $\sigSp=\hts(\sigS,\RO{H})$. Key generation remains the same as in $\sigS$.}
		\label{fig:Def:hts}
	\end{figure}

}

{\heading{The Fiat-Shamir transform.}
	To an identification scheme $\IdScheme = (\IG, \Commit, \Respond, \VerifyId)$
	with commitment space $\CommSpace$, 
	and random oracle
	$\ROChallenge:  \CommSpace \times \MSpace \rightarrow \ChallengeSpace$
	for some message space $\MSpace$,
	we associate
	\[ \FS [\IdScheme, \ROChallenge]
	:= \SignatureScheme := (\IG, \Sign, \VerifySig) \enspace ,\]
	where algorithms $\Sign$ and $\VerifySig$ of $\SignatureScheme$ are defined in \cref{fig:Def:FiatShamir}.
	
	We will also consider the modified Fiat-Shamir transform, in which lines \ref{line:Def:FS:ChallengeHash1} and \ref{line:Def:FS:ChallengeHash2} are replaced with $\challenge := \ROChallenge(\commitment, m, \pk)$.

	\begin{figure}[h] \begin{center} \fbox{ \small
				
		\begin{minipage}[t]{4cm}
			
			\nicoresetlinenr
			
			\underline{$\Sign(\sk, m)$}
			\begin{nicodemus}
				
				\item $(\commitment, \state) \leftarrow \Commit(\sk)$
				
				\item $\challenge := \ROChallenge(\commitment, m)$ \label{line:Def:FS:ChallengeHash1}
				
				\item $\response \leftarrow \Respond(\sk, \commitment, \challenge, \state)$
				
				\item \pcreturn $\sigma := (\commitment, \response)$
				
			\end{nicodemus}		
			
		\end{minipage}
	
		\quad
		
		\begin{minipage}[t]{4cm}
			
			\underline{$\VerifySig(\pk, m, \sigma = (\commitment, \response))$}
			\begin{nicodemus}
				
				\item $\challenge := \ROChallenge(\commitment, m)$ \label{line:Def:FS:ChallengeHash2}
				
				\item \pcreturn $\VerifyId(\pk, \commitment, \challenge, \response)$
				
			\end{nicodemus}		
			
		\end{minipage}
	
		} \end{center}
		\caption{Signing and verification algorithms of $\SignatureScheme = \FS[\IdScheme, \ROforHedging]$.
		}
		\label{fig:Def:FiatShamir}
	\end{figure}
}

\section{From $\UFCMAZero$ to $\UFfaultCMADifferentSet{}$ (Proof of \cref{theorem:UFfaultCMA})}%
\label{sec:HedgedFiatShamir:ProofFCMA}

\newcounter{CounterCMAtoFCMA} %
\setcounter{CounterCMAtoFCMA}{1}

{
	\newcounter{SimulateNine}
	\setcounter{SimulateNine}{\theCounterCMAtoFCMA}
	\stepcounter{CounterCMAtoFCMA}
}

{
	\newcounter{SimulateFive}
	\setcounter{SimulateFive}{\theCounterCMAtoFCMA}
	\stepcounter{CounterCMAtoFCMA}
}

{
	\newcounter{SimulateSix}
	\setcounter{SimulateSix}{\theCounterCMAtoFCMA}
	\stepcounter{CounterCMAtoFCMA}
}

{
	\newcounter{SimulateSeven}
	\setcounter{SimulateSeven}{\theCounterCMAtoFCMA}
	\stepcounter{CounterCMAtoFCMA}
}

{
	\newcounter{SimulateFour}
	\setcounter{SimulateFour}{\theCounterCMAtoFCMA}
	\stepcounter{CounterCMAtoFCMA}
}

\newcounter{LastGameCMAtoFCMA}
\setcounter{LastGameCMAtoFCMA}{\theCounterCMAtoFCMA}

\newcommand{\gameSimulateNine}{G_{\theSimulateNine}\xspace}
\newcommand{\gameSimulateFive}{G_{\theSimulateFive}\xspace}
\newcommand{\gameSimulateSix}{G_{\theSimulateSix}\xspace}
\newcommand{\gameSimulateSeven}{G_{\theSimulateSeven}\xspace}
\newcommand{\gameSimulateFour}{G_{\theSimulateFour}\xspace}

\newcommand{\LastGameCMAtoFCMA}{G_{\theLastGameCMAtoFCMA}\xspace} 
We now give the proof for \cref{theorem:UFfaultCMA}:
\thmfour*

Following the proof structure of \cite{EC:AOTZ20},
we will break down the proof into several sequential steps.
Consider the sequence of games, given in \cref{fig:UFfaultCMA:games}.
With each game-hop, we take one more index $i$ for which we replace execution of $\oracleFaultSIGN$
with a simulation that can be executed without knowledge of $\sk$,
see line \cref{line:UFfaultCMA:UseSimulations}.
The workings of these simulations will be made explicit in the proof for the respective game-hop.
Similar to \cite{EC:AOTZ20}, the order of the indices for which we start simulating is 9, 5, 6, 7, 4.

For a scheme that cannot be assumed to be subset-revealing, we will only proceed until game $G_3$,
and then use game $G_3$ to argue that we can turn any adversary
against the $\UFfaultCMADifferentSet{\lbrace 5, 6, 9 \rbrace}$ security of $\SignatureScheme$
into an  $\UFCMAZero$ adversary (see \cref{lem:UFfaultCMA:Adversary569}).

If we can assume the scheme to be subset-revealing, we will proceed until game $G_5$,
and then use game $G_5$ to argue that we can turn any adversary against the $\UFfaultCMADifferentSet{\lbrace 4, 5, 6, 7, 9 \rbrace}$ security of $\SignatureScheme$ into an  $\UFCMAZero$ adversary (see \cref{lem:UFfaultCMA:AdversaryAllFaults}).

Note that our sequential proof is given for statistical $\specialHVZK$.
The reason why we do not give our proof in the computational setting
right away is that it would then be required to make all of our changes at once,
rendering the proof overly involved, while not providing any new insights.
At the end of this section, we show how to generalise the proof to the computational setting.

\begin{figure}[h] \begin{center} \makebox[\textwidth][c]{\fbox{ \small
			
		\begin{minipage}[t]{4.7cm}
			
			\nicoresetlinenr
			\underline{\textbf{Games} $G_0$ - $G_5$}
			\begin{nicodemus}
				
				\item $(\pk, \sk) \leftarrow \IG(\param)$
				
				\item $(m^*, \sigma^*) \leftarrow \Ad{A}^{\oracleFaultSIGN, \qRO{\ROChallenge}} (\pk)$
				
				\item \pcif $m^* \in \ListOfMessages$ \pcreturn 0
				
				\item Parse $(\commitment^*, \response^*) := \sigma^*$
				
				\item $\challenge^* := \ROChallenge(\commitment^*, m^*)$
				
				\item \pcreturn $\VerifyId(\pk, \commitment^*, \challenge^*, \response^*)$
				
			\end{nicodemus}		
			
		\end{minipage}
		
		\;
		
		\begin{minipage}[t]{4.8cm}
			
			\underline{$\oracleFaultSIGN(m, i \in \SetFaultFunctions, \phi)$}
			
			\begin{nicodemus}
				
				\item $S := \emptyset$ \gcom{$G_0$}
				
				\item $S := \lbrace 9 \rbrace$ \gcom{$G_1$-$G_5$}
				
				\item $S := S \cup \lbrace 5 \rbrace$ \gcom{$G_2$-$G_5$}
				
				\item $S := S \cup \lbrace 6 \rbrace$ \gcom{$G_3$-$G_5$}
				
				\item $S := S \cup \lbrace 7 \rbrace$ \gcom{$G_4$-$G_5$}
				
				\item $S := S \cup \lbrace 4 \rbrace$ \gcom{$G_5$}
				
				\item \pcif $i \in S$ 
				
					\item \quad $\sigma \leftarrow \simSignature_i(m, \phi)$ \label{line:UFfaultCMA:UseSimulations}
				
				\item \pcelse $\sigma \leftarrow \getSignature(m, i, \phi)$
				
				\item \pcreturn $\sigma$
				
			\end{nicodemus}
			
		\end{minipage}
		
		\;
		
		\begin{minipage}[t]{4.4cm}
			
			\underline{$\getSignature(m, i, \phi)$}
			
			\begin{nicodemus}
				
				\item $f_i := \phi$ and $f_j := Id  \  \forall \ j \neq i$
				
				\item $(\commitment, \state) \leftarrow \Commit(\sk)$
				
				\item $(\commitment, \state) := f_4(\commitment, \state)$
				
				\item $(\hat{\commitment}, \hat{m}, \hat{\pk}) := f_5(\commitment, m, \pk)$
				
				\item $\challenge := f_6(\ROChallenge(\hat{m}, \hat{\commitment}, \hat{\pk}))$
				
				\item $\response \leftarrow \Respond(f_7(\sk, c, \state))$
				
				\item $\ListOfMessages := \ListOfMessages \cup \lbrace \hat{m} \rbrace$
				
				\item \pcreturn $\sigma := f_9(\commitment, \response)$
				
			\end{nicodemus}

		\end{minipage}
			
}}
	\end{center}
	\caption{
		Games $G_0$ - $G_5$ for the proof of \cref{theorem:UFfaultCMA}.
		Helper methods $\getSignature$ and $\simSignature_i$ (where $i \in \lbrace 4, 5, 6, 7, 9 \rbrace$) are internal and cannot be accessed directly by $\Ad{A}$.
		Recall that we require queried indices $i$ to be contained in $\SetFaultFunctions$ (see \cref{fig:Def:UF_Fault_CMA}).
	}
	\label{fig:UFfaultCMA:games}
\end{figure}

{\heading{Game $G_0$.}
	Since game $G_0$ is the original $\UFfaultCMA$ game,
	\[ \succfun^{\UFfaultCMA}_{\SigSchemeHedged}(\Ad{A})
	= \Pr [ G_0^{\Ad{A}} \Rightarrow 1]\enspace.\]
}

{\heading{Games $\gameSimulateNine$ - $\gameSimulateSix$.}
	In games $\gameSimulateNine$ to $\gameSimulateSix$, we sequentially start to simulate faulty signatures for fault indices 9, 5 and 6.
	
	\begin{lemma}\label{lem:UFfaultCMA:9}%
		There exists an algorithm $\simSignature_9$ such that
		for any adversary $\Ad{A}$ against the $\UFfaultCMA$ security of $\SignatureScheme$,
		issuing at most $q_{S,9}$ queries to $\oracleFaultSIGN$ on index 9,
		$q_{S}$ queries to $\oracleFaultSIGN$ in total,
		and at most $q_\ROChallenge$ queries to $\ROChallenge$,
		\begin{align}
			\gameDist{SimulateNine}{A}	\label{eq:UFfaultCMA:9}
			\leq q_{S,9} \cdot \left(
					\boundSpecialStatisticalHVZK
					+ \frac {3}{2}  \sqrt{ (q_H + q_{S} + 1) \cdot \maxEntropyCommit} 
				\right) \enspace .
		\end{align}
	\end{lemma}
	The details on algorithm $\simSignature_9$ and the proof for \cref{eq:UFfaultCMA:9} are given in \cref{sec:HedgedFiatShamir:ProofFault9}.
	\begin{lemma}\label{lem:UFfaultCMA:5}%
		There exists an algorithm $\simSignature_5$ such that
		for any adversary $\Ad{A}$ against the $\UFfaultCMA$ security of $\SignatureScheme$,
		issuing at most $q_{S,5}$ queries to $\oracleFaultSIGN$ on index 5,
		$q_{S}$ queries to $\oracleFaultSIGN$ in total,
		and at most $q_\ROChallenge$ queries to $\ROChallenge$,
		\begin{align}
		\gameDist{SimulateFive}{A}	\label{eq:UFfaultCMA:5}
			\leq q_{S,5} \cdot \left(
				\boundSpecialStatisticalHVZK 
				+ \frac {3}{2}  \sqrt{ (q_H + q_{S} + 1) \cdot 2 \maxEntropyCommit}
			\right) \enspace .
		\end{align}
	\end{lemma}
	The details on algorithm $\simSignature_5$ and the proof for \cref{eq:UFfaultCMA:5} are given in \cref{sec:HedgedFiatShamir:ProofFault5}.	
	\begin{lemma}\label{lem:UFfaultCMA:6}%
		
		There exists an algorithm $\simSignature_6$ such that
		for any adversary $\Ad{A}$ against the $\UFfaultCMA$ security of $\SignatureScheme$,
		issuing at most $q_{S,6}$ queries to $\oracleFaultSIGN$ on index 6,
		$q_{S}$ queries to $\oracleFaultSIGN$ in total,
		and at most $q_\ROChallenge$ queries to $\ROChallenge$,
		
		\begin{align}
		\gameDist{SimulateSix}{A}	\label{eq:UFfaultCMA:6}
			\leq q_{S,6} \cdot \left(
				\boundSpecialStatisticalHVZK
				+ \frac {3}{2}  \sqrt{ (q_H + q_{S} + 1) \cdot \maxEntropyCommit}
			\right) \enspace . 
		\end{align}
	\end{lemma}
	The details on algorithm $\simSignature_6$ and the proof for \cref{eq:UFfaultCMA:6} are given in \cref{sec:HedgedFiatShamir:ProofFault6}.
	What we have shown by now is that
	\begin{align}
	| \Pr [ G_0^{\Ad{A}} \Rightarrow 1] - \Pr [ \gameSimulateSix^{\Ad{A}} \Rightarrow 1]|
	\leq q_{S,\lbrace 5,6, 9 \rbrace} \cdot \left( \boundSpecialStatisticalHVZK
			+ \frac {3}{2} \sqrt{(q_H + q_{S} + 1) \cdot 2 \maxEntropyCommit}
		\right) \enspace, 
	\end{align}
	where $q_{S,\lbrace 5,6, 9 \rbrace}$ denotes the maximal number of queries to $\oracleFaultSIGN$ on all indices $i \in \lbrace 5,6, 9 \rbrace$.
	We are now ready to give our first security statement.
	
	\begin{lemma}\label{lem:UFfaultCMA:Adversary569}%
		
		For any adversary $\Ad{A}$ against the $\UFfaultCMADifferentSet{\lbrace 5,6,9 \rbrace}$ security of $\SignatureScheme$, 
		there exists an adversary $\Ad{B}$ such that
		
		\begin{align*}
			\Pr [ \gameSimulateSix^{\Ad{A}} \Rightarrow 1]
			\leq	\succfun^{\UFCMAZero}_{\FS[\IdScheme, \ROChallenge]}(\Ad{B}) \enspace ,
		\end{align*}
		and $\Ad{B}$ has the same running time as $\Ad{A}$.
	\end{lemma}	
	The proof is given in \cref{sec:HedgedFiatShamir:ProofAdversary569}.
}
Collecting the probabilities, we obtain
\begin{align*}
\succfun^{\UFfaultCMADifferentSet{\lbrace 5,6,9 \rbrace}}_{\FS[\IdScheme, \ROChallenge]}(\Ad{A})
\leq & \succfun^{\UFCMAZero}_{\FS[\IdScheme, \ROChallenge]}(\Ad{B}) \\
	& + q_{S} \cdot \left(
		\boundSpecialStatisticalHVZK
		+ \frac {3}{2}  \sqrt{(q_H + q_{S} + 1) \cdot 2\maxEntropyCommit}
	\right) \enspace .
\end{align*}

{\heading{Games $\gameSimulateSeven$ - $\gameSimulateFour$.}
	In games $\gameSimulateSeven$ to $\gameSimulateFour$, we sequentially start to simulate faulty signatures for fault indices 7 and 4.
	
	\begin{lemma}\label{lem:UFfaultCMA:7}%
		Suppose that $\IdScheme$ is subset-revealing. 
		Then there exists an algorithm $\simSignature_7$ such that
		for any adversary $\Ad{A}$ against the $\UFfaultCMA$ security of $\SignatureScheme$,
		issuing at most $q_{S,7}$ queries to $\oracleFaultSIGN$ on index 7,
		$q_{S}$ queries to $\oracleFaultSIGN$ in total,
		and at most $q_\ROChallenge$ queries to $\ROChallenge$,
		\begin{align}
			\gameDist{SimulateSeven}{A}	\label{eq:UFfaultCMA:7}
				\leq q_{S,7} \cdot \left( 
					\boundSpecialStatisticalHVZK
					+ \frac {3}{2}  \sqrt{ (q_H + q_{S} + 1) \cdot \maxEntropyCommit}
				\right)
		\enspace . 
		\end{align}
	\end{lemma}
	The details on algorithm $\simSignature_7$ and the proof for \cref{eq:UFfaultCMA:7} are given in \cref{sec:HedgedFiatShamir:ProofFault7}.
	\begin{lemma}\label{lem:UFfaultCMA:4}%
		Suppose that $\IdScheme$ is subset-revealing. 
		There exists an algorithm $\simSignature_4$	such that
		for any adversary $\Ad{A}$ against the $\UFfaultCMA$ security of $\SignatureScheme$,
		issuing at most $q_{S,4}$ queries to $\oracleFaultSIGN$ on index 4,
		$q_{S}$ queries to $\oracleFaultSIGN$ in total,
		and at most $q_\ROChallenge$ queries to $\ROChallenge$,
		\begin{align}
		\gameDist{SimulateFour}{A}	\label{eq:UFfaultCMA:4}
			\leq q_{S,6} \cdot \left( 
				\boundSpecialStatisticalHVZK
				+ \frac {3}{2}  \sqrt{ (q_H + q_{S} + 1) \cdot 2 \maxEntropyCommit}
			\right)
		\enspace . 
		\end{align}
	\end{lemma}
	The details on algorithm $\simSignature_4$ and the proof for \cref{eq:UFfaultCMA:4} are given in \cref{sec:HedgedFiatShamir:ProofFault4}.
	What we have shown by now is that
	\[ | \Pr [ \gameSimulateSix^{\Ad{A}} \Rightarrow 1] - \Pr [ \gameSimulateFour^{\Ad{A}} \Rightarrow 1]|
	\leq  q_{S,\lbrace 4, 7 \rbrace} \cdot \left(
		\boundSpecialStatisticalHVZK
		+ \frac {3}{\sqrt{2}} \sqrt{(q_H + q_{S} + 1) \cdot \maxEntropyCommit}
	\right) \enspace,\]
	where $q_{S,\lbrace 4, 7 \rbrace}$ denotes the maximal number of queries to $\oracleFaultSIGN$ on all indices $i \in \lbrace 4, 7 \rbrace$.
	We are now ready to give our second security statement.
	\begin{lemma}\label{lem:UFfaultCMA:AdversaryAllFaults}%
		For any adversary $\Ad{A}$ against the $\UFfaultCMADifferentSet{\lbrace 4, 5, 6, 7, 9 \rbrace}$ security of $\SignatureScheme$, 
		there exists an adversary $\Ad{B}$ such that
		\begin{align*}
		\Pr [\gameSimulateFour^{\Ad{A}} \Rightarrow 1]
		\leq \succfun^{\UFCMAZero}_{\FS[\IdScheme, \ROChallenge]}(\Ad{B}) \enspace ,
		\end{align*}
		and $\Ad{B}$ has the same running time as $\Ad{A}$.
		The proof is given in \cref{sec:HedgedFiatShamir:ProofAdversaryAllFaults}.
	\end{lemma}
	Collecting the probabilities, we obtain
	\begin{align*}
		\succfun^{\UFfaultCMADifferentSet{\lbrace 4, 5, 6, 7, 9 \rbrace}}_{\FS[\IdScheme, \ROChallenge]}(\Ad{A})
		\leq & \succfun^{\UFCMAZero}_{\FS[\IdScheme, \ROChallenge]}(\Ad{B}) \\
		& + q_{S} \cdot \left(
			\boundSpecialStatisticalHVZK
			+ \frac {3}{\sqrt{2}}  \sqrt{(q_H + q_{S} + 1) \cdot \maxEntropyCommit}
		\right) \enspace ,
	\end{align*}
	given that $\IdScheme$ is subset-revealing.
}

\heading{Generalising the proof for computational $\specialHVZK$.}
To generalise the proof, we observe that every game-hop consists of two steps: Adaptive reprogramming and, subsequently, replacing honest transcripts with simulated ones.
To obtain the result for computational $\specialHVZK$, we have to reorder the games:
We will first reprogram the random oracle for \textit{all} fault indices \textit{at once}, 
with oracle $\oracleFaultSIGN$ reprogramming the random oracle for each fault index as specified in the sequential proof (see Sections \ref{sec:HedgedFiatShamir:ProofFault9} to \ref{sec:HedgedFiatShamir:ProofFault4}).
Combined reprogramming yields an upper bound of $\frac{3q_S}{\sqrt{2}} \sqrt{(q_H + q_{S} + 1) \cdot \maxEntropyCommit}$.
After these changes, the random oracle is a-posteriori reprogrammed such that it is consistent with the transcripts,
and hence, the transition to simulated transcripts can be reduced to distinguishing the
special computational multi-$\HVZK$ games (see \cref{Def:HVZK:Computational}).
In more detail, the $\HVZK$ reduction can simply use its own transcript oracle $\getTransChallenge$, and simulate the adaptive reprogramming like our $\UFCMAZero$ reductions,
see, e.g., the reduction given in \cref{sec:HedgedFiatShamir:ProofAdversary569}.

\subsection{Game $\gameSimulateNine$: Simulating $\oracleFaultSIGN$ for index 9 (Proof of \cref{lem:UFfaultCMA:9})}%
\label{sec:HedgedFiatShamir:ProofFault9}

As a warm-up, we will first consider simulations with respect to fault index 9.
Recall that index 9 denotes the fault type which allows $\Ad{A}$
to fault the resulting (honestly generated) signature
(see line~\ref{line:Simulate9:Fault9} in \cref{fig:UFfaultCMA:games:Simulate9}).
To prove \cref{lem:UFfaultCMA:9},
let $\Ad{A}$ be an adversary against the $\UFfaultCMA$ security of $\SignatureScheme$,
issuing at most $q_{S,9}$ queries to $\oracleFaultSIGN$ on index 9,
$q_{S}$ queries to $\oracleFaultSIGN$ in total,
and at most $q_\ROChallenge$ queries to $\ROChallenge$.
We define the signature simulation algorithm $\simSignature_9$ as in \cref{fig:UFfaultCMA:games:Simulate9}.

\begin{figure}[h] \begin{center} \fbox{ \small
	
	\nicoresetlinenr
	
	\begin{minipage}[t]{4.3cm}
		
		\underline{$\oracleFaultSIGN(m, i = 9, \phi)$}
		
		\begin{nicodemus}
			
			\item $(\commitment, \state) \leftarrow \Commit(\sk)$
			
			\item $\challenge := \ROChallenge(\commitment, m, \pk)$
			
			\item $\response \leftarrow \Respond(\sk, \challenge, \state)$
			
			\item $\ListOfMessages := \ListOfMessages \cup \lbrace m \rbrace$
			
			\item \pcreturn $\sigma := \phi(\commitment, \response)$ \label{line:Simulate9:Fault9}
			
		\end{nicodemus}
		
	\end{minipage}
	
	\;
	
	\begin{minipage}[t]{4.3cm}
		
		\underline{$\simSignature_9(m, \phi)$}
		
		\begin{nicodemus}
			
			\item $\challenge \uni \ChallengeSpace$
			
			\item $(\commitment, \response) \leftarrow \Sim(\pk, \challenge)$
			
			\item $\ROChallenge := \ROChallenge^{(\commitment, m, \pk) \mapsto \challenge}$
			
			\item $\ListOfMessages := \ListOfMessages \cup \lbrace m \rbrace$
			
			\item \pcreturn $\sigma := \phi(\commitment, \response)$
			
		\end{nicodemus}
		
	\end{minipage}	
			
}
\end{center}
	\caption{
		Original oracle $\oracleFaultSIGN$ for the case that $i=9$,
		and signature simulation algorithm $\simSignature_9$ for the proof of \cref{lem:UFfaultCMA:9}.
	}
	\label{fig:UFfaultCMA:games:Simulate9}
\end{figure}

To proceed from game $G_0$ to $\gameSimulateNine$, we use an argument similar to the one given in \cref{thm:FS_NoMA_to_CMA_tight}:
During execution of $\oracleFaultSIGN(m, 9, \phi)$, we first derandomise the challenges and reprogram $\ROChallenge$ such that it is rendered a-posteriori-consistent with the resulting transcripts,
resulting in an invocation of \cref{theorem:reprGameBased},
where $R = q_{S,9}$, $q = q_H + q_{S} + 1$, and $p_{\max} = \maxEntropyCommit$.
As the second step, we then make use of the fact that we assume $\IdScheme$ to be statistically $\specialHVZK$,
and hence, honestly generated transcripts can be replaced with simulated ones during execution of $\oracleFaultSIGN(m, 9, \phi)$.

After these changes, $\oracleFaultSIGN(m, 9, \phi) = \simSignature_9(m, \phi)$ and
\begin{align*}
	\gameDist{SimulateNine}{A}
	\leq	q_{S,9} \cdot \boundSpecialStatisticalHVZK
			+ \frac {3q_{S,9}}{2}  \sqrt{ (q_H + q_{S} + 1) \cdot \maxEntropyCommit} \enspace .
\end{align*} 
\subsection{Game $\gameSimulateFive$: Simulating $\oracleFaultSIGN$ for index 5 (Proof of \cref{lem:UFfaultCMA:5})}%
\label{sec:HedgedFiatShamir:ProofFault5}

Recall that index 5 denotes the fault type which allows $\Ad{A}$
to fault the triplet $(\commitment, m, \pk)$, 
when taken as input to random oracle $\ROChallenge$ to compute the challenge $\challenge$ 
(see line~\ref{line:Simulate5:Fault5} in \cref{fig:UFfaultCMA:games:Simulate5}).
To prove \cref{lem:UFfaultCMA:5},
let $\Ad{A}$ be an adversary against the $\UFfaultCMA$ security of $\SignatureScheme$,
issuing at most $q_{S,5}$ queries to $\oracleFaultSIGN$ on index 5,
$q_{S}$ queries to $\oracleFaultSIGN$ in total,
and at most $q_\ROChallenge$ queries to $\ROChallenge$.
We define the signature simulation algorithm $\simSignature_5$ as in \cref{fig:UFfaultCMA:games:Simulate5}.

\begin{figure}[h] \begin{center} \fbox{ \small
	
	\nicoresetlinenr
	
	\begin{minipage}[t]{4.3cm}
		
		\underline{$\oracleFaultSIGN(m, i = 5, \phi)$}
		
		\begin{nicodemus}
			
			\item $(\commitment, \state) \leftarrow \Commit(\sk)$
			
			\item $(\hat{\commitment}, \hat{m}, \hat{\pk}) := \phi(\commitment, m, \pk)$ 	
				\label{line:Simulate5:Fault5}
			
			\item $\challenge := \ROChallenge(\hat{\commitment},  \hat{m}, \hat{\pk}))$ 
			
			\item $\response \leftarrow \Respond(\sk, \challenge, \state)$
			
			\item $\ListOfMessages := \ListOfMessages \cup \lbrace \hat{m} \rbrace$
			
			\item \pcreturn $\sigma := (\commitment, \response)$
			
		\end{nicodemus}
		
	\end{minipage}
	
	\; 
			
	\begin{minipage}[t]{4.3cm}
		
		\underline{$\simSignature_5(m, \phi)$}
		
		\begin{nicodemus}
			
			\item $\challenge \uni \ChallengeSpace$
			
			\item $(\commitment, \response) \leftarrow \Sim(\pk, \challenge)$
			
			\item $(\hat{\commitment}, \hat{m}, \hat{pk}) := \phi(\commitment, m, pk)$ 			
				
			\item $\ROChallenge := \ROChallenge^{(\hat{\commitment},  \hat{m}, \hat{\pk}) \mapsto \challenge}$
				
			\item $\ListOfMessages := \ListOfMessages \cup \lbrace \hat{m} \rbrace$
			
			\item \pcreturn $\sigma := (\commitment, \response)$
			
		\end{nicodemus}
		
	\end{minipage}

}
\end{center}
	\caption{
		Original oracle $\oracleFaultSIGN$ for the case that $i=5$,
		and signature simulation algorithm $\simSignature_5$ for the proof of \cref{lem:UFfaultCMA:5}.
	}
	\label{fig:UFfaultCMA:games:Simulate5}
\end{figure}

To proceed from game $G_{\before{SimulateFive}}$ to $\gameSimulateFive$,
we adapt the argument of \cref{sec:HedgedFiatShamir:ProofFault9}:
During execution of $\oracleFaultSIGN(m, 5, \phi)$, we first derandomise the challenges and reprogram $\ROChallenge$ such that it is rendered a-posteriori-consistent with with the resulting transcripts,
resulting in an invocation of \cref{theorem:reprGameBased},
where $R = q_{S,5}$ and $q = q_H + q_{S} + 1$.
To make $p_{\max}$ explicit, let $\phi_{\commitment}$ ($\phi_{m}$, $\phi_{\pk}$) denote the share of $\phi$
acting on $\commitment$ (m, $\pk$).
We can now identify reprogramming positions $x$ with $(\phi_{m}(m), \phi_{\commitment}(\commitment), \phi_{\pk}(\pk))$.
The distribution $p$ consists hence of the constant distribution that always returns $\phi_{m}(m)$ and $\phi_{\pk}(\pk)$,
as these parts of the reprogramming position are already fixed,
together with the distribution $\phi_{\commitment}(\Commit(\sk))$.
Note that $\phi_{\commitment}$ is either the identity, a bit flip, or a function that fixes one bit of $\commitment$, hence $p_{\max} \leq 2 \maxEntropyCommit$.

As the second step, we can again make use of the fact that we assume $\IdScheme$ to be statistically $\specialHVZK$,
and honestly generated transcripts can be replaced with simulated ones during execution of $\oracleFaultSIGN(m, 5, \phi)$.

After these changes, $\oracleFaultSIGN(m, 5, \phi) = \simSignature_5(m, \phi)$ and
\begin{align*}
	\gameDist{SimulateFive}{A}
	\leq	q_{S,5} \cdot \boundSpecialStatisticalHVZK
	+ \frac {3q_{S,5}}{2}  \sqrt{ (q_H + q_{S} + 1) \cdot 2 \maxEntropyCommit} \enspace .
\end{align*} 
\subsection{Game $\gameSimulateSix$: Simulating $\oracleFaultSIGN$ for index 6 (Proof of \cref{lem:UFfaultCMA:6})}%
\label{sec:HedgedFiatShamir:ProofFault6}

Recall that index 6 denotes the fault type which allows $\Ad{A}$ to fault
the output $\challenge = \ROChallenge(\commitment, m, \pk)$ of the challenge hash function $\ROChallenge$
(see line~\ref{line:Simulate6:Fault6} in \cref{fig:UFfaultCMA:games:Simulate6}).
To prove \cref{lem:UFfaultCMA:6},
let $\Ad{A}$ be an adversary against the $\UFfaultCMA$ security of $\SignatureScheme$,
issuing at most $q_{S,6}$ queries to $\oracleFaultSIGN$ on index 6,
$q_{S}$ queries to $\oracleFaultSIGN$ in total,
and at most $q_\ROChallenge$ queries to $\ROChallenge$.
We define the signature simulation algorithm $\simSignature_6$ as in \cref{fig:UFfaultCMA:games:Simulate6}.

\begin{figure}[h] \begin{center} \fbox{ \small
		
		\nicoresetlinenr
		
		\begin{minipage}[t]{4.6cm}
			
			\underline{$\oracleFaultSIGN(m, i = 6, \phi)$}
			
			\begin{nicodemus}
				
				\item $(\commitment, \state) \leftarrow \Commit(\sk)$
				
				\item $\challenge := \ROChallenge(\commitment, m, \pk)$
				
				\item $\response \leftarrow \Respond(\sk, \phi(\challenge), \state)$ \label{line:Simulate6:Fault6}
				
				\item $\ListOfMessages := \ListOfMessages \cup \lbrace m \rbrace$
				
				\item \pcreturn $\sigma := (\commitment, \response)$ 
				
			\end{nicodemus}
			
		\end{minipage}
	
		\quad 
		
		\begin{minipage}[t]{4.3cm}
			
			\underline{$\simSignature_6(m, \phi)$}
			
			\begin{nicodemus}
				
				\item $\challenge \uni \ChallengeSpace$
				
				\item $(\commitment, \response) \leftarrow \Sim(\pk, \phi(\challenge))$
				
				\item \pcif $\phi(\challenge) \notin \ChallengeSpace$
				
					\item \quad $\response := \bot$
				
				\item $\ROChallenge := \ROChallenge^{(\commitment, m, \pk) \mapsto \challenge}$
				
				\item $\ListOfMessages := \ListOfMessages \cup \lbrace m \rbrace$
				
				\item \pcreturn $\sigma := (\commitment, \response)$
				
			\end{nicodemus}
			
		\end{minipage}
			
}
	\end{center}
	\caption{
		Original oracle $\oracleFaultSIGN$ for the case that $i=6$,
		and signature simulation algorithm $\simSignature_6$ for the proof of \cref{lem:UFfaultCMA:6}.
	}
	\label{fig:UFfaultCMA:games:Simulate6}
\end{figure}

To proceed from game $G_{\before{SimulateSix}}$ to $\gameSimulateSix$,
we again adapt the argument from \cref{sec:HedgedFiatShamir:ProofFault9}:
During execution of $\oracleFaultSIGN(m, 6, \phi)$,
we first derandomise the challenges and reprogram $\ROChallenge$ such that it is rendered a-posteriori-consistent with with the resulting transcripts,
resulting in an invocation of \cref{theorem:reprGameBased},
where $R = q_{S,6}$ and $q = q_H + q_{S} + 1$.
Like in \cref{sec:HedgedFiatShamir:ProofFault9}, $p_{\max} = \maxEntropyCommit$.

As the second step, we can again make use of the fact that we assume $\IdScheme$ to be statistically $\specialHVZK$,
and hence, honestly generated transcripts can be replaced with simulated ones during execution of $\oracleFaultSIGN(m, 6, \phi)$.
Note that as the challenges are faulty, however, we have to simulate rejection
whenever faulting the challenge results in an invalid challenge, i.e.,
whenever $\phi(c) \notin \ChallengeSpace$.

Since \ifeprint the scheme \fi  $\IdScheme$ is validity aware (see \cref{Def:ValidityAwareness}),
it holds that after these changes, $\oracleFaultSIGN(m, 6, \phi) = \simSignature_6(m, \phi)$ and
\begin{align*}
	\gameDist{SimulateSix}{A}
	\leq	q_{S,6} \cdot \boundSpecialStatisticalHVZK
	+ \frac {3q_{S,6}}{2}  \sqrt{ (q_H + q_{S} + 1) \cdot \maxEntropyCommit} \enspace .
\end{align*} 
\subsection{$\UFCMAZero$ adversary for game $\gameSimulateSix$, for $\SetFaultFunctions = \lbrace 5,6,9 \rbrace$ (Proof of \cref{lem:UFfaultCMA:Adversary569})}%
\label{sec:HedgedFiatShamir:ProofAdversary569}

Recall that in game $\gameSimulateSix$, faulty signatures are simulated
for all indices $i \in \lbrace 5,6,9 \rbrace$.
Since adversaries against the $\UFfaultCMADifferentSet{\lbrace 5,6,9 \rbrace}$ security of $\SignatureScheme$
only have access to $\oracleFaultSIGN(m, i, \phi)$ for $i \in \lbrace 5, 6, 9 \rbrace$,
the game derives all oracle answers by a call to one of the simulated oracles $\simSignature_i(m, \phi)$, where $i \in \lbrace 5, 6, 9 \rbrace$.
To prove \cref{lem:UFfaultCMA:Adversary569}, we construct an $\UFCMAZero$ adversary $\Ad{B}$ in \cref{fig:UFfaultCMA:Adversary569}.

\begin{figure}[h] \begin{center} \fbox{ \small
			
	\begin{minipage}[t]{4.8cm}
		
		\nicoresetlinenr
		
		\underline{\textbf{Adversary} $\Ad{B}^\qRO{{\ROChallenge}}(\pk)$}
		\begin{nicodemus}
			
			\item $(m^*, \sigma^*) \leftarrow \Ad{A}^{\oracleFaultSIGN, \qRO{\ROChallenge'}} (\pk)$
			
			\item \pcif $m^* \in \ListOfMessages$ ABORT
			
			\item \pcreturn $(m^*, \sigma^*)$
			
		\end{nicodemus}		
		
		\ \\
		
		\underline{$\oracleFaultSIGN(m, i \in \lbrace 5,6,9 \rbrace, \phi)$}
		
		\begin{nicodemus}
			
			\item $\sigma \leftarrow \simSignature_i(m, \phi)$
			
			\item \pcreturn $\sigma$
			
		\end{nicodemus}
	
		\ \\
		
		\underline{$\ROChallenge'(\commitment, m, \pk)$}
		
		\begin{nicodemus}
			
			\item \pcif $\exists \challenge$ s. th. $(\commitment, m, \pk, \challenge)
			\in \List{\ROChallenge'}$
			
			\item \quad \pcreturn $\challenge$
			
			\item \pcelse \pcreturn $\ROChallenge(\commitment, m, \pk)$
			
		\end{nicodemus}
	
		\ \\
		
		\underline{$\simSignature_9(m, \phi)$}
		
		\begin{nicodemus}
			
			\item $\challenge \uni \ChallengeSpace$
			
			\item $(\commitment, \response) \leftarrow \Sim(\pk, \challenge)$
			
			\item \pcif $\exists \challenge'$ s. th. $(\commitment, m, \pk \challenge') \in \List{\ROChallenge'}$
			
			\item \quad $\List{\ROChallenge'} := \List{\ROChallenge'} \setminus \lbrace (\commitment, m, \pk  \challenge')\rbrace$
			
			\item  $\List{\ROChallenge'} := \List{\ROChallenge'} \cup \lbrace (\commitment, m, \pk, \challenge) \rbrace$
			
			\item $\ListOfMessages := \ListOfMessages \cup \lbrace m \rbrace$
			
			\item \pcreturn $\sigma := \phi(\commitment, \response)$
			
		\end{nicodemus}
		
	\end{minipage}
	
	\quad
	
	\begin{minipage}[t]{4.9cm}
		
		\underline{$\simSignature_5(m, \phi)$}
		
		\begin{nicodemus}
			
			\item $\challenge \uni \ChallengeSpace$
			
			\item $(\commitment, \response) \leftarrow \Sim(\pk, \challenge)$
			
			\item $(\hat{\commitment}, \hat{m}, \hat{pk}) := \phi(\commitment, m, pk)$ 	
			
			\item \pcif $\exists \challenge'$ s. th. $(\hat{\commitment},  \hat{m}, \hat{\pk}, \challenge') \in \List{\ROChallenge'}$
			
			\item \quad $\List{\ROChallenge'} := \List{\ROChallenge'} \setminus \lbrace (\hat{\commitment},  \hat{m}, \hat{\pk}, \challenge')\rbrace$

			\item  $\List{\ROChallenge'} := \List{\ROChallenge'} \cup \lbrace (\hat{\commitment},  \hat{m}, \hat{\pk}, \challenge) \rbrace$
			
			\item $\ListOfMessages := \ListOfMessages \cup \lbrace \hat{m} \rbrace$
			
			\item \pcreturn $\sigma := (\hat{\commitment}, \response)$
			
		\end{nicodemus}
	
		\ \\
		
		\underline{$\simSignature_6(m, \phi)$}
		
		\begin{nicodemus}
			
			\item $\challenge \uni \ChallengeSpace$
			
			\item $(\commitment, \response) \leftarrow \Sim(\pk, \phi(\challenge))$
			
			\item \pcif $\phi(\challenge) \notin \ChallengeSpace$
			
			\item \quad $\response := \bot$
			
			\item \pcif $\exists \challenge'$ s. th. $(\commitment, m, \pk \challenge') \in \List{\ROChallenge'}$
			
			\item \quad $\List{\ROChallenge'} := \List{\ROChallenge'} \setminus \lbrace (\commitment, m, \pk,  \challenge')\rbrace$
			
			\item  $\List{\ROChallenge'} := \List{\ROChallenge'} \cup \lbrace (\commitment, m, \pk, \challenge) \rbrace$
			
			\item $\ListOfMessages := \ListOfMessages \cup \lbrace m \rbrace$
			
			\item \pcreturn $\sigma := (\commitment, \response)$
			
		\end{nicodemus}
		
	\end{minipage}

}
\end{center}
	\caption{$\UFCMAZero$ Adversary $\Ad{B}$ for the proof of \cref{lem:UFfaultCMA:Adversary569}.
	}
	\label{fig:UFfaultCMA:Adversary569}
\end{figure}

Since in game $\gameSimulateSix$, all signatures are defined relative to simulated transcripts,
and the random oracle is reprogrammed accordingly, $\Ad{B}$ perfectly simulates $\gameSimulateSix$
and has the same running time as $\Ad{A}$.

Furthermore, $\Ad{A}$ can not win if $m^*$ was a query to $\oracleFaultSIGN$.
Therefore, it is ensured that no reprogramming did occur on $m^*$ and $\Ad{A}$'s signature is also valid in $\Ad{B}$'s $\UFCMAZero$ game.

\begin{align*}
\Pr [ \gameSimulateSix^{\Ad{A}} \Rightarrow 1]
\leq	\succfun^{\UFCMAZero}_{\FS[\IdScheme, \ROChallenge]}(\Ad{B}) \enspace .
\end{align*}

\subsection{Faulting the response input (Proof of \cref{lem:UFfaultCMA:7})}%
\label{sec:HedgedFiatShamir:ProofFault7}

Recall that index 7 denotes the fault type which allows $\Ad{A}$ to fault
the input $(\sk, \challenge, \state)$ to the response function $\Respond$
(see line~\ref{line:Simulate7:Fault7} in \cref{fig:UFfaultCMA:games:Simulate7}) and that we assume that $\IdScheme$ is subset-revealing.
To prove \cref{lem:UFfaultCMA:7},
let $\Ad{A}$ be an adversary against the $\UFfaultCMA$ security of $\SignatureScheme$,
issuing at most $q_{S,7}$ queries to $\oracleFaultSIGN$ on index 7,
$q_{S}$ queries to $\oracleFaultSIGN$ in total,
and at most $q_\ROChallenge$ queries to $\ROChallenge$.
We define the signature simulation algorithm $\simSignature_7$ as in \cref{fig:UFfaultCMA:games:Simulate7}.

\begin{figure}[h] \begin{center} \fbox{ \small
			
	\nicoresetlinenr
	
	\begin{minipage}[t]{4.2cm}
		
		\underline{$\oracleFaultSIGN(m, i =7, \phi)$}
		
		\begin{nicodemus}
			
			\item $(\commitment, \state) \leftarrow \Commit(\sk)$
			
			\item $\challenge := \ROChallenge(\commitment, m, \pk)$
			
			\item $\response \leftarrow \Respond(\phi(\sk, \challenge, \state))$ \label{line:Simulate7:Fault7}
			
			\item $\ListOfMessages := \ListOfMessages \cup \lbrace \hat{m} \rbrace$
			
			\item \pcreturn $\sigma := (\commitment, \response)$
			
		\end{nicodemus}
		
	\end{minipage}

	\quad 
	
	\begin{minipage}[t]{6cm}
		
		\underline{$\simSignature_7(m, \phi)$}
		
		\begin{nicodemus}
			
			\item $\challenge \uni \ChallengeSpace$
			
			\item Parse $(\phi_{\sk}, \phi_{\challenge}, \phi_{\state}) := \phi$
			
			\item \pcif $\phi_{\challenge} \neq Id$ \gcom{$\phi$ targets $\challenge$}
			
				\item \quad $(\commitment, \response) \leftarrow \Sim(\pk, \phi(\challenge))$
				
				\item \quad \pcif $\phi(\challenge) \notin \ChallengeSpace$
				
					\item \quad \quad $\response := \bot$
					
			\item \pcelse 
			
				\item \quad $(\commitment, \response) \leftarrow \Sim(\pk, \challenge)$
			
				\item \quad \pcif $\phi_{\state} \neq Id$ \gcom{$\phi$ targets $\state$}
			
					\item \quad \quad $I \leftarrow \mathsf{DeriveSet}(\challenge)$ 
					
					\item \quad \quad Parse $(\state_i)_{i \in I} := \response$
					
					\item \quad \quad $z := (\phi_{\state, i}(\state_i))_{i \in I}$
					
			\item $\ROChallenge := \ROChallenge^{(\commitment, m, \pk) \mapsto \challenge}$
				
			\item $\ListOfMessages := \ListOfMessages \cup \lbrace m \rbrace$
			
			\item \pcreturn $\sigma := (\commitment, \response)$
			
		\end{nicodemus}
		
	\end{minipage}
}
	\end{center}
	\caption{
		Original oracle $\oracleFaultSIGN$ for the case that $i=7$,
		and signature simulation algorithm $\simSignature_7$ for the proof of \cref{lem:UFfaultCMA:7}.
	}
	\label{fig:UFfaultCMA:games:Simulate7}
\end{figure}

If fault function $\phi$ is targeted at $\challenge$, the situation is essentially the same as for fault index 6, and thus, the simulation strategy is identical to that of $\simSignature_6$ (see \cref{sec:HedgedFiatShamir:ProofFault6}).
If fault function $\phi$ is targeted at $\sk$,
$\phi$ has no effect whatsoever since we assume $\IdScheme$ to be subset-revealing,
meaning that the responses returned by $\Respond$ do not depend on $\sk$ (see \cref{Def:SubsetRevealing}).
The simulation strategy is hence identical to that of $\simSignature_9$.
The simulation algorithm covers both cases by 
dissecting $\phi$ into the shares $\phi_{\sk}$ ($\phi_{\challenge}, \phi_{\state}$) 
acting on $\sk$ ($\challenge$, $\state$) and treating the cases where $\phi_{\challenge} \neq Id$ ($\phi_{\sk} \neq Id$) similar to $\simSignature_6$ ($\simSignature_9$).

It remains to discuss the case  where $\phi$ is targeted at $\state$.
Since we assume $\IdScheme$ to be subset-revealing (see \cref{Def:SubsetRevealing}),
we observe that
$\Respond(\phi(\sk, \challenge, \state))
= \Respond(\sk, \challenge, \phi_{\state}(\state))
=  ((\phi_{\state}(\state))_i)_{i \in I}$,
where $I = \mathsf{DeriveSet}(\challenge)$.
Hence, computing $z \leftarrow \Respond(\phi(\sk, \challenge, \state))$ is equivalent
to deriving $I \leftarrow \mathsf{DeriveSet}(\challenge)$,
only considering the shares $\phi_{\state, i}$ of $\phi_{\state}$ that act on $\state_i$,
and returning $(\phi_{\state, i}(\state_i))_{i \in I}$.
With this alternative description of the original $\Respond$ algorithm,
it can easily be verified that even for the case where $\phi$ is targeted at the state,
honest transcripts can be replaced with simulated transcripts by letting $\phi$ act on the response $z$ as described above.

After these changes, $\oracleFaultSIGN(m, 7, \phi) = \simSignature_7(m, \phi)$ and
\begin{align*}
	\gameDist{SimulateSeven}{A}
	\leq	q_{S,7} \cdot \boundSpecialStatisticalHVZK
	+ \frac {3q_{S,7}}{2}  \sqrt{ (q_H + q_{S} + 1) \cdot \maxEntropyCommit} \enspace .
\end{align*}

\subsection{Faulting the commitment output (Proof of \cref{lem:UFfaultCMA:4})}%
\label{sec:HedgedFiatShamir:ProofFault4}

Recall that index 4 denotes the fault type which allows $\Ad{A}$ to fault
the output of $\Commit(\sk)$
(see line~\ref{line:Simulate4:Fault4} in \cref{fig:UFfaultCMA:games:Simulate4}).
To prove \cref{lem:UFfaultCMA:4},
let $\Ad{A}$ be an adversary against the $\UFfaultCMA$ security of $\SignatureScheme$,
issuing at most $q_{S,4}$ queries to $\oracleFaultSIGN$ on index 4,
$q_{S}$ queries to $\oracleFaultSIGN$ in total,
and at most $q_\ROChallenge$ queries to $\ROChallenge$.
We define the signature simulation algorithm $\simSignature_4$ as in \cref{fig:UFfaultCMA:games:Simulate4}.

\begin{figure}[h] \begin{center} \fbox{ \small
		
	\nicoresetlinenr
		
	\begin{minipage}[t]{4.2cm}
		
		\underline{$\oracleFaultSIGN(m, i =4, \phi)$}
		
		\begin{nicodemus}
			
			\item $(\commitment, \state) \leftarrow \Commit(\sk)$
			
			\item $(\commitment, \state) := \phi(\commitment, \state)$ \label{line:Simulate4:Fault4}
			
			\item $\challenge := \ROChallenge(\commitment, m, \pk)$
			
			\item $\response \leftarrow \Respond(\sk, c, \state)$ 
			
			\item $\ListOfMessages := \ListOfMessages \cup \lbrace \hat{m} \rbrace$
			
			\item \pcreturn $\sigma := (\commitment, \response)$
			
		\end{nicodemus}
			
	\end{minipage}
		
	\quad
		
	\begin{minipage}[t]{6cm}
		
		\underline{$\simSignature_4(m, \phi)$}
		
		\begin{nicodemus}
			
			\item $\challenge \uni \ChallengeSpace$
			
			\item $(\commitment, \response) \leftarrow \Sim(\pk, \challenge)$
			
			\item Parse $(\phi_{\commitment}, \phi_{\state}) := \phi$
			
			\item \pcif $\phi_{\commitment} \neq Id$ \gcom{$\phi$ targets $\commitment$}
			
				\item \quad $(\hat{\commitment}, \hat{m}, \hat{pk}) := \phi(\commitment, m, pk)$ 			
				
				\item \quad $\ROChallenge := \ROChallenge^{(\hat{\commitment},  \hat{m}, \hat{\pk}) \mapsto \challenge}$
			
				\item \quad $\ListOfMessages := \ListOfMessages \cup \lbrace \hat{m} \rbrace$
				
			\item \pcelse
			
				\item \quad \pcif $\phi_{\state} \neq Id$ \gcom{$\phi$ targets $\state$}
			
					\item \quad \quad $I \leftarrow \mathsf{DeriveSet}(\challenge)$ 
					
					\item \quad \quad Parse $(\state_i)_{i \in I} := \response$
					
					\item \quad \quad  $z := (\phi_{\state, i}(\state_i))_{i \in I}$
			
				\item \quad $\ROChallenge := \ROChallenge^{(\commitment, m, \pk) \mapsto \challenge}$
				
				\item \quad $\ListOfMessages := \ListOfMessages \cup \lbrace m \rbrace$

			\item \pcreturn $\sigma := (\commitment, \response)$
			
		\end{nicodemus}
		
	\end{minipage}
}
	\end{center}
	\caption{
		Original oracle $\oracleFaultSIGN$ for the case that $i=4$,
		and signature simulation algorithm $\simSignature_4$ for the proof of \cref{lem:UFfaultCMA:4}.
	}
	\label{fig:UFfaultCMA:games:Simulate4}
\end{figure}

If fault function $\phi$ is targeted at $\commitment$,
the situation is essentially the same as for fault index 5, and thus, the simulation strategy is identical to that of $\simSignature_5$ (see \cref{sec:HedgedFiatShamir:ProofFault5}).
If fault function $\phi$ is targeted at $\state$,
the situation is essentially the same as for fault index 7, and thus, the simulation strategy is identical to that of $\simSignature_7$ (see \cref{sec:HedgedFiatShamir:ProofFault7}).
Putting both cases together, we obtain
\begin{align*}
	\gameDist{SimulateFour}{A}
	\leq	q_{S,6} \cdot \boundSpecialStatisticalHVZK
	+ \frac {3q_{S,6}}{2}  \sqrt{ (q_H + q_{S} + 1) \cdot 2 \maxEntropyCommit}  \enspace .
\end{align*} 
\subsection{$\UFCMAZero$ adversary for game $\gameSimulateFour$, for $\SetFaultFunctions = \lbrace 4, 5, 6, 7, 9 \rbrace$ (Proof of \cref{lem:UFfaultCMA:AdversaryAllFaults})}%
\label{sec:HedgedFiatShamir:ProofAdversaryAllFaults}

Recall that in game $\gameSimulateFour$, faulty signatures are simulated
for all indices $i \in \lbrace 4, 5, 6, 7, 9 \rbrace$.
For adversaries against the $\UFfaultCMADifferentSet{\lbrace 4, 5, 6, 7, 9 \rbrace}$ security of $\SignatureScheme$,
the game derives all oracle answers by a call to one of the simulated oracles $\simSignature_i(m, \phi)$.
To prove \cref{lem:UFfaultCMA:AdversaryAllFaults}, observe that we can now extend
adversary $\Ad{B}$ defined in \cref{fig:UFfaultCMA:Adversary569} 
such that it is capable to perfectly simulate game $\gameSimulateFour$
by running the simulations, and simulating the random oracle to $\Ad{A}$, accordingly.
(I.e., $\Ad{B}$ runs $\Ad{A}$ with oracle access to $\ROChallenge'$ that is first set to $\ROChallenge$, and that gets reprogrammed, with $\Ad{B}$ keeping track of the classical queries to $\oracleFaultSIGN$.)

Again, $\Ad{A}$ can not win if $m^*$ was a query to $\oracleFaultSIGN$,
hence a valid signature is also valid in $\Ad{B}$'s $\UFCMAZero$ game and

\begin{align*}
	\Pr [ \gameSimulateFour^{\Ad{A}} \Rightarrow 1]
	\leq	\succfun^{\UFCMAZero}_{\FS[\IdScheme, \ROChallenge]}(\Ad{B}) \enspace .
\end{align*}

{%
}

\section{From $\UFfaultCMADifferentSet{}$ to $\UFnonceFaultCMADifferentSet{}$ (Proof of \cref{theorem:UFNonceFaultCMA})}%
\label{sec:HedgedFiatShamir:ProofNonceFCMA}

\newcounter{CounterFCMAtoNonceFCMA} %
\setcounter{CounterFCMAtoNonceFCMA}{1}

{
	\newcounter{RandomiseFCMAtoNonceFCMA}
	\setcounter{RandomiseFCMAtoNonceFCMA}{\theCounterFCMAtoNonceFCMA}
	\stepcounter{CounterNoMAtoCMA}
}

{
	\newcounter{ExtractFCMAtoNonceFCMA}
	\setcounter{ExtractFCMAtoNonceFCMA}{\theCounterFCMAtoNonceFCMA}
}

\newcounter{LastGameFCMAtoNonceFCMA}
\setcounter{LastGameFCMAtoNonceFCMA}{\theCounterFCMAtoNonceFCMA}
 
We now present a proof for \cref{theorem:UFNonceFaultCMA} which we repeat for convenience.
\thmfive*

In our proof %
we will use a quantum extraction argument from \cite{C:AmbHamUnr19}, which we now recall.

{\heading{One-way to Hiding as a query extraction argument.}
	In \cite[Theorem 3]{C:AmbHamUnr19}, Ambainis et al. generalised the query-extraction argument from \cite{EC:Unruh14}.
	In their generalisation, they considered a distinguisher $\Ad{D}$ that has quantum access
	to an oracle $\RO{O} \in \lbrace \RO{O_1}, \RO{O_2} \rbrace$
	such that oracles $\RO{O_1}$ and $\RO{O_2}$ coincide on all inputs except on some subset $S$.
	It was shown that the difference in behaviour of $\Ad{D}^\qRO{O_1}$ and $\Ad{D}^\qRO{O_2}$
	can be upper bounded in terms of the extractability of input elements $x \in S$.
	The following theorem is a simplified restatement of \cite[Theorem 3]{C:AmbHamUnr19}.

	\begin{theorem}\label{thm:O2HExtract}
		
		Let $X$ and $Y$ be sets, and let $S \subset X$ be random.
		Let $\RO{O_1}, \RO{O_2} \in Y^X$ be random functions
		such that $\RO{O_1}(x) = \RO{O_1}(x)$ for all $x \notin S$,
		and let $\inputVar$ be a random bitstring.
		$(S, \RO{O_1}, \RO{O_2}, \inputVar)$ may have an arbitrary joint distribution.
		Then, for all quantum algorithms $\Ad{D}$ issuing at most $q$ queries to $\RO{O}$,
		
		\begin{equation}
			| \Pr [1 \leftarrow \Ad{D}^{\qRO{O_1}}(\inputVar) ]
				- \Pr [1 \leftarrow \Ad{D}^{\qRO{O_2}}(\inputVar) ]|
			\leq 2 q \cdot \sqrt{p_{\FIND}} \nonumber \enspace ,
		\end{equation}
				
		where 
		\[ p_{\FIND} := Pr[x \in S : x \leftarrow \Ad{Ext}^{\qRO{O_1}} (\inputVar)] \enspace .\]
		
		The same result holds with $\Ad{Ext}^{\qRO{O_2}}$
		(instead of  $\Ad{Ext}^{\qRO{O_1}}$) in the definition of $p_{\FIND}$.

		\begin{figure}[h] \begin{center} \fbox{ \small
					
			\begin{minipage}[t]{5.5cm}
				
				\nicoresetlinenr
				\underline{\textbf{Extractor} $\Ad{Ext}^{\qRO{O}} (\inputVar)$}
				\begin{nicodemus}
					
					\item $j \uni \lbrace 1, \cdots, q_{\ROforHedging} \rbrace$
					
					\item Run $\Ad{D}^{\qRO{O}} (\inputVar)$
					until $j$th query to $\RO{O}$
					
					\item $x \leftarrow \Measure$ query input register
					
					\item \pcreturn $x$
					
				\end{nicodemus}	
				
			\end{minipage}
					
		}
		\end{center} \end{figure}
	
	\end{theorem}
}

{} 
We now proceed to the proof of \cref{theorem:UFNonceFaultCMA}.
Let $\Ad{A}$ be an adversary against the $\UFnonceFaultCMADifferentSet{\SetFaultFunctions'}$ security
of $\SigSchemeHedged = \TrafoHedging [\SignatureScheme, \ROforHedging]$
for $\SetFaultFunctions' := \SetFaultFunctions \cup \lbrace 1 \rbrace$,
issuing at most $q_{S}$ queries to $\oracleNonceFaultSIGN$,
at most $q_\ROChallenge$ queries to $\ROChallenge$,
and at most $q_\ROforHedging$ queries to $\ROforHedging$.
In the random oracle model, the proof would work as follows: Either $\ROforHedging$ is never queried on any faulted version of $\sk$, or it is.
In the case that such query does not exist, the $\UFnonceFaultCMADifferentSet{\SetFaultFunctions'}$ experiment is completely simulatable by a reduction against the $\UFfaultCMA$ security of the underlying scheme $\SignatureScheme$,
as the signing randomness looks uniformly random to the adversary.
(Note that we made the assumption that $\Ad{A}$ issues no query $(m, \nonce)$ to $\oracleNonceFaultSIGN$ more than once.)
In the case that such a query $\phi(\sk)$ exists, it can be used to break $\UFfaultCMA$ security by going over all possible secret key candidates,
i.e., by going over all bit-flip functions, and checking whether any of those candidates can be used to generate a valid signature.

In principle, our QROM proof does the same. Consider the sequence of games given in \cref{fig:Games:UFNonceFaultCMA}:
We decouple the signing randomness from the secret key in game $G_1$.
Again, game $G_1$ can be simulated by a reduction $\Ad{B_1}$ against the $\UFfaultCMA$ security of the underlying scheme $\SignatureScheme$.
To upper bound the distance between games $G_0$ and $G_1$, we will use \cref{thm:O2HExtract}.
(In order to give a more detailed description of how \cref{thm:O2HExtract} can be used, we ``zoom in'' and give two intermediate helper games $G_{\nicefrac{1}{3}}$ and $G_{\nicefrac{2}{3}}$.)
Applying \cref{thm:O2HExtract}, we can upper bound the distance between games $G_0$ and $G_1$ in terms of the probability that measuring a random query to $\ROforHedging$ yields $\phi(\sk)$.
We then give a reduction $\Ad{B_2}$ that wins whenever the latter happens, with the same strategy as in the ROM sketch.

\begin{figure}[h] \begin{center} \fbox{ \small
			
	\begin{minipage}[t]{6cm}
		
		\nicoresetlinenr
		\underline{\textbf{Games} $G_0$ - $G_{1}$}
		\begin{nicodemus}
			
			\item $(\pk, \sk) \leftarrow \IG(\param)$
			
			\item $(m^*, \sigma^*) \leftarrow \Ad{A}^{\oracleNonceFaultSIGN, \qRO{\ROChallenge}, \qRO{\ROforHedging}} (\pk)$
			
			\item \pcif $m^* \in \ListOfMessages$ \pcreturn 0
			
			\item Parse $(\commitment^*, \response^*) := \sigma^*$
			
			\item $\challenge^* := \ROChallenge(\commitment^*, m^*)$
			
			\item \pcreturn $\VerifyId(\pk, \commitment^*, \challenge^*, \response^*)$
			
		\end{nicodemus}
		
{%
}		
	\end{minipage}
			
	\;
	
	\begin{minipage}[t]{5.5cm}
		
		\underline{$\oracleNonceFaultSIGN(m, \nonce, i \in\SetFaultFunctions', \phi)$}
		
		\begin{nicodemus}
			
			\item \pcif $i=1$
				
				\item \quad $f_1 := \phi$ 
				
				\item \quad $r := \ROforHedging(f_1(\sk), m, \nonce)$ \gcom{$G_0$}
				
				\item \quad $r \uni \RSpaceSign$ \label{line:UFNonceFaultCMA:Rerandomise1}\gcom{$G_1$}
				
				\item \quad $\sigma \leftarrow \getSignature(m, r, 2, \id)$
			
			\item \pcelse
			
				\item \quad $r := \ROforHedging(\sk, m, \nonce)$ \gcom{$G_0$}
				
				\item \quad $r \uni \RSpaceSign$ \label{line:UFNonceFaultCMA:Rerandomise2}\gcom{$G_1$}
				
				\item \quad $\sigma \leftarrow \getSignature(m, r, i, \phi)$
			
			\item \pcreturn $\sigma$
			
		\end{nicodemus}
	
		\ \\
		
		\underline{$\getSignature(m, r, i, \phi)$}
		
		\begin{nicodemus}
			
			\item $f_i := \phi$ and $f_j := \id \forall j \neq i$
			
			\item $(\commitment, \state) \leftarrow \Commit(\sk; r)$
			
			\item $(\commitment, \state) := f_4(\commitment, \state)$
			
			\item $(\hat{\commitment}, \hat{m}, \hat{pk}) := f_5(\commitment, m, pk)$
			
			\item $\challenge := f_6(\ROChallenge(\hat{\commitment},  \hat{m}, \hat{pk}))$
			
			\item $\response \leftarrow \Respond(f_7(\sk, c, \state))$
			
			 \item $\ListOfMessages := \ListOfMessages \cup \lbrace \hat{m} \rbrace$
			
			\item \pcreturn $\sigma := f_9(\commitment, \response)$
			
		\end{nicodemus}
		
	\end{minipage}

}
	\end{center}
	\caption{
		Games $G_0$ - $G_{1}$ for the proof of \cref{theorem:UFNonceFaultCMA}.
		Helper method $\getSignature$ is internal and cannot be accessed directly by $\Ad{A}$.
	}
	\label{fig:Games:UFNonceFaultCMA}
\end{figure}

\heading{Game $G_0$.}
{
	The (purely conceptual) difference between game $G_0$ and the original $\UFnonceFaultCMADifferentSet{}$ game is that after computing the signing randomness according to $\SigSchemeHedged$, 
	we outsource the rest of the signature computation to helper method $\getSignature$.
	In the case that $i=1$, $\getSignature$ is executed with index 2 and $\id$, as the rest of the signature generation is unfaulted.
	
	\[ \succfun^{\UFnonceFaultCMADifferentSet{\SetFaultFunctions'}}_{\SigSchemeHedged}(\Ad{A})
		= \Pr [ G_0^{\Ad{A}} \Rightarrow 1]\enspace.\]
}

\heading{Game $G_{1}$.}
{
	In game $G_{1}$,
	we re-randomise the $\Commit$ algorithm by letting $r \uni \RSpaceSign$
	instead of $r := \ROforHedging(f_1(\sk), m, \nonce)$,
	see lines~\ref{line:UFNonceFaultCMA:Rerandomise1} and \ref{line:UFNonceFaultCMA:Rerandomise2}.
	To upper bound $\Pr [ G_1^{\Ad{A}} \Rightarrow 1]$, consider $\UFfaultCMA$ Adversary $\Ad{B_1}$
	given in \cref{fig:UFNonceFaultCMA:AdversaryFaultCMA:1}.
	Adversary $\Ad{B_1}$ has access to the faulty signing oracle $\oracleFaultSIGN$ that is provided by game
	$\UFfaultCMA$, and that covers all faults except the ones that would have occurred with respect to index 1, i.e., the ones that fault the secret key as input to $\ROforHedging$.
	Due to our change described above, however, randomness $r$ is drawn independently of $\sk$ in game $G_1$,
	hence the $\Commit$ algorithm is randomised.
	The output of $\oracleFaultSIGN$ therefore allows $\Ad{B_1}$ to perfectly simulate game $G_1$ to $\Ad{A}$.
	Furthermore, any valid forgery game $G_1$ is also a valid forgery in $\Ad{B_1}$'s $\UFfaultCMA$ game.
	Hence,
	\[ \Pr [ G_1^{\Ad{A}} \Rightarrow 1]
		\leq \succfun^{\UFfaultCMA}_{\SignatureScheme}(\Ad{B_1}) 
	\enspace.\]
	
	\begin{figure}[h] \begin{center} \fbox{ \small
				
		\begin{minipage}[t]{6cm}
			
			\nicoresetlinenr
			\underline{\textbf{Adversary} $\Ad{B_1}^\qRO{{\ROChallenge}}(\pk)$}
			\begin{nicodemus}
				
				\item $(m^*, \sigma^*) \leftarrow \Ad{A}^{\oracleNonceFaultSIGN, \qRO{\ROChallenge}, \qRO{\ROforHedging}} (\pk)$
				
				\item \pcreturn $(m^*, \sigma^*)$
				
			\end{nicodemus}		
						
		\end{minipage}
		
		\;
				
		\begin{minipage}[t]{4.9cm}
			
			\underline{$\oracleNonceFaultSIGN(m, \nonce, i \in\SetFaultFunctions', \phi)$}
			\begin{nicodemus}
				
				\item \pcif $i = 1$
				
					\item \quad $\sigma \leftarrow \oracleFaultSIGN(m, 2, \id)$
				
				\item \pcelse $\sigma \leftarrow \oracleFaultSIGN(m, i, \phi)$

				\item \pcreturn $\sigma$
				
			\end{nicodemus}
			
		\end{minipage}
				
	}
		\end{center}
		\caption{$\UFfaultCMA$ Adversary $\Ad{B_1}$,
			with access to its own faulty signing oracle $\oracleFaultSIGN$,
			for the proof of \cref{theorem:UFNonceFaultCMA}.
		}
		\label{fig:UFNonceFaultCMA:AdversaryFaultCMA:1}
		
	\end{figure}
	
	It remains to upper bound $\gameDist{LastGameFCMAtoNonceFCMA}{A}$.
	To this end, we will make use of the query extraction variant of one-way to hiding (see \cref{thm:O2HExtract}).
	In order to keep our proof as accessible as possible, we introduce intermediate helper games $G_{\nicefrac{1}{3}}$ and $G_{\nicefrac{2}{3}}$
	in \cref{fig:Games:UFNonceFaultCMA:IntermediateForO2H}.

	\begin{figure}[t] \begin{center} \fbox{ \small
				
		\begin{minipage}[t]{5.4cm}
			
			\nicoresetlinenr
			\underline{\textbf{Games} $G_{\nicefrac{1}{3}}$ - $G_{1}$}
			\begin{nicodemus}
				
				\item $(\pk, \sk) \leftarrow \IG(\param)$
				
				\item $\RO{O} := \ROforHedging'$\gcom{$G_{\nicefrac{1}{3}}$} 	
					\label{line:UFNonceFaultCMA:ReplaceOracle} 
				
				\item $\RO{O} := \ROforHedging$ \gcom{$G_{\nicefrac{2}{3}}$-$G_{1}$}
				
				\item $(m^*, \sigma^*) \leftarrow \Ad{A}^{\oracleNonceFaultSIGN, \qRO{\ROChallenge}, \qRO{\RO{O}}} (\pk)$ 
					\label{line:UFNonceFaultCMA:ReplaceOracleAccess} 
				
				\item \pcif $m^* \in \ListOfMessages$ \pcreturn 0
				
				\item Parse $(\commitment^*, \response^*) := \sigma^*$
				
				\item $\challenge^* := \ROChallenge(\commitment^*, m^*)$
				
				\item \pcreturn $\VerifyId(\pk, \commitment^*, \challenge^*, \response^*)$
				
			\end{nicodemus}		
			
			\ \\
			
			\underline{\textbf{Extractor} $\Ad{E}^{\qRO{O}, \qRO{H}} (\pk, \sk, \List{\ROforHedging'})$}
			\begin{nicodemus}
				
				\item $j \uni \lbrace 1, \cdots, q_{\ROforHedging} \rbrace$
				
				\item Run $\Ad{A}^{\oracleNonceFaultSIGN, \qRO{\ROChallenge}, \qRO{O}} (\pk)$\\
					\text{\quad }
					until $j$th query to $\RO{O}$
				
				\item $(\sk', m, \nonce) \leftarrow \Measure$ query\\
					\text{\quad }input reg.
				
				\item \pcreturn $\sk'$
				
			\end{nicodemus}	
		
		\end{minipage}

		\;
		
		\begin{minipage}[t]{6cm}
			
			\underline{$\oracleNonceFaultSIGN(m, \nonce, i \in\SetFaultFunctions', \phi)$}
			
			\begin{nicodemus}
				
				\item \pcif $i=1$
				
					\item \quad $f_1 := \phi$ 
					
					\item \quad $r := \ROforHedging'(f_1(\sk), m, \nonce)$ \label{line:UFNonceFaultCMA:ReplaceOracleUsage1}
					\gcom{$G_{\nicefrac{1}{3}}$, $G_{\nicefrac{2}{3}}$, $\Ad{E}$}
					
					\item \quad $r \uni \RSpaceSign$ \gcom{$G_1$}
					
					\item \quad $\sigma \leftarrow \getSignature(m, r, 2, \id)$
				
				\item \pcelse
				
					\item \quad $r := \ROforHedging'(\sk, m, \nonce)$ \label{line:UFNonceFaultCMA:ReplaceOracleUsage2}
					\gcom{$G_{\nicefrac{1}{3}}$, $G_{\nicefrac{2}{3}}$, $\Ad{E}$}
					
					\item \quad $r \uni \RSpaceSign$ \gcom{$G_1$}
					
					\item \quad $\sigma \leftarrow \getSignature(m, r, i, \phi)$
				
				\item \pcreturn $\sigma$
				
			\end{nicodemus}
	
	\end{minipage}
	
	}
	\end{center}
		\caption{Intermediate helper games $G_{\nicefrac{1}{3}}$ and $G_{\nicefrac{2}{3}}$,
			justifying the game-hop from game $G_0$ to $G_1$, and query extractor $\Ad{E}$.
			Alternative oracle $\ROforHedging'$
			(see lines~\ref{line:UFNonceFaultCMA:ReplaceOracle}, \ref{line:UFNonceFaultCMA:ReplaceOracleUsage1} and~\ref{line:UFNonceFaultCMA:ReplaceOracleUsage2})
			is constructed by letting $\ROforHedging'(\sk', m, n) := \ROforHedging(\sk', m, n)$ for all input $(\sk',m,n)$ such that $\sk'$ cannot result from faulting $\sk$, and completing $\ROforHedging'$ randomly.
			Helper method $\getSignature$ %
			remains as in \cref{fig:Games:UFNonceFaultCMA}.
		}
		\label{fig:Games:UFNonceFaultCMA:IntermediateForO2H}
	\end{figure}
	
	As a preparation, we first consider intermediate game $G_{\nicefrac{1}{3}}$,
	in which we completely replace random oracle $\ROforHedging$
	with another random oracle $\ROforHedging'$ (see lines~\ref{line:UFNonceFaultCMA:ReplaceOracle}, \ref{line:UFNonceFaultCMA:ReplaceOracleUsage1} and~\ref{line:UFNonceFaultCMA:ReplaceOracleUsage2}),
	where $\ROforHedging'$ is defined as follows:
	Let $\List{\sk}$ denote the set of secret keys that could occur by faulting the secret key with a one-bit fault injection.
	We let $\ROforHedging'$ concur with $\ROforHedging$ for all inputs such that the input secret key is not in $\List{\sk}$, i.e., for all $\sk' \notin \List{\sk}$ and all $(m, \nonce) \in \MSpace \times \NonceSpace$,
	we let $\ROforHedging'(\sk', m, \nonce) := \ROforHedging(\sk', m, \nonce)$.
	We can then complete it to a random oracle on $\SKSpace \times \MSpace \times \NonceSpace$
	by picking another random oracle $\ROforHedging'' : \List{\sk}\times \MSpace \times \NonceSpace$,
	and letting $\ROforHedging'(\sk', m, \nonce) := \ROforHedging''(\sk', m, \nonce)$ for all $\sk' \in \List{\sk}$ and all $(m, \nonce) \in \MSpace \times \NonceSpace$.
	Since $\ROforHedging'$ still is a random oracle, and since we also use $\ROforHedging'$ to derive the signing randomness, this change is purely conceptual and
	
	\[\Pr [G_0^{\Ad{A}} \Rightarrow 1]  = \Pr [G_{\nicefrac{1}{3}}^{\Ad{A}} \Rightarrow 1] \enspace.\]
	
	In game $G_{\nicefrac{2}{3}}$, we prepare to rid the randomness generation of the secret key:
	We switch back to providing $\Ad{A}$ with oracle access to the original random oracle $\ROforHedging$,
	but we keep using $\ROforHedging'$ to derive the signing randomness.
	After this change, oracle $\ROforHedging'$ is not directly accessible by $\Ad{A}$ anymore,
	but only indirectly via the signing queries.
	Since we assume that $\Ad{A}$ issues no query $(m, \nonce)$ to $\oracleNonceFaultSIGN$ more than once,
	we can also replace these values with freshly sampled randomness as in game $G_1$, i.e.,
	
	\[\Pr [G_{\nicefrac{2}{3}}^{\Ad{A}} \Rightarrow 1]  = \Pr [G_1^{\Ad{A}} \Rightarrow 1] \enspace.\]

	So far, we have shown that
	\[ \succfun^{\UFnonceFaultCMADifferentSet{\SetFaultFunctions'}}_{\SigSchemeHedged}(\Ad{A})
	 	\leq \succfun^{\UFfaultCMA}_{\SignatureScheme}(\Ad{B_1})
	 	+ |\Pr [G_{\nicefrac{1}{3}}^{\Ad{A}} \Rightarrow 1]-\Pr[G_{\nicefrac{2}{3}}^{\Ad{A}} \Rightarrow 1]|
	 	\enspace.\]
	
	In order to upper bound
	$|\Pr [G_{\nicefrac{1}{3}}^{\Ad{A}} \Rightarrow 1]-\Pr[G_{\nicefrac{2}{3}}^{\Ad{A}} \Rightarrow 1]|$,
	we invoke \cref{thm:O2HExtract}:
	Distinguishing between the two games can be reduced to extracting one of the faulted secret keys from the queries to $\RO{G}$.
	To make this claim more formal, consider the query extractor $\Ad{E}$ from  \cref{thm:O2HExtract}, whose explicit description we give in \cref{fig:Games:UFNonceFaultCMA:IntermediateForO2H}.
	Extractor $\Ad{E}$ is run with access to oracle $\RO{O} \in \lbrace \ROforHedging, \ROforHedging' \rbrace$,
	which it will forward to $\Ad{A}$.
	It runs $\Ad{A}$ until $\Ad{A}$'s $i$th oracle query to $\RO{O}$,
	measures the query input register,
	and thereby obtains a triplet $(\sk', m, \nonce)$ of classical input values.
	Since we are only interested in points where $\ROforHedging$ and $\ROforHedging'$ differ,
	it is sufficient to let $\Ad{E}$ output the secret key candidate $\sk'$.
	Note that $\Ad{E}$ is able to simulate the signing oracle regardless of which oracle $\RO{O}$ it has access to:
	Recall that \cref{thm:O2HExtract} makes no assumption on the runtime of the query extractor, nor on the size of its input.
	Hence, the alternative oracle $\ROforHedging'$ can simply be encoded as part of the extractor's input,
	which we denote by adding $\List{\ROforHedging'}$ to $\Ad{E}$'s input.
	Since $\Ad{E}$ perfectly simulates game $G_{\nicefrac{1}{3}}$ if $\RO{O} = \ROforHedging'$,
	and game $G_{\nicefrac{1}{3}}$ if $\RO{O} = \ROforHedging$,
	\cref{thm:O2HExtract} yields
	
	\[|\Pr [G_{\nicefrac{1}{3}}^{\Ad{A}} \Rightarrow 1]-\Pr[G_{\nicefrac{2}{3}}^{\Ad{A}} \Rightarrow 1]|
		\leq 2q_{\ROforHedging}
				\cdot \sqrt{\Pr[\sk' \in \List{\sk} :
						\sk' \leftarrow \Ad{E}^{\qRO{\ROforHedging}, \qRO{H}} (\pk, \sk, \ROforHedging')]}
		\enspace .
	\]
	
	It remains to bound the success probability of the extractor $\Ad{E}$. At this point, the signing randomness is independent of $\ROforHedging$.
	We can hence also replace $\Ad{E}$ with an extractor $\Ad{E'}$ that uses freshly sampled randomness to sign,
	without any change in the extraction probability.
	(Again, we require that $\Ad{A}$ issues no query $(m, \nonce)$ to $\oracleNonceFaultSIGN$ more than once.)

	To bound the success probability of $\Ad{E'}$, consider $\UFfaultCMA$ Adversary $\Ad{B_2}$, which is given in \cref{fig:UFNonceFaultCMA:AdversaryFaultCMA:2}.
	Like $\Ad{B_1}$, adversary $\Ad{B_2}$ has access to the faulty signing oracle $\oracleFaultSIGN$ provided by game $\UFfaultCMA$, and it uses $\oracleFaultSIGN$ to answer signing queries.
	$\Ad{B_2}$ perfectly simulates the view of $\Ad{A}$ when $\Ad{A}$ is run by extractor $\Ad{E'}$,
	and the probability that $\Ad{E'}$ returns some $\sk' \in \List{\sk}$ is hence exactly the probability that 
	$\Ad{B_2}$ obtains some $\sk' \in \List{\sk}$ by measuring in line~\ref{line:UFNonceFaultCMA:B2Measures}.
	After running $\Ad{A}$ until the $j$th query to $\ROforHedging$,
	and extracting a secret key candidate $\sk'$ from this query,
	$\Ad{B_2}$ computes the list $\List{\sk'}$ of candidate secret keys that could occur
	by faulting $\sk'$ with a one-bit fault injection (including the identity function).
	Since bit flips are involutory, and set-bit functions can be reversed by set-bit functions, 
	$\sk' \in \List{\sk}$ iff $\sk \in \List{\sk'}$.
	Hence, if $\Ad{B_2}$ obtains some $\sk' \in \List{\sk}$ by measuring, then $\Ad{B_2}$ will encounter $\sk$ during execution of its loop and therefore generate a valid signature.
	\[ \Pr[\sk' \in \List{\sk} :
	\sk' \leftarrow \Ad{E'}^{\qRO{\ROforHedging}, \qRO{H}} (\pk, \sk, \ROforHedging')] \leq \succfun^{\UFfaultCMA}_{\SignatureScheme}(\Ad{B_2}) \enspace.\]
	
	\begin{figure}[t] \begin{center} \fbox{ \small
				
		\begin{minipage}[t]{6cm}
			
			\nicoresetlinenr
			
			\underline{\textbf{Adversary} $\Ad{B_2}^\qRO{{\ROChallenge}}(\pk)$}
			\begin{nicodemus}
				
				\item $j \uni \lbrace 1, \cdots, q_{\ROforHedging} \rbrace$
				
				\item Run $\Ad{A}^{\oracleNonceFaultSIGN, \qRO{\ROChallenge}, \qRO{\ROforHedging}} (\pk)$\\
				\text{\quad }
				until $j$th query to $\ROforHedging$
				
				\item $\sk' \leftarrow \Measure$ query input register \label{line:UFNonceFaultCMA:B2Measures}
				
				\item $m^* \uni \MSpace \setminus \ListOfMessages'$
				
				\item \pcfor $sk'' \in \List{\sk'}$\label{line:UFNonceFaultCMA:CandidatesSK}
				
				\item \quad $\sigma \leftarrow \Sign(\sk'', m)$
				
				\item \quad \pcif $\VerifySig(m, \sigma) = 1$
				
				\item \quad \quad \pcreturn $(m, \sigma)$
				
				\item $\pcreturn \bot$
				
			\end{nicodemus}	
			
		\end{minipage}
		
		\;
				
		\begin{minipage}[t]{4.9cm}
			
			\underline{$\oracleNonceFaultSIGN(m, \nonce, i \in\SetFaultFunctions', \phi)$}
			\begin{nicodemus}
				
				\item \pcif $i = 1$
				
					\item \quad $\sigma \leftarrow \oracleFaultSIGN(m, 2, \id)$
				
				\item \pcelse $\sigma \leftarrow \oracleFaultSIGN(m, i, \phi)$
				
				\item \pcif $i = 5$ \pcand $\phi$ affects $m$
				
					\item \quad $\ListOfMessages' := \ListOfMessages' \cup \lbrace \phi_m(m) \rbrace$
				
				\item \pcelse $\ListOfMessages' := \ListOfMessages' \cup \lbrace m \rbrace$
				
				\item \pcreturn $\sigma$
				
			\end{nicodemus}
			
		\end{minipage}
				
	}
		\end{center}
		\caption{$\UFfaultCMA$ Adversary $\Ad{B_2}$,
			with access to its own faulty signing oracle $\oracleFaultSIGN$,
			for the proof of \cref{theorem:UFNonceFaultCMA}.
			List $\List{\sk'}$ (see line~\ref{line:UFNonceFaultCMA:CandidatesSK}) denotes the list of secret keys that could occur by faulting $\sk'$ with a one-bit fault injection.
		}
		\label{fig:UFNonceFaultCMA:AdversaryFaultCMA:2}
		
	\end{figure}

}

\section{A matching attack - details}
\label{ap:attack}
We now provide the proof of \Cref{lem:basic-attack} which refers to \cref{fig:attack}. We first repeat both for convenience.
\lemmatwo*

\begin{reusefigure}{fig:attack}
\begin{center} \makebox[\textwidth][c]{\fbox{
			
	\nicoresetlinenr	
	
	\begin{minipage}[t]{7cm}
		
		\underline{Before potential reprogramming:}
		\begin{nicodemus}
		    \item Prepare registers $XY$ in $\frac{1}{\sqrt{2^n}}\sum_{x \in  [2^n]}\ket{x, 0}_{XY}$
		    \item Query $\RO{O}$ using registers $XY$
		    \item \pcfor $i = 0,...,q-1$:
			    \item  \quad Apply $\sigma$ on register $X$
			    \item  \quad Query $\RO{O}$ using registers $XY$
		\end{nicodemus}
		
	\end{minipage}
	\begin{minipage}[t]{6cm}
		
		\underline{After potential reprogramming:}
		\begin{nicodemus}
		    \item Query ${\RO{O}'}$ using using registers $XY$
		    \item \pcfor $i = q-2,...,0$:
			    \item \quad Apply $\sigma^{-1}$ on register $X$
			    \item \quad Query ${\RO{O}'}$ using registers $XY$
		    \item Measure $X$ according to $\{\Pi_0, \Pi_1\}$
		    \item Output $b$ if the state projects onto $\Pi_b$.
		\end{nicodemus}
		
	\end{minipage}
}}\end{center}
	\ifnum\tightOnSpace=1 \vspace{-0.4cm} \fi
	\caption{Distinguisher for a single reprogrammed point.}
	\label{fig:attack:appendix}
\end{reusefigure}
Recall that our distinguisher makes $q$ queries to $\RO{O}$, the oracle before the potential reprogramming, and $q$ queries to $\RO{O}'$, the oracle after the potential reprogramming. In our attack, we fix an arbitrary cyclic permutation $\sigma$ on $[2^n]$, and for the fixed reprogrammed point $x^*$, we define \ifeprint the set \fi $S = \{x^*, \sigma^{-1}(x^*),...,\sigma^{-q+1}(x^*)\}$, $\overline{S} = \{0,1\}^n \setminus S$,
$\Pi_0 = \frac{1}{2}\left( \ket{S} + \ket{\overline{S}}\right)\left( \bra{S} + \bra{\overline{S}}\right)$ and $\Pi_1 = I - \Pi_0$.

\begin{proof}[Proof of \Cref{lem:basic-attack}]

{%
}

First, notice that before measuring the final state, the value of the register $XY$ is 
\begin{align}\label{eq:state-distinguiser}
	\ket{\Psi}_{XY} := 
		\frac{1}{\sqrt{2^n}}
		\sum_{x}\ket{x}\ket{\bigoplus_{j = 0}^{q} \left(\RO{O}\left(\sigma^{j}(x)\right) \oplus \RO{O}'\left(\sigma^{j}(x)\right) \right)}
	\enspace . 
\end{align}
Let us consider the case where $\RO{O}(x^*)= \RO{O}'(x^*)$, i.e., when the random oracle was not reprogrammed, or reprogrammed to the same value.
We then have that 
\begin{align*}
	\ket{\Psi}_{XY}
    = \frac{1}{\sqrt{2^n}}\sum_{x}\ket{x}\ket{0} \enspace ,
\end{align*}
and in this case, the probability that the distinguisher outputs $0$ is
\begin{align}
  &
  \left|\left(\frac{1}{\sqrt{2}}\left( \bra{S} + \bra{\overline{S}}\right)\right)\left(
  \frac{1}{\sqrt{2^n}}\sum_{x}\ket{x} \right)\right|^2 \nonumber 
  \\
  &= \left(\frac{q}{\sqrt{2^{n+1}q}} + \frac{2^{n}-q}{\sqrt{2^{2n+1} - 2^{n+1}q}}\right)^2 \nonumber \\
  & =
  \frac{q}{2^{n+1}} + \frac{2^{2n} - 2^{n+1}q + q^2}{2^{2n+1} - 2^{n+1}q}
  + \frac{q(2^{n}-q)}{\sqrt{2^{n}q}\sqrt{2^{2n} - 2^{n}q}} \nonumber \\
  &\geq
  \frac{2^{2n}}{2^{2n+1} - 2^{n+1}q} +  \frac{\sqrt{q}}{\sqrt{2^n - q}}\left(1 - \frac{q}{2^n} - \frac{\sqrt{q}}{\sqrt{2^n - q}}\right) \nonumber \\
  &\geq \frac{1}{2}  + 
  \frac{\sqrt{q}}{2\sqrt{2^n}} 
  \enspace , \label{eq:prob1-attack}
\end{align}
where we removed some positive terms in the first inequality, and in the second inequality we used the fact that $q < 2^{n-3}$.

On the other hand, if $\RO{O}(x^*) \ne \RO{O}'(x^*)$, then the final state is
\[
	\ket{\Psi}_{XY} = \sum_{x \in S} \ket{x}\ket{\RO{O}(x^*) \oplus \RO{O}'(x^*)} + \sum_{x \in \overline{S}} \ket{x}\ket{0} \enspace ,
\]
and when we trace out the output qubit we have the mixed state
\[
  \frac{q}{2^n} \kb{S} + \frac{2^n -q}{2^n} \kb{\overline{S}} \enspace .
\]

Since
\[
  \left(\frac{1}{\sqrt{2}}\left( \bra{S} + \bra{\overline{S}}\right)\right)
  \ket{S} = 
  \left(\frac{1}{\sqrt{2}}\left( \bra{S} + \bra{\overline{S}}\right)\right)
  \ket{\overline{S}} =
  \frac{1}{\sqrt{2}} \enspace,
\]
we have that the distinguisher then outputs $0$ with probability 
\begin{align}
\frac{1}{2} .\label{eq:prob2-attack}
\end{align}

The advantage of the distiguisher is therefore
\begin{align}
	&\Pr[\Ad{D}(\cdot) =  0 | \text{no reprogramming}] - \Pr[\Ad{D}(\cdot) =  0  | \text{reprogramming}] \nonumber
	\\
	&= \Pr[\Ad{D}(\cdot) =  0 |  \RO{O}(x^*)=\RO{O}'(x^*)] - \frac{1}{2^m} \Pr[\Ad{D}(\cdot) =  0  | \RO{O}(x^*)=\RO{O}'(x^*)] \nonumber \\
	& \quad - (1 - \frac{1}{2^m}) \Pr[\Ad{D}(\cdot) =  0 | \RO{O}(x^*)\ne \RO{O}'(x^*)] \nonumber \\
	&\geq  (1 - \frac{1}{2^m}) \cdot\frac{\sqrt{q}}{2\sqrt{2^n}}  \geq \frac{\sqrt{q}}{4\sqrt{2^n}} \enspace , \nonumber
\end{align}
where the terms after the first equality correspond to the probability that the distinguisher outputs $0$ when the function was not reprogrammed,
when the function was reprogrammed on $x^*$ for $\RO{O}'(x^*) \ne \RO{O}(x)$,
and when the function was reprogrammed on $x^*$, but $\RO{O}(x^*) = \RO{O}'(x^*)$, respectively.
The following inequality holds due to \Cref{eq:prob1-attack,eq:prob2-attack},
and the last inequality holds since $m \geq 1$.

It remains to show that the attack can be performed efficiently.
The only step of the distinguisher whose efficiency is not straightforward is the measurement on the basis $(\Pi_0,\Pi_1)$.
(As an example for an efficiently implementable cyclic permutation, consider the operation of adding $1\mod 2^n$).

In order to prove the efficiency of the attack, we first show how to construct the state $\ket{S}$ and $\eps$-approximate the state $\ket{\overline{S}}$  in time $O(q\log{\frac{1}\eps})$.

To create the state $\ket{S}$, we first create the superposition $\frac{1}{\sqrt{q}} \sum_{i = 1}^q \ket{i}\ket{\sigma^i(x^*)}$ and then erase the first register.
Notice that this can be done with $O(q)$ queries to $\sigma$ and the result is \ifeprint$$\else$\fi\frac{1}{\sqrt{q}} \sum_{i = 1}^q \ket{0}\ket{\sigma^i(x^*)} = \ket{0}\ket{S}.\ifeprint$$\else$\fi

This procedure not only gives us a circuit to construct $\ket{S}$, but also a circuit that perfectly uncomputes it (by just running the circuit backwards with appropriate number of auxiliary qubits).

In order to $\eps$-approximate $\ket{\overline{S}}$, we can use the following lemma from \cite{EC:AlaMajRus20}:
\begin{lemma}[Restatement of Lemma $3$ of \cite{EC:AlaMajRus20}]
	Let $S \subseteq \{0,1\}^n$ and $U_S$ a circuit that uncomputes $\ket{S}$ and $\eps > 0$.  There exists a quantum algorithm $\mathcal{P}^{U_S}$ that runs in time $\poly(|S|,\log\frac{1}{\eps})$ that satisfies
	\[
	\norm{\mathcal{P}^{U_S}\ket{0}\ket{0} - \ket{\overline{S}}\ket{0}}^2 \leq  \eps.
	\]
\end{lemma}
  
Finally, notice that in this case there is an efficient circuit $C$ that efficiently computes the state $\frac{1}{\sqrt{2}}\left(\ket{S} + \ket{\overline{S}}\right)$, by first creating the state $\ket{+}\ket{0}$, conditioned on the first qubit being $0$, create the state $\ket{\overline{S}}$, conditioned on the first qubit being $1$, create the state $\ket{S}$, and then, with $q$ extra queries to $\sigma$, flips the first qubit from $\ket{1}$ to $\ket{0}$ conditioned on the contents of the second register being in $S$.
Using circuit $C$, we can $\eps$-approximate the projection of some state $\rho$ onto $\{\Pi_0,\Pi_1\}$ in time $\poly(q,\log\frac{1}{\eps})$: append the necessary auxiliary qubits to $\rho$, perform $C^{\dagger}$ and then measure all qubits in the computational basis. Then, the output is $0$ iff  all qubits are $0$.
\end{proof}

\subsection{Generalization to multiple reprogrammings}
In this section, we discuss the generalization of the proposed attack to the case with multiple reprogrammed points. We provide an intuition why this extension works and we leave  its formal analysis for future work.

To warm up, let us consider the simpler case where both the number of potentially reprogrammed points and the number of queries is polynomial in $n$. Notice that in this case, since all reprogrammed values are chosen uniformly at random, the probability that there exists some $i,j \in [q]$ and potentially reprogrammed points  $x, x' \in \{0,1\}^n$, such that $f_k^i(x) = f_k^j(x')$ is exponentially small. Conditioning on the fact that this event does not happen, the previous analysis follows directly.

For the most general case (but considering the number of queries being $o(2^n)$, otherwise one could just use the classical attack), the same analysis holds by bounding the number of collisions (i.e., the number of $i,j \in [q]$  and $x, x' \in \{0,1\}^n$, such that $f_k^i(x) = f_k^j(x')$) and that these collisions cancel out (i.e., $\RO{O}(f_k^i(x)) \oplus \RO{O}(f_k^j(x')) \oplus \RO{O}'(f_k^i(x)) = \RO{O}'(f_k^j(x')) = 0$), which can be proven by using tail bounds. Then a more careful analysis following the outline of \Cref{lem:basic-attack} also works.

\end{document}